\documentclass[12pt]{article}%
\usepackage{afterpage}
\usepackage{mathptmx}
\usepackage{bm}
\usepackage{amsmath}
\usepackage{amsfonts}
\usepackage{tabularx}
\usepackage{amssymb}
\usepackage{booktabs}
\usepackage{float}
\usepackage{graphicx}
\newcolumntype{C}{>{\centering\arraybackslash}X}
\usepackage[subrefformat=parens]{subcaption}
\usepackage{color}
\usepackage{comment}
\usepackage{rotating}
\usepackage[letterpaper,margin=1in]{geometry}
\usepackage[breaklinks=true,bookmarksopen=true,colorlinks=true,citecolor=blue,urlcolor=blue]%
{hyperref}
\usepackage[authoryear]{natbib}
\usepackage{booktabs}
\usepackage{enumerate}
\usepackage{pdfsync}
\usepackage{pdflscape}
\usepackage{multirow}
\usepackage{setspace}
\usepackage{multicol}
\newcommand*\samethanks[1][\value{footnote}]{\footnotemark[#1]}
\providecommand{\U}[1]{\protect\rule{.1in}{.1in}}
\newtheorem{theorem}{Theorem}[section]

\newtheorem{assumption}{Assumption}

\newtheorem{example}{Example}[section]

\newtheorem{lemma}[theorem]{Lemma}

\newtheorem{proposition}[theorem]{Proposition}

\newenvironment{proof}[1][Proof]{\noindent\textbf{#1.} }{\ \rule{0.5em}{0.5em}}

\makeatletter

\long
\def\comment#1{}
\@addtoreset{lemma}{Appendix B} \makeatother

\newcommand\Perp{\protect\mathpalette{\protect\independenT}{\perp}}
\def\independenT#1#2{\mathrel{\rlap{$#1#2$}\mkern2mu{#1#2}}}

\newcommand{\E}{\mathbb{E}}

\begin{document}

\title{Identification and Estimation in Fuzzy Regression Discontinuity Designs with Covariates\thanks{Some of our results appeared originally as the ``binary case'' in our working paper \cite{caetano2017over}. We thank Luis Alvarez, Isaiah Andrews, Sebastian Calonico, Peter Hull, Toru Kitagawa, Rodrigo Pinto, David Slichter, Jann Spiess, Tymon S{\l}oczy{\'n}ski, three anonymous referees, and participants on several seminars and conferences for helpful comments. All errors are our own.}}
\author{Carolina Caetano\thanks{University of Georgia.}, Gregorio Caetano\samethanks, Juan Carlos Escanciano\thanks{Universidad Carlos III de Madrid. Research funded by Ministerio de Ciencia e Innovación grant PID2021-127794NB-I00 and Comunidad de Madrid, grants EPUC3M11 (VPRICIT) and H2019/HUM-589.}}
\date{January 2026}
\maketitle
\vspace{-0.15in}

\begin{abstract}
We study fuzzy regression discontinuity designs with covariates and characterize the weighted averages of conditional local average treatment effects (WLATEs) that are point identified. Any identified WLATE equals a Wald ratio of conditional reduced-form and first-stage discontinuities. We highlight the Compliance-Weighted LATE (CWLATE), which weights cells by squared first-stage discontinuities and maximizes first-stage strength. For discrete covariates, we provide simple estimators and robust bias-corrected inference. In simulations calibrated to common designs, CWLATE improves stability and reduces mean squared error relative to standard fuzzy RDD estimators when compliance varies. An application to Uruguayan cash transfers during pregnancy yields precise RDD-based effects on low birthweight.
\end{abstract}
\vspace{1.2\baselineskip} 
\section{Introduction}

Regression Discontinuity Design (RDD) is a leading quasi-experimental design \citep{Imbens_Lemieux,Lee_Lemieux,Cattaneo_Escanciano,Cattaneo_Idrobo_Titiunik}. In fuzzy RDDs, small first-stage discontinuities often render local Wald estimators imprecise and sensitive to tuning choices \citep[e.g.,][]{Feir_Lemieux_Marmer,Noack_Rothe}. At the same time, conditional first-stage discontinuities typically vary with observed covariates, implying heterogeneous compliance. This motivates two questions: (i) which causal parameters are point identified in fuzzy RDDs with covariates, and (ii) among those, which are best supported by the design in the sense of leveraging the strongest first-stage information.

We answer both questions within a unified framework. We characterize a broad class of \emph{weighted local average treatment effects} (WLATEs): weighted averages of covariate-specific LATEs at the cutoff with nonnegative weights summing to one. Our first result is an identification theorem: under conditions described below, a WLATE is point identified if and only if every covariate value with positive weight (i) has a nonzero conditional first-stage discontinuity at the cutoff, and (ii) appears with positive probability on both sides arbitrarily close to the cutoff. Any identified WLATE can be written as a Wald ratio in a transformed model where the ``treatment'' is the covariate-specific first-stage discontinuity and the ``outcome'' is the corresponding reduced-form discontinuity; equivalently, each WLATE is associated with an ``instrument'' given by a function of covariates, with different choices inducing different complier weights. This representation is a device for organizing what can be identified by cutoff discontinuities in fuzzy RDDs.

Within the identified WLATE class, we single out the \emph{Compliance-Weighted LATE (CWLATE)}. Choosing an ``instrument'' proportional to the conditional first-stage discontinuity implies weights proportional to the \emph{squared} first-stage. Subject to the sign and support restrictions mandated by identification, this ``instrument'' has maximum correlation with the first-stage among admissible choices, making the CWLATE the member of the class that is most aligned with available first-stage information. In this sense, the CWLATE is the ``easiest to estimate'' in our class.\footnote{We borrow the ``easy-to-estimate'' terminology from \cite{goldsmith2024contamination}, though here ``easy'' refers to strongest identification rather than to efficiency.}

We provide an economic interpretation of the identified WLATEs as the effect of a targeted policy where incentives are assigned by drawing a covariate value from a chosen distribution, and then offering the incentive to a random individual with that covariate value at the cutoff. Specifically, the ``instrument'' of a given WLATE is the ratio between the likelihood of that covariate value under the policy and the likelihood of that covariate value in the population. The ``instrument'' thus defines the relative tilt of the policy in favor of or against certain covariate values in comparison to the incidence of those values in the population. In that framework, the standard RDD estimand can be interpreted as the effect of a policy where incentives are offered randomly at the cutoff. In contrast, the CWLATE can be interpreted as the effect of a policy that targets compliers, by offering proportionally more incentives to groups that are more likely to have compliers. 

Compliance-weighting has antecedents in the IV literature  (e.g., \cite{joffe2003weighting,borusyak2020non,huntington2020instruments,coussens2021improving, CaetanoEscanciano21, hazard2022improving, abadie2023instrumental}), where the predominant concern is efficiency in regular conditional moment models. Fuzzy RDDs, by contrast, target \emph{local} and intrinsically nonregular parameters defined by discontinuities at a cutoff, so the usual efficiency bounds are not well-defined in the sense of \citet{chamberlain1992efficiency}. Accordingly, we take an identification-first approach: we characterize the full class of point identified WLATEs, and show that the CWLATE is the member of this class least prone to identification issues.

Our paper also complements a growing literature on fuzzy RDDs with covariates that uses covariates predominantly to improve precision, while retaining the fuzzy-RDD LATE as the target estimand. Covariate-adjusted local regressions are advocated in \citet{Imbens_Lemieux} and analyzed in depth by \citet{Calonico_Cattaneo_Farrell_Titiunik}, who characterize conditions under which such adjustment identifies the unconditional LATE and yields efficiency gains. \citet{Frolich} and \citet{Frolich_Huber} provide nonparametric representations and estimators of the unconditional LATE based on covariate-specific effects, and \citet{noack2021flexible} proposes flexible adjustment methods that can accommodate large covariate vectors. We share the concern that local Wald ratios can be highly sensitive to small first-stage discontinuities, but we shift emphasis to \emph{estimand selection}: when conditional first-stage discontinuities vary with covariates, fuzzy RDDs point identify a class of cutoff-local effects, and the CWLATE concentrates on where the conditional first-stage is strongest. Our weak-monotonicity condition parallels related discussions in compliance-weighted IV interpretations in \citet{kolesa2013estimation,sloczynski2020should,huntington2020instruments}.

We treat covariates as locally pre-determined features used solely to condition local discontinuities. Identification relies on four requirements near the cutoff: (a) local independence of unobservables and treatment effects from the running variable conditional on covariates; (b) weak monotonicity, i.e. for each covariate value either there are compliers or defiers near the cutoff, but not both; (c) overlap (each positively weighted covariate value occurs both just above and just below the cutoff); and (d) a nonzero conditional first-stage.  Our conditions are equivalent to the standard assumptions for RDD, but conditional on covariates. Importantly, we do not require smoothness or balance in the marginal covariate distribution near the cutoff, nor do we require strong monotonicity (which completely rules out defiers). As such, the conditions used to establish the WLATE class and the identification of the CWLATE are generally weaker than the conditions in the previous literature on fuzzy RDDs with covariates.

For estimation, we focus on discrete (or coarsened) covariates. We propose simple plug-in estimators for the identified WLATE class, based on local-polynomial estimates of covariate-specific reduced-form and first-stage discontinuities, and develop robust bias-corrected (RBC) inference for the CWLATE, building on \citet{Calonico_Cattaneo_Titiunik}; the Online Appendix provides the asymptotic theory and parallel results for more general WLATEs.

We investigate the finite sample performance of the CWLATE estimator in a Monte Carlo study, and show that,  when compliance varies across covariates, the CWLATE estimator is markedly more stable and typically attains lower mean square error (MSE) than the standard RDD estimators (with or without covariate adjustment). 

We also revisit \citet{amarante2016cash} on cash transfers during pregnancy in Uruguay (eligibility via a predicted-income cutoff): conventional fuzzy RDD estimates are imprecise in preferred specifications, whereas CWLATE-based estimates point to sizable reductions in low birthweight among compliers, illustrating how aligning the estimand with where the design is strongest can ``rescue'' an otherwise weak-appearing fuzzy RDD in settings with heterogeneous compliance.

The rest of the paper proceeds as follows. Section \ref{identification} formalizes the model, states the WLATE identification result, introduces the CWLATE, and discusses policy interpretations. Section \ref{RDDestimation} describes estimation and RBC inference. Section \ref{sec: MC} reports Monte Carlo evidence. Section \ref{application} presents the empirical application. Section \ref{conclusion} concludes. An Online Appendix contains additional results and proofs of the main results.


\section{Identification in RDD with Covariates}
\label{identification}

\newcommand{\Econd}[2]{\mathbb{E}\!\left[\,#1 \,\middle|\, #2 \,\right]}

\subsection{Model Setup}

We consider a fuzzy RDD with covariates. Let $X_i\in\{0,1\}$ denote treatment and $(Y_i(1),Y_i(0))$ the potential outcomes. The observed outcome is

\vspace{-.5cm}

\[
Y_i \;=\; Y_i(0) \;+\; \beta_i X_i,
\]
where $\beta_i = Y_i(1)-Y_i(0)$ may be heterogeneous and correlated with treatment. In addition to $(Y_i,X_i)$, we observe a continuously distributed running variable $Z_i$ with cutoff $z_0$, and a vector of covariates $W_i$ with support $\mathcal W$.

For each $z$, let $X_i(z)$ denote the potential treatment if $Z_i=z$. Define

\vspace{-.5cm}

\[
\Delta_i(e) \;=\; X_i(z_0+e)-X_i(z_0-e), 
\qquad
\Delta_i \;=\; \lim_{e\downarrow0}\Delta_i(e),
\]
whenever the limit exists. Units with $\Delta_i=1$ are compliers, $\Delta_i=-1$ defiers, and $\Delta_i=0$ always- or never-takers. The symbol $\Perp$ denotes independence.

\begin{assumption}[RDD with Covariates]\label{ass:rdd}
(i) $\mathbb{E}[Y_i(0)\mid Z_i=z,W_i]$ is continuous in $z$ at $z_0$ a.s.; 
(ii) there exists $\bar\epsilon>0$ such that, for all $z\in(z_0-\bar\epsilon,z_0+\bar\epsilon)$, $(\beta_i,X_i(z))\Perp Z_i\mid W_i$; 
(iii) $\mathbb{E}[|\beta_i|]<\infty$; 
(iv) (weak monotonicity) for all $w\in\mathcal W$, either $\mathbb{P}(\Delta_i\ge0\mid W_i=w)=1$ or $\mathbb{P}(\Delta_i\le0\mid W_i=w)=1$.
\end{assumption}

Assumption \ref{ass:rdd} is weaker than the assumptions in the RDD with covariates literature. Importantly, the distribution of $W_i|Z_i=z$ may be discontinuous at the cutoff. In fact, $W_i$ may cause $Y_i(0),$ $Y_i(1),$ $X_i$ and $Z_i$, but may not be caused by them near the cutoff. Additionally, we allow the existence of defiers, provided that there aren't also compliers with the same covariate value. 

For any random variable $V_i$, define the conditional limits
\begin{equation}
\label{eq:cond-limits}
\begin{aligned}
\Econd{V_i}{Z_i=z_0^+,\, W_i=w}
&= \lim_{z\downarrow z_0}\, \Econd{V_i}{Z_i=z,\, W_i=w},\\
\Econd{V_i}{Z_i=z_0^-,\, W_i=w}
&= \lim_{z\uparrow z_0}\, \Econd{V_i}{Z_i=z,\, W_i=w},
\end{aligned}
\end{equation}
and the associated conditional discontinuity

\vspace{-.5cm}

\begin{equation}
\delta_V(w)
= \Econd{V_i}{Z_i=z_0^+,\, W_i=w}
 - \Econd{V_i}{Z_i=z_0^-,\, W_i=w}.
\end{equation}

Covariate values observed on both sides arbitrarily close to the cutoff belong to 

\vspace{-1.2cm}

\[\mathcal W_{z_0}:
=\{w\in\mathcal W:\; \forall\,\epsilon>0, \, \mathbb{P}\!\big(W_i=w \,\big|\, Z_i\in(z_0,\, z_0+\epsilon)\big),\mathbb{P}\!\big(W_i=w \,\big|\, Z_i\in(z_0-\epsilon,\, z_0)\big)>0\}.\] 

\subsection{Conditional LATEs and WLATEs}

For $w\in\mathcal W$, define the conditional local average treatment effect (conditional LATE)

\vspace{-.5cm}

\[
\beta(w)
= \E\!\big[\beta_i\mid W_i=w,\,Z_i=z_0,\,\Delta_i\neq0\big].
\]
Let $\omega(W_i)\ge0$ a.s.\ with $\E[\omega(W_i)]=1$, and define the WLATE

\vspace{-.5cm}

\[
\beta_\omega \;=\; \E\!\big[\omega(W_i)\,\beta(W_i)\big].
\]

Henceforth, we assume all moments involved in WLATE definitions are finite.

\begin{theorem}[Identification of Conditional LATEs and WLATEs]\label{thm:WLATE}
(a) Under Assumption \ref{ass:rdd}: for any $w\in\mathcal W$, $\delta_Y(w)=\beta(w)\,\delta_X(w)$. (b) The conditional LATE $\beta(w)$ is point identified iff (i) $\delta_X(w)\neq0$ and (ii) $w\in\mathcal W_{z_0}$, in which case $\beta(w)=\delta_Y(w)/\delta_X(w)$. (c) A WLATE $\beta_\omega=\E[\omega(W_i)\beta(W_i)]$ is point identified iff these two conditions hold for all $w$ with $\omega(w)>0$. (d) Moreover, whenever $\beta_\omega$ is identified there exists a measurable $b(W_i)$ with $b(W_i)\delta_X(W_i)\ge0$ a.s.\ and $\E[b(W_i)\delta_X(W_i)]>0$ such that
\begin{equation}\label{eq:iv-rep}
\beta_\omega=\frac{\E[b(W_i)\,\delta_Y(W_i)]}{\E[b(W_i)\,\delta_X(W_i)]}.
\end{equation}
The weights, then, can be written as $\omega(W_i)={b(W_i)\,\delta_X(W_i)}/{\E[b(W_i)\,\delta_X(W_i)]}$. Conversely, any ratio of the form \eqref{eq:iv-rep} corresponds to an identified WLATE with weights $\omega(W_i)$ as just described.
\end{theorem}

Equation \eqref{eq:iv-rep} resembles an IV-like Wald ratio. In this analogy, $\delta_X(W_i)$ plays the role of the ``treatment'' variable, $\delta_Y(W_i)$ plays the role of the ``outcome'' variable, and $b(W_i)$ plays the role of the ``instrumental'' variable. This is different from a typical IV setting, however, where the conditions are typically satisfied only for a specific IV. In our case, there is no ``endogeneity'' issue, because the cutoff conditions in Assumption \ref{ass:rdd} already guarantee that $\delta_Y(w)=\beta(w)\delta_X(w)$. Thus, any measurable $b(W_i)$ that satisfies $b(W_i)\delta_X(W_i)\geq 0$ a.s. and $\E[b(W_i)\delta_X(W_i)]>0$ can play the role of the ``IV'' for an identified WLATE.

Note that the WLATEs are averages of treatment effects among compliers at the cutoff. The identification of the $\beta(w)$ relies on Assumption \ref{ass:rdd}, which is equivalent to the standard RDD conditions, but conditional on $W_i$. Interpretations of WLATEs as averages of treatment effects away from the cutoff would require additional assumptions for extrapolation that are beyond the scope of this paper (see e.g. \citealt{angrist2015wanna}).

\subsection{Examples}

We briefly show how common RDD targets fit into Theorem~\ref{thm:WLATE}. Full expressions for weights and ``instruments'' are given in Table \ref{tab: LATEs}. Let $f_W$ denote the density of $W_i$.

\begin{example}[Unconditional LATE]\label{ex:U-simple}
The target of the fuzzy RDD with covariates estimator in \cite{Calonico_Cattaneo_Farrell_Titiunik}, $\beta_U$, is a WLATE under additional conditions provided in that paper (which include strong monotonicity and continuity of the covariate distribution across the cutoff). In our framework, $\beta_U$ can be written as \eqref{eq:iv-rep}, with $b(W_i)$ proportional to the relative likelihood of $W_i$ at $z_0$.
\end{example}

\begin{example}[Average Conditional LATE]\label{ex:A-simple}
The simple average $\beta_A=\E[\beta(W_i)]$ is a WLATE with $\omega(W_i)\equiv1$. By Theorem~\ref{thm:WLATE}, this requires $\delta_X(w)\neq0$ and $w\in\mathcal W_{z_0}$ for all $w$ in the support of $W_i$. In our framework, it corresponds to \eqref{eq:iv-rep} choosing a $b(W_i)$ that underweights high-compliance groups.
\end{example}

\begin{example}[Counterfactual WLATE]\label{ex:C-simple}
Given a counterfactual covariate distribution $f^\ast$ with support $\mathcal W^\ast\subseteq\mathcal W_{z_0}$, the counterfactual target $\beta_C=\int\beta(w)f^\ast(w)\,dw$ is a WLATE with $\omega(w)=f^\ast(w)/f_W(w)$, identified when $\delta_X(w)\neq0$ for all $w\in\mathcal W^\ast$. In our framework, it corresponds to \eqref{eq:iv-rep} with a $b(W_i)$ that switches the measure from $f_W$ to $f^*$ and underweights high-compliance groups.
\end{example}

\begin{example}[Maximal Average Social Welfare]\label{ex:S-simple}
Following the literature on Empirical Welfare
Maximization (see e.g. \citealt{kitagawa2018should}), if a planner treats only cells with nonnegative $\beta(w)$ and strong monotonicity holds, the corresponding WLATE $\beta_S$ is identified when $\delta_X(w)>0$ for all weighted cells. In our framework, it corresponds to \eqref{eq:iv-rep} choosing a $b(W_i)$ that drops groups where $\delta_Y(W_i)<0$ and underweights high-compliance groups.
\end{example}

\begin{table}[H]
\protect\caption{WLATE comparisons,  $\beta_{\omega}=\E[\omega(W_i)\beta(W_i)]$ \label{tab: LATEs}}
\vspace{-0.25in}
\begin{center}
\begin{tabular}
[c]{ccccc}\toprule
\textbf{WLATE} & \textbf{ID.CONDITION} &  $\bm{\omega(w)}$&$\bm{b(w)}$ & \textbf{IDENTIFICATION}   \\
\midrule\midrule
$\beta_{U}$ &\hspace{-1cm}{\footnotesize $\mathbb{P}(\delta_X(W_i)> 0|Z_i=z_0)>0$}&$\frac{\delta_X(w){f_{W|z_{0}}(w)}/{f_W(w)}}{ \E[\delta_X(W_i)|Z_i=z_0]}$&$\frac{f_{W|z_{0}}}{f_{W}(w)}$ &$\frac{\E[\delta_Y(W_i)|Z_i=z_0]}{\E[\delta_X(W_i)|Z_i=z_0]}$
\\\midrule
 $\beta_{A}$ &\hspace{-1.6cm}{\footnotesize $\mathbb{P}(\delta_X(W_i)\neq 0)=1$}& {\footnotesize $1$}& $\frac{1}{\delta_{X}(w)}$ & {\footnotesize $\E$}$\left[\frac{\delta_Y(W_i)}{\delta_X(W_i)}\right]$ 
\\\midrule
$\beta_{C}$ & \hspace{-.1cm}{\footnotesize $\mathbb{P}(\delta_X(W_i)\neq 0|W_i\in \mathcal{W}^*)=1$}& $\frac{f^{\ast
}(w)}{f_{W}(w)}$ &$\frac{f^{\ast}(w)}{f_{W}(w)\delta_{X}(w)}$ & {\footnotesize $\E^*$}$\left[\frac{\delta_Y(W_i)}{\delta_X(W_i)}\right]$ 
\\\midrule
$\beta_{S}$ &{\footnotesize $\mathbb{P}(\delta_X(W_i)\neq0|\beta(W_i)\neq 0)=1$}& $\frac{1(\delta_{Y}(w)\geq 0)}{\mathbb{P}(\delta_Y(W_i)\geq 0)}$ & $\frac{1(\delta_{Y}\geq 0)}{\delta_{X}(w)}$ &{\footnotesize $\E$}$\left[\frac{\delta_Y(W_i)}{\delta_X(W_i)}\Big|\text{\footnotesize  $\delta_Y(W_i)\geq 0$}\right]$
 \\\midrule
$\beta_{CW}$ &\hspace{-1.7cm}{\footnotesize $\mathbb{P}(\delta_X(W_i)\neq 0)>0$}& $\frac{\delta_{X}^{2}(w)}{\E\left[  \delta
_{X}^{2}(W_{i})\right]  }$ & {\footnotesize $\delta_{X}(w)$} & $\frac{\E[\delta_{X}(w)\delta_Y(w)]}{\E[  \delta
_{X}^{2}(W_{i})]  }$ 
\\
\midrule
\bottomrule
\end{tabular}
\end{center}
\vspace{-0.1in}
\footnotesize \begin{singlespace} Note: $\beta_U$: Unconditional LATE (assumes strong monotonicity and continuity of covariate distribution across the cutoff, $f_{W|z_{0}}$ denotes the conditional density of $W_{i}$ given $Z,$ evaluated at $z_0$); $\beta_A$: Average LATE; $\beta_C$: Counterfactual Average LATE ($\mathcal{W}^*$ is the support of $f^*,$ and  $\E^*$ is the expectation taken over $f^*$); $\beta_S$: Maximal Average Social Welfare Gain (assumes strong monotonicity); $\beta_{CW}$: Compliance-Weighted LATE (defined in the next section). \end{singlespace}
\medskip
\end{table}

Example \ref{ex:C-simple} illustrates how large the WLATE class is. Note that one can choose weights such that the $\beta(w)$ are averaged over a distribution outside the cutoff, or even over an artificial distribution that has no relation to the data-generating process at all. All WLATEs must still be interpreted as an average effect at the cutoff $z_0$ because the $\beta(w)$ are local-to-the-cutoff treatment effects, but potentially the WLATE may use a different mix of those observations compared to the mix implied by the actual distribution near the cutoff.

\subsection{Compliance-Weighted Identification}
\label{sec: robust identification}

The identification from the representation in (\ref{eq:iv-rep}) motivates us to maximize the correlation between the ``instrument'' $b(W_i)$ and the ``treatment'' $\delta_X(W_i)$. Specifically, consider measurable $b(\cdot)$ with $0<Var(b(W_i))<\infty$ satisfying $b(W_i)\delta_X(W_i)\ge0$ a.s. Let $\rho(b(W_i),\delta_X(W_i))$ denote the correlation between $b(W_i)$ and $\delta_X(W_i)$. We are interested in $b^{\ast}(\cdot),$ where

\vspace{-1.2cm}

\begin{equation}
\label{eq:opt-b}
b^{\ast}\in \arg\max_{b}\, \big|\rho\!\big(b(W_i),\delta_X(W_i)\big)\big|
\quad\text{s.t.}\quad
b(W_i)\,\delta_X(W_i)\ge 0 \ \text{a.s.},\text{ and }
0<Var\!\big(b(W_i)\big)<\infty.
\end{equation}

\begin{theorem}[Compliance-Weighted ``Instrument'']\label{thm:CW-instrument}
Let Assumption \ref{ass:rdd}, $\mathbb{P}(\delta_X(W_i)\neq0)>0,$ and $\E\!\big[\delta_X^2(W_i)\big]<\infty$ hold. Then $b^\ast$ solves \eqref{eq:opt-b} for all admissible data generating processes if and only if $b^\ast(w)=c\,\delta_X(w)$ a.s. for some $c>0$.
\end{theorem}

Thus, in the transformed model of Theorem \ref{thm:WLATE}, the admissible ``instrument'' that is maximally aligned with the first-stage is proportional to $\delta_X(W_i)$ itself. The associated WLATE equals
\begin{equation}
\label{eq:CWLATE}
\beta_{CW}
\;=\;
\frac{\E\!\big[\delta_X(W_i)\,\delta_Y(W_i)\big]}
     {\E\!\big[\delta_X^2(W_i)\big]}.
\end{equation}
By Theorem \ref{thm:WLATE}, the CWLATE is identified when (i) all weighted cells appear on both sides of the cutoff and (ii) $\mathbb{P}(\delta_X(W_i)\neq0)>0$.

While the RDD estimand $\beta_U$ (under strong monotonicity) weights conditional LATEs proportionally to the conditional first-stage at the cutoff, the compliance-weighted LATE $\beta_{CW}$ weights conditional LATEs proportionally to the square of the conditional first-stage in the whole distribution. The estimands differ in two ways. First, because the distribution of $W_i$ at the cutoff may be different from the distribution of $W_i$ in the population. Second, because $\beta_{CW}$ intensifies the tilting towards complier-heavy groups. Note that, if strong monotonicity holds and the first difference above is negligible, the ordering of the two estimands is informative of the direction of the selection on gains/costs of treatment. Specifically, $\beta_{CW}>\beta_U$ if, and only if, higher compliance rates are associated with higher treatment effects. Analogously, $\beta_{CW}<\beta_U$ if, and only if, higher compliance rates are associated with lower treatment effects, and $\beta_{CW}=\beta_U$ if, and only if, no association exists (see formal statement and proof in the end of Appendix \ref{app: id proofs}).

\subsection{WLATEs as Expected Policy Effects}\label{subsec:policy-brief}

WLATEs admit a policy interpretation under strong monotonicity. Consider a targeted policy where incentive packets are offered to individuals through a two-step process. First, a value of $w$ is drawn according to a policy-assigned probability function $p$. The policymaker can favor certain values $w$ over others by making $p(w)$ relatively higher. Once $w$ is drawn, the incentive is given randomly to someone with $W_i=w$ and $Z_i=z_0$. 

The effect of this policy is the WLATE with ``instrument'' $b(W_i)=p(W_i)/f_{W}(W_i),$ where recall $f_W$ is the density of $W_i$ in the population. In this setting, $b(w)$ denotes the relative targeting ``tilt'' of the policy in comparison with a policy where no value of $w$ is targeted.

The converse is also true. A WLATE with ``instrument'' $b(W_i)$ corresponds to the policy where $w$ is drawn from the distribution $p(w)=b(w)f_W(w)/\E[b(W_i)]$. In particular, the CWLATE corresponds to a policy with probabilities given by $p(w)=\delta_X(w)f_W(w)/\E[\delta_X(W_i)],$ which ``tilts'' incentive-offers towards complier-heavy groups. Contrast this with the standard RDD estimand $\beta_U$, which corresponds to a policy with $p(w)=(f_{W|z_0}^+(w)+f_{W|z_0}^-(w))/2,$ where incentives are given according to the population likelihood of $w$ at the cutoff.

Details of policy interpretation of the WLATEs and formal results are in Online Appendix~A.

\section{Estimation of the CWLATE}\label{RDDestimation}

We estimate the CWLATE in \eqref{eq:CWLATE} when $W_i$
has finite support (or is coarsened into finitely many cells).
Let $\mathcal W_i=\{w_1,\ldots,w_m\}$, $\pi_j=\mathbb{P}(W_i=w_j)>0$, and
$\{Y_i,X_i,Z_i,W_i\}_{i=1}^n$ be i.i.d. Define, for $V\in\{Y_i,X_i\}$,

\vspace{-.8cm}
\[
\mu_V^\pm(w)=\E[V\mid Z_i=z_0^\pm,W_i=w],\qquad \delta_V(w)=\mu_V^+(w)-\mu_V^-(w).
\]
The CWLATE equals
\begin{equation}\label{eq:CWLATE_discrete}
\beta_{CW}
=\frac{\sum_{j=1}^m \pi_j\,\delta_X(w_j)\,\delta_Y(w_j)}
       {\sum_{j=1}^m \pi_j\,\delta_X^2(w_j)}
\equiv \frac{\tau_Y}{\tau_X}.
\end{equation}
Relative to the standard fuzzy RDD estimand, CWLATE replaces \emph{unconditional}
discontinuities by an \emph{aggregation of cell-specific} discontinuities, with weights
proportional to squared first-stages.

\paragraph{Plug-in estimator.}
Let $\hat\pi_j=n^{-1}\sum_{i=1}^n 1\{W_i=w_j\}$. Estimate $\delta_V(w_j)$ by
local polynomials within each cell and set
\begin{equation}\label{eq:cw_est}
\hat\beta_{CW}
=\frac{\sum_{j=1}^m \hat\pi_j\,\hat\delta_X(w_j)\,\hat\delta_Y(w_j)}
       {\sum_{j=1}^m \hat\pi_j\,\hat\delta_X^2(w_j)}
\equiv\frac{\hat\tau_Y}{\hat\tau_X}.
\end{equation}

\paragraph{Local linear implementation (stacked form).}
For kernel $k$ and bandwidth $h_n$, define $r_1(z)=(1,z)'$,
$\tilde W_i=(1\{W_i=w_1\},\ldots,1\{W_i=w_m\})'$, and
$X_{i,1}=r_1(Z_i)\otimes\tilde W_i$, where $\otimes$ is the Kronecker product.
Let $k_{h_n}^+(z)=h_n^{-1}k(z/h_n)1\{z\ge0\}$ and $k_{h_n}^-(z)=h_n^{-1}k(z/h_n)1\{z<0\}$.
For $s\in\{+,-\}$,
\begin{equation}\label{eq:betahatpm_CW}
\hat\beta_V^{\,s}(h_n)=\arg\min_\beta\sum_{i=1}^n\big(V_i-\beta'X_{i,1}\big)^2 k_{h_n}^{\,s}(Z_i),
\qquad \hat\mu_V^{\,s}(h_n)=e_0'\hat\beta_V^{\,s}(h_n),
\end{equation}
where $e_0=[I_m,0]$ selects the $m$ intercepts. Then
$\hat\delta_V(h_n)=\hat\mu_V^+(h_n)-\hat\mu_V^-(h_n)$, and $\hat\delta_V(w_j)$ is its $j$th component.
Equivalently, one may run separate local linear regressions within each cell $W_i=w_j$.

\paragraph{Delta-method expansion.}
With $\tau_Y,\tau_X$ as in \eqref{eq:CWLATE_discrete}, a delta method yields
\begin{equation}\label{eq:expansion_CW}
\hat\beta_{CW}-\beta_{CW}
=\frac{\hat\tau_Y-\tau_Y}{\tau_X}
-\frac{\tau_Y(\hat\tau_X-\tau_X)}{\tau_X^2}
+\hat{\mathcal R},
\end{equation}
where $\hat{\mathcal R}=o_p\big((nh_n)^{-1/2}+h_n^2\big)$ under the regularity
conditions in Appendix~\ref{AppendixSectionRRD}. CWLATE is less sensitive to weak cells because
$\tau_X=\sum_j \pi_j\delta_X^2(w_j)$ downweights small $|\delta_X(w_j)|$ automatically.

\subsection*{Bias-corrected inference for the CWLATE}

With local-linear estimation of the cell discontinuities, $\hat\beta_{CW}$ has
a leading smoothing bias of order $h_n^2$. Let $b_n$ denote a pilot bandwidth
used to estimate this leading bias term via higher-order local polynomials (as done in 
in \citet{Calonico_Cattaneo_Titiunik}). Define the debiased CWLATE estimator

\vspace{-.5cm}

\begin{equation}\label{eq:cw_debiased}
\hat\beta_{CW}^{\,bc}
\;=\;
\hat\beta_{CW}-h_n^2\,\hat B_{CW,1,2}(h_n,b_n),
\end{equation}
where $\hat B_{CW,1,2}(h_n,b_n)$ is a consistent estimator of the leading bias
constant (its explicit expression is given in Appendix~\ref{AppendixSectionRRD}).

Bias estimation affects the variance. Let $\mathbf V_{CW,1,2}^{bc}(h_n,b_n)$
denote the RBC asymptotic variance of $\hat\beta_{CW}^{\,bc}$ and
$\hat{\mathbf V}_{CW,1,2}^{bc}(h_n,b_n)$ its consistent estimator; both are
reported in Appendix~\ref{AppendixSectionRRD}.

\begin{theorem}[RBC inference for the CWLATE]\label{thm:CW_RBC}
Let the kernel, smoothness, and the RBC bandwidth conditions
for $(h_n,b_n)$ stated in Appendix~\ref{AppendixSectionRRD} hold, then

\vspace{-.5cm}

\[
\big(\mathbf V_{CW,1,2}^{bc}(h_n,b_n)\big)^{-1/2}
\big(\hat\beta_{CW}^{\,bc}-\beta_{CW}\big)
\ \xrightarrow{d}\ N(0,1),
\]
and $\hat{\mathbf V}_{CW,1,2}^{bc}(h_n,b_n)$, given in Appendix~\ref{AppendixSectionRRD}, is consistent for
$\mathbf V_{CW,1,2}^{bc}(h_n,b_n)$. Hence Wald confidence intervals based on
standard normal critical values are asymptotically valid.
\end{theorem}

\section{Monte Carlo Simulations}\label{sec: MC}

We compare the CWLATE estimator to (i) the standard fuzzy RDD and (ii) fuzzy RDD with covariates, using the Stata command \texttt{RDROBUST} (\citealt{Calonico_Cattaneo_Titiunik}).

\subsection{Simulation Setup}\label{sec: MC setup}

Data are generated through the d.g.p. 

\vspace{-1.3cm}
\begin{align*}
Y_i &= 1 + 2X_i + 0.3Z_i - 0.1X_iZ_i + \beta_{XW} X_i W_i + U_{Y_i},\\
D_i &= 1\{Z_i\ge 0\},\quad 
X_i = 1\{-1 + 0.2Z_i + 0.1D_iZ_i + 1.2D_i + \alpha_{DW}D_iW_i + U_{X_i} > 0\},
\end{align*}

\vspace{-.4cm}

\noindent where $(U_{X_i},U_{Y_i}) \sim \mathcal{N}((0,0)', (1, .5;.5 , 2))$. The variables $Z_i$ and $W_i$ are independent, with $Z \sim \mathcal{N}(0,2)$, and $W_i$ takes two values, $-1$ and $1$, with equal likelihood. We consider two different values of the parameter that influences the heterogeneity in first-stages across different values of $W_i$, $\alpha_{DW}\in\{0,1\}$, and two different values of the analogous parameter that influences the heterogeneity in treatment effects, $\beta_{XW}\in\{0,2\}$. Note that Assumption \ref{ass:rdd} holds in this model. In fact, strong monotonicity holds for the parameter values we consider, so that the standard RDD estimators in the literature, which do not adapt to weak monotonicity, can be used as benchmark in this simulation.

For each $(\alpha_{DW},\beta_{XW})$ and $n\in\{300,\ldots,5000\}$ we run 10{,}000 replications and compute MSE for: the standard RDD $\hat\beta_U$, RDD with covariates $\hat\beta_U^{cov}$, and CWLATE $\hat\beta_{CW}$. We implement the two RDD estimators via \texttt{RDROBUST} with MSE‐optimal bandwidths (\citealt{Calonico_Cattaneo_Titiunik,Calonico_Cattaneo_Farrell_Titiunik}). To isolate \emph{weighting} as the only difference, CWLATE uses the \emph{same} bandwidth chosen for $\hat\beta_U^{cov}$ and restricts to observations within that bandwidth, thus differences in estimands cannot be due to differences in $f_W$ and $f_{W|z_0}$.

Figures~\ref{fig: First Stage Plots}--\ref{fig: heterogeneity plots} illustrate typical first‐stage and reduced‐form plots (bins of width $0.067$) for $n=5000$. The first-stage discontinuities are large, so this d.g.p. does not suffer from a weak-IV problem (note that the d.g.p. does not change with the sample size).

\begin{figure}[H]
\begin{center}
\caption{First-Stage Plots}\label{fig: First Stage Plots}
\begin{subfigure}[b]{0.45\textwidth}\caption{$\alpha_{DW}=0$}
\includegraphics[scale=0.55]{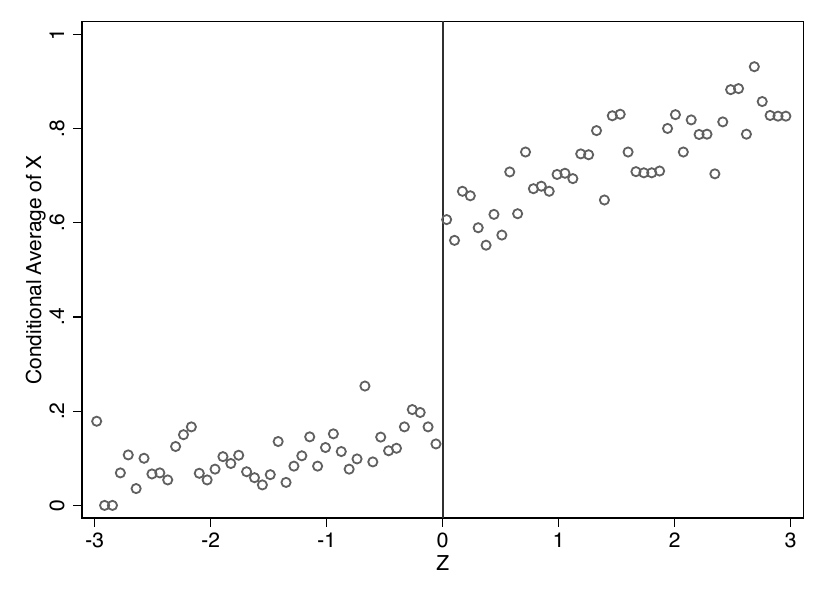}\end{subfigure}\hspace{.1in}
\begin{subfigure}[b]{0.45\textwidth}\caption{$\alpha_{DW}=1$}
\includegraphics[scale=0.55]{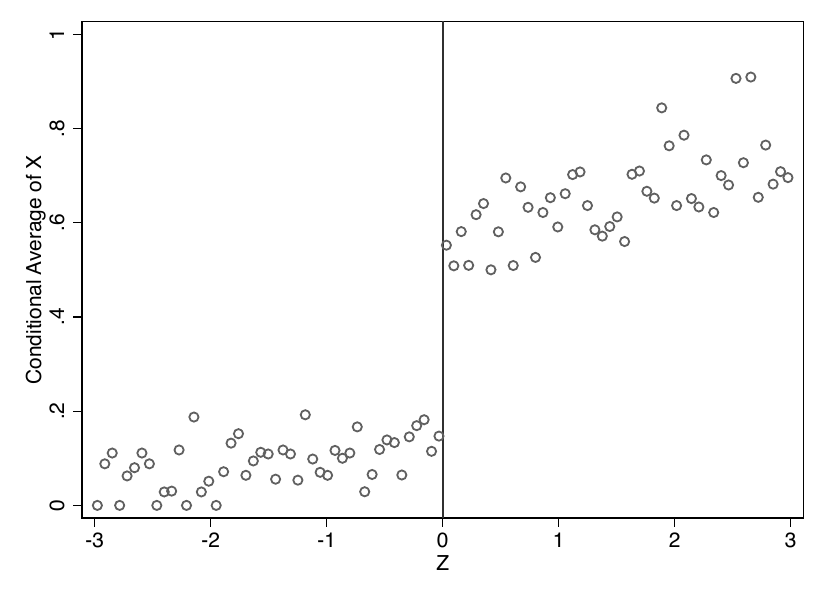}\end{subfigure}
\end{center}
\vspace{-.5in}
{\footnotesize \singlespacing Notes: Averages for $Z_i$ in bins of size $0.067$. Sample size is 5000.}
\end{figure}

\begin{figure}[H]
\begin{center}
\caption{Reduced Form Plots}\label{fig: RF Stage Plots}
\begin{subfigure}[b]{0.45\textwidth}\caption{$\alpha_{DW}=0,\ \beta_{XW}=0$}
\includegraphics[scale=0.55]{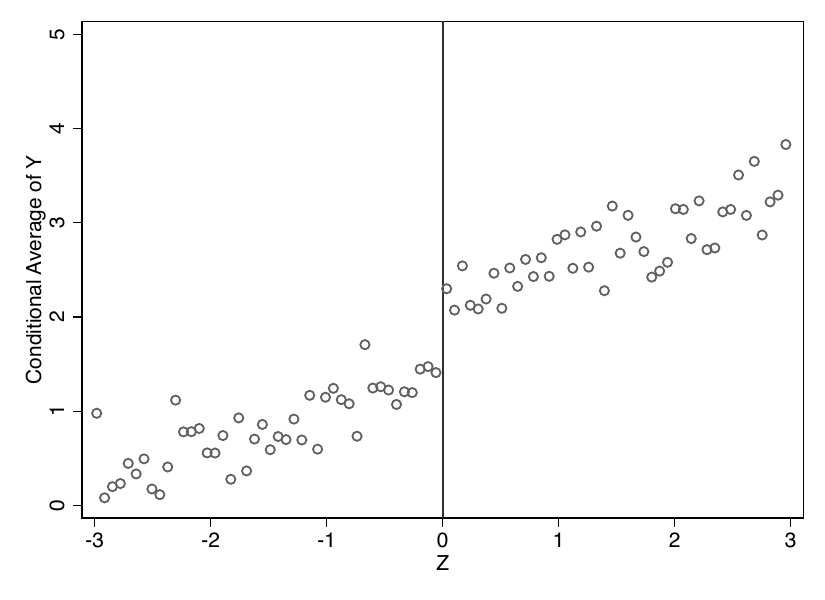}\end{subfigure}\hspace{.1in}
\begin{subfigure}[b]{0.45\textwidth}\caption{$\alpha_{DW}=1,\ \beta_{XW}=0$}
\includegraphics[scale=0.55]{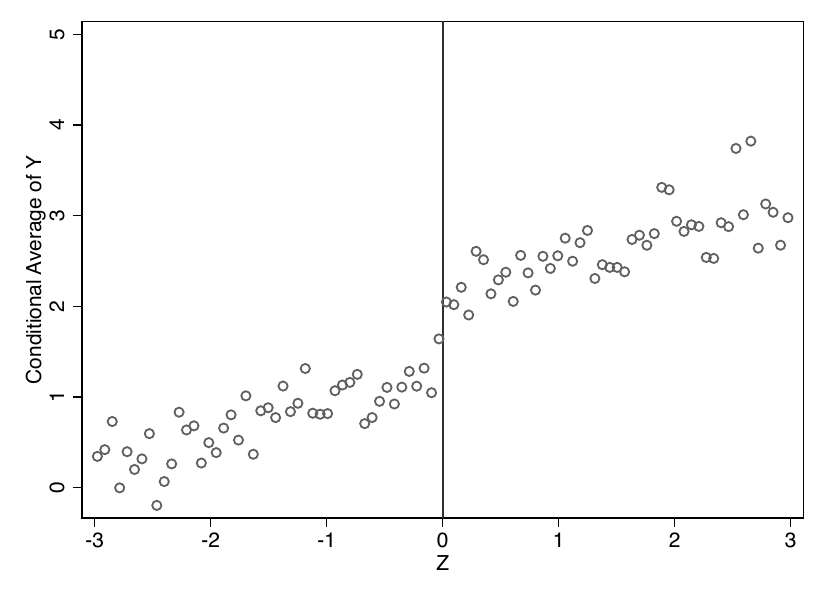}\end{subfigure}\\
\begin{subfigure}[b]{0.45\textwidth}\caption{$\alpha_{DW}=0,\ \beta_{XW}=2$}
\includegraphics[scale=0.55]{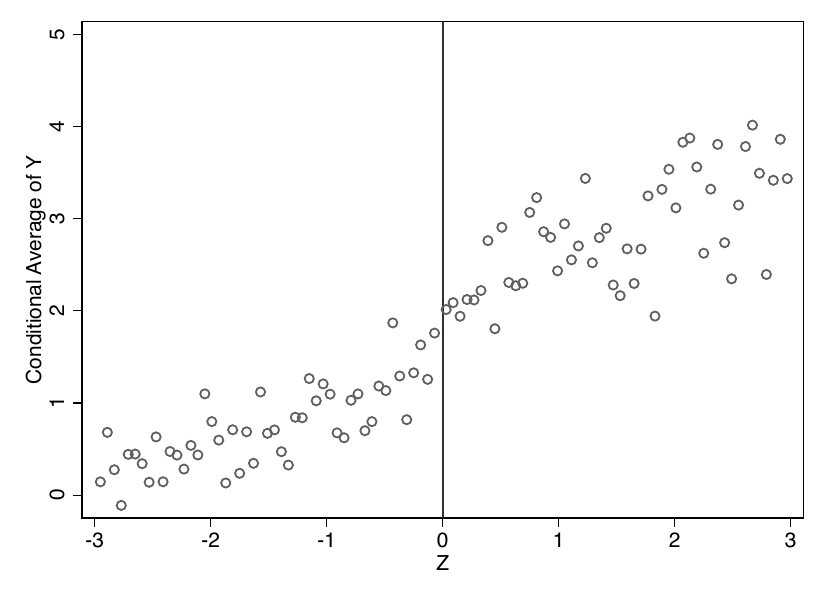}\end{subfigure}\hspace{.1in}
\begin{subfigure}[b]{0.45\textwidth}\caption{$\alpha_{DW}=1,\ \beta_{XW}=2$}
\includegraphics[scale=0.55]{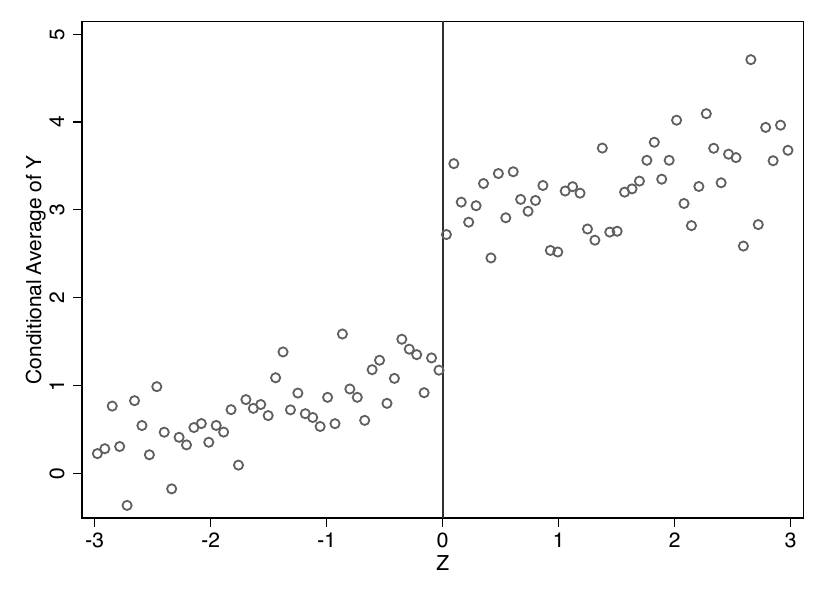}\end{subfigure}
\end{center}
\vspace{-.5in}
{\footnotesize \singlespacing Notes: Same binning as Fig.~\ref{fig: First Stage Plots}. Sample size is 5000.}
\end{figure}

\begin{figure}[H]
\begin{center}
\caption{Conditional Plots ($\alpha_{DW}=1,\ \beta_{XW}=2$)}\label{fig: heterogeneity plots}
\begin{subfigure}[b]{0.45\textwidth}\caption{First-Stage $|\,W_i=1$}
\includegraphics[scale=0.55]{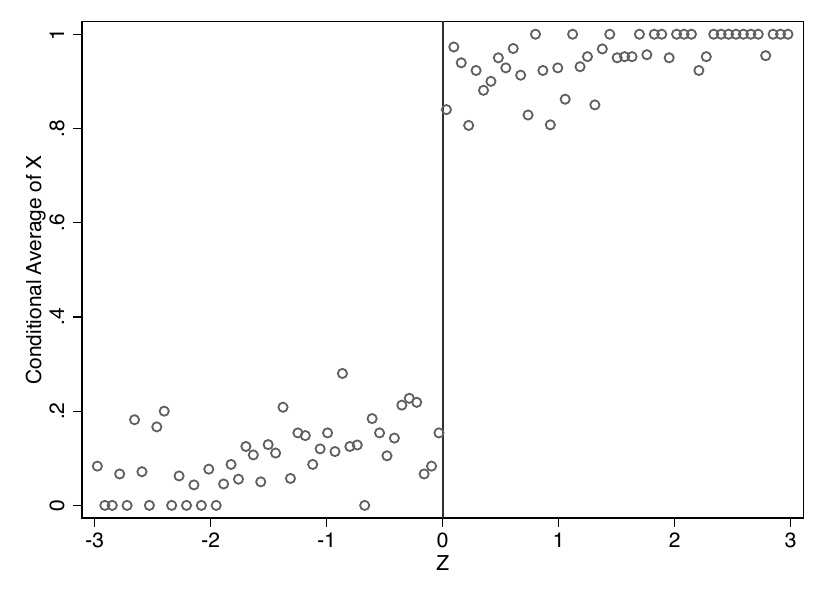}\end{subfigure}\hspace{.1in}
\begin{subfigure}[b]{0.45\textwidth}\caption{First-Stage $|\,W_i=-1$}
\includegraphics[scale=0.55]{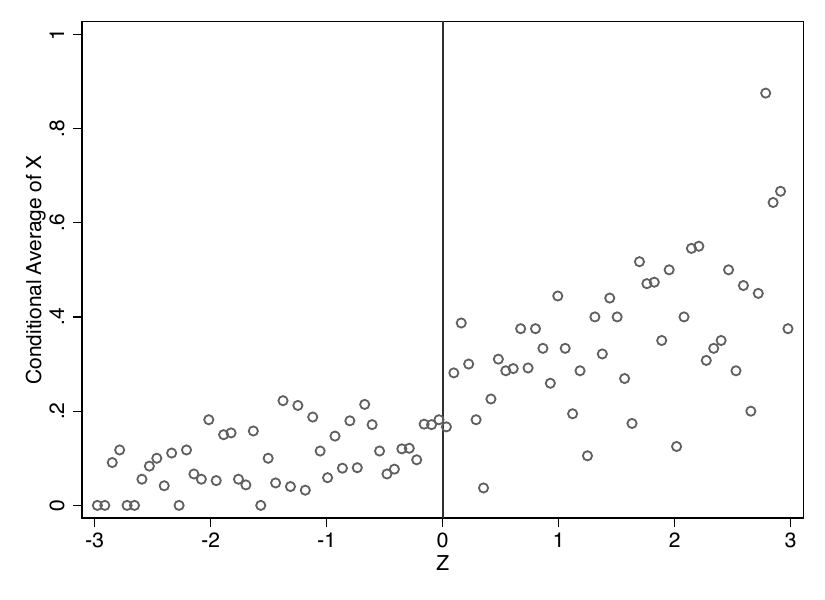}\end{subfigure}
\begin{subfigure}[b]{0.45\textwidth}\caption{Reduced Form $|\,W_i=1$}
\includegraphics[scale=0.55]{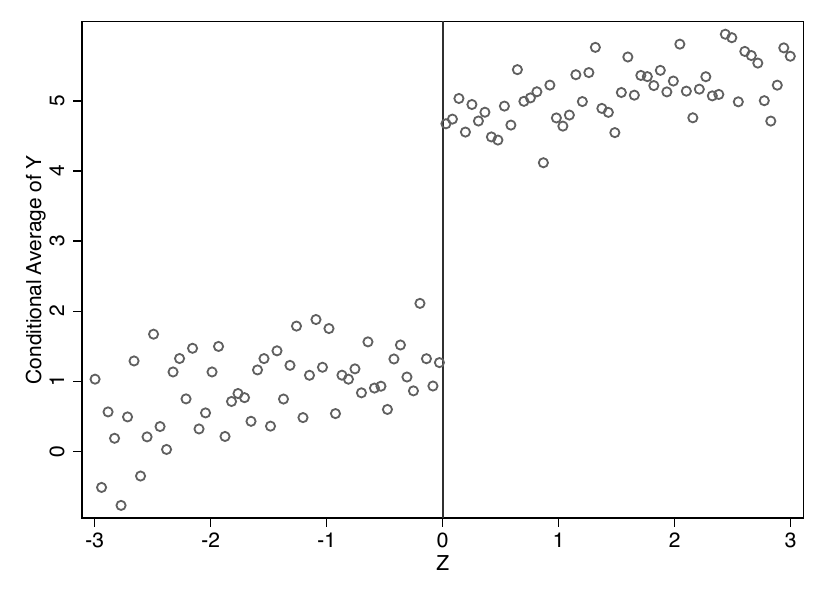}\end{subfigure}\hspace{.1in}
\begin{subfigure}[b]{0.45\textwidth}\caption{Reduced Form $|\,W_i=-1$}
\includegraphics[scale=0.55]{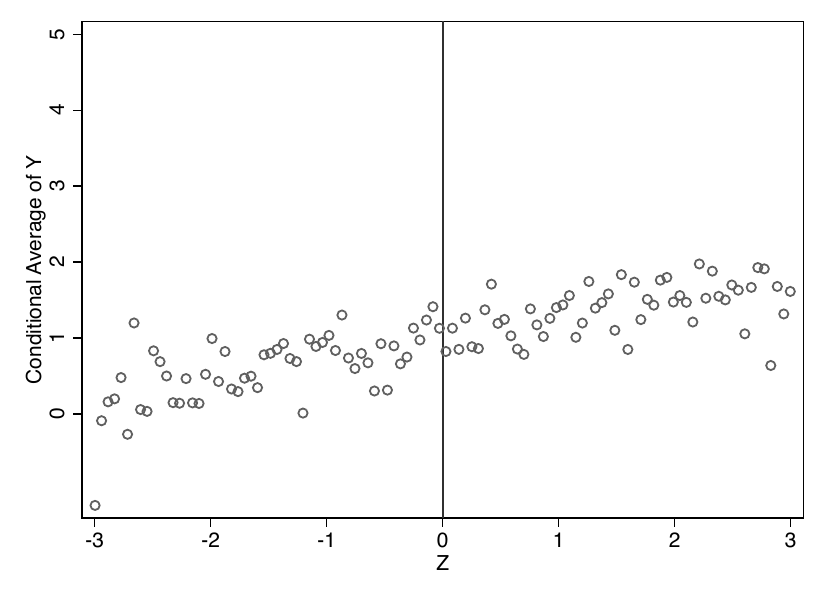}\end{subfigure}
\end{center}
\vspace{-.5in}
{\footnotesize \singlespacing Notes: Same binning as Fig.~\ref{fig: First Stage Plots}. Sample size is 5000.}
\end{figure}

Note that this Monte Carlo was designed to favor the performance of standard RDD estimators for the following reasons: (a) we assumed $W_i\Perp (Z_i,U_{X_i},U_{Y_i})$ and added no extra $Z_i\!\times\!W_i$ interactions; allowing them strengthens CWLATE further; (b) the baseline first-stage is relatively strong; weaker first-stages enlarge CWLATE's advantage, and may even lead to violations of strong monotonicity, which is required for the standard RDD estimators; (c) the bandwidth used for the CWLATE is the optimal one for $\hat{\beta}^{cov}_U$, and thus not necessarily the optimal one for the CWLATE itself.

\subsection{Results}

Table~\ref{tab: RDD} reports MSEs relative to the corresponding estimand (Unconditional LATE $\beta_U$ for $\hat\beta_U$ and $\hat\beta_U^{cov}$; Compliance-Weighted LATE $\beta_{CW}$ for $\hat\beta_{CW}$).

The main findings are as follows.
\begin{enumerate}[(i)]
\item For small samples, the standard RDD estimators display weak‐identification patterns, while CWLATE is stable in all cases.
\item The first two panels show the edge case where there is no first-stage heterogeneity ($\alpha_{DW}=0$).\footnote{In this case, $\delta_X(W_i)$ is constant, and therefore $\rho(b(W_i),\delta_X(W_i))=0$ for all WLATEs, thus the CWLATE has no advantage in theory. Moreover, note that Theorem \ref{thm:CW-instrument} considers only ``instruments'' with a positive variance, which the CWLATE does not satisfy in this case.} In large samples, all estimators perform very similarly when there is no heterogeneity in treatment effects ($\beta_{XW}=0$), with the standard RDD estimators taking the lead when the heterogeneity in treatment effects increases ($\beta_{XW}=2$).

\item In the last two panels, the first-stages are heterogeneous ($\alpha_{DW}=1$). Here, Theorem \ref{thm:CW-instrument} establishes that the CWLATE has the strongest identification. Unsurprisingly, the CWLATE estimator outperforms the RDD estimators in all cases. In fact, although all estimators perform worse as the treatment effects heterogeneity increases, the comparative advantage of the CWLATE estimator remains (compare the third and fourth panels, and also the third and fourth panels in Table \ref{tab:increaseheterogeneity} in Appendix \ref{sec: MC more het}, where $\beta_{XW}=5$ and 10, respectively). This is not surprising, since the result that establishes that the CWLATE has the strongest identification (Theorem \ref{thm:CW-instrument}) does not depend on $\beta(w)$.
\end{enumerate}

\begin{table}[H]
\caption{MSE of Unconditional LATE and Compliance-Weighted LATE Estimators}\label{tab: RDD}
\vspace{-.2in}
\begin{center}
\scalebox{.9}{\small \begin{tabularx}{1.13\linewidth}{c|rrr|rrr|rrr|rrr}

\toprule
 Obs.  & \multicolumn{6}{c}{\(\alpha_{DW}=0\)} & \multicolumn{6}{c}{\(\alpha_{DW}=1\)} \\  & \multicolumn{3}{c}{\(\beta_{XW}=0\)} & \multicolumn{3}{c}{\(\beta_{XW}=2\)} & \multicolumn{3}{c}{\(\beta_{XW}=0\)} & \multicolumn{3}{c}{\(\beta_{XW}=2\)} \\ & \(\hat{\beta}_{U}\) & \(\hat{\beta}^{cov}_{U}\) & \(\hat{\beta}_{CW}\) & \(\hat{\beta}_{U}\) & \(\hat{\beta}^{cov}_{U}\) & \(\hat{\beta}_{CW}\) & \(\hat{\beta}_{U}\) & \(\hat{\beta}^{cov}_{U}\) & \(\hat{\beta}_{CW}\) & \(\hat{\beta}_{U}\) & \(\hat{\beta}^{cov}_{U}\) & \(\hat{\beta}_{CW}\) \tabularnewline
\midrule \addlinespace[\belowrulesep]
300&300.00&1024.76&3.51&805.38&1285.75&4.26&641.67&637.22&1.27&3363.10&6153.34&1.70 \tabularnewline \addlinespace[.05in]
500&3.24&1.82&1.71&6.27&14.81&2.59&8.30&19.44&.64&267.24&18.65&.75 \tabularnewline \addlinespace[.05in]
1000&.46&.47&.44&.84&.71&1.02&2.82&1.21&.28&.89&.85&.31 \tabularnewline \addlinespace[.05in]
2000&.21&.21&.21&.37&.31&.50&.25&.26&.13&.36&.35&.14 \tabularnewline \addlinespace[.05in]
3000&.13&.13&.13&.24&.20&.33&.17&.17&.09&.23&.23&.09 \tabularnewline \addlinespace[.05in]
4000&.10&.10&.10&.17&.14&.24&.12&.12&.07&.17&.17&.07 \tabularnewline \addlinespace[.05in]
5000&.08&.08&.08&.14&.12&.20&.09&.09&.05&.13&.13&.06 \tabularnewline \addlinespace[.05in]
\bottomrule 

\end{tabularx}}
\end{center}
\end{table}

In Appendix \ref{sec: MC coarse}, we consider the scenario where the covariates that modulate the heterogeneity can take many values. We show that the performance of the CWLATE reduces only a little when the researcher uses a coarsened (binarized) version of the covariate, instead of the full covariate.

The results above suggest that standard RDD estimators tend to be more fragile to heterogeneity and to smaller samples than the CWLATE, which is to be expected given Theorem \ref{thm:CW-instrument}. In Appendix \ref{sec:betterestimator} we illustrate the issues of identification of the standard RDD estimand, $\beta_U,$ by considering the behavior of the three estimators when all of them target $\beta_U$. Here, the CWLATE is an inconsistent estimator of $\beta_U$. Yet, the first two panels in Table \ref{tab: RDD2} and Figure \ref{fig: MSE comp} show that the variance and weak identification problems of the standard RDD are serious enough that an inconsistent estimator can perform better even when the inconsistency is substantial. 

This does not mean that researchers should use the CWLATE as an estimator of $\beta_U$. Using an inconsistent estimator to target $\beta_U$ would result in wrongly-sized tests and wrong-coverage confidence intervals, as we show in the last two panels in Table \ref{tab: RDD2}. Rather, our point is that substantive claims about the effects of treatment on compliers might be better established with different weighted averages that can be estimated with considerably more precision. This point is further illustrated in the next section.

\section{Empirical Application}\label{application}

We consider the study of the effect of women receiving cash transfers while pregnant on the likelihood that their newborns have low birthweight (i.e., less than $2{,}500$ grams), as in \cite{amarante2016cash}, henceforth AMMV. A monthly cash-transfer policy of around \$100 was implemented in Uruguay from April 2005 to December 2007. Eligibility to receive cash transfers depended on a continuously distributed predicted income score, which was a linear combination of several household characteristics obtained from an interview. After this score was computed, eligibility was determined based on whether it fell below a pre-determined cutoff.\footnote{Neither the interviewers nor the households were informed about the exact formula to compute the predicted score, nor the exact cutoff, mitigating concerns of manipulation around the cutoff.} AMMV restrict the sample to households with a pregnant woman and ask whether the treatment (defined as the household receiving at least one cash transfer during pregnancy) affected the likelihood that the child was born with low birthweight.

Figure~\ref{fig: application first stage} shows the first-stage plot, where it is clear that families just below the cutoff are discontinuously more likely to have received a cash transfer while pregnant than those just above the cutoff.
\begin{figure}[H]
\caption{Unconditional First-Stage}%
\label{fig: application first stage}%
\vspace{-.25in}
\begin{center}
\includegraphics[scale=0.6]{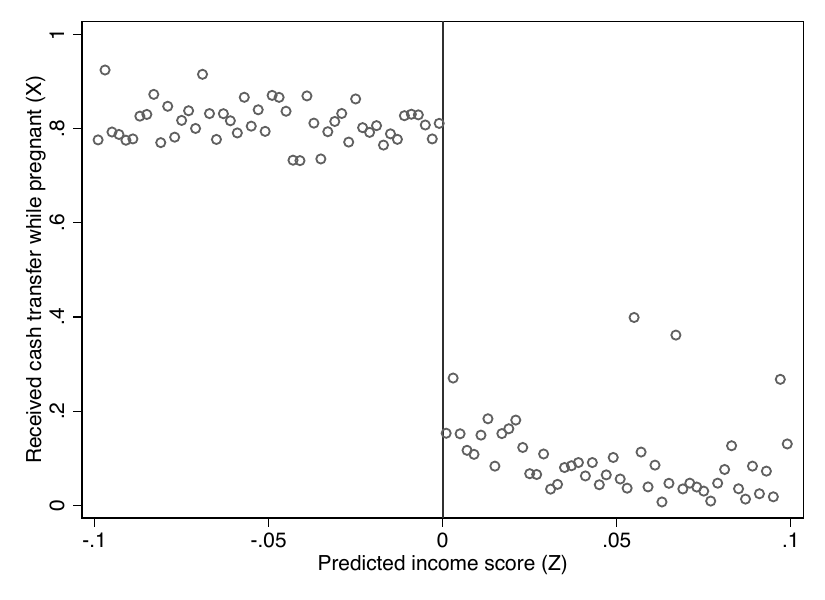}
\end{center}

\vspace{-0.4in}
{\footnotesize \singlespacing \textit{Notes}: The plot shows the average of $X$ for each value of $Z_i$, measured in bins of size  $0.002$.}
\end{figure}

\begin{figure}[!ht]
\caption{First-Stage by $W_i$}%
\label{fig: application first stage by W}%
\vspace{-.25in}
\begin{center}
\includegraphics[scale=0.55]{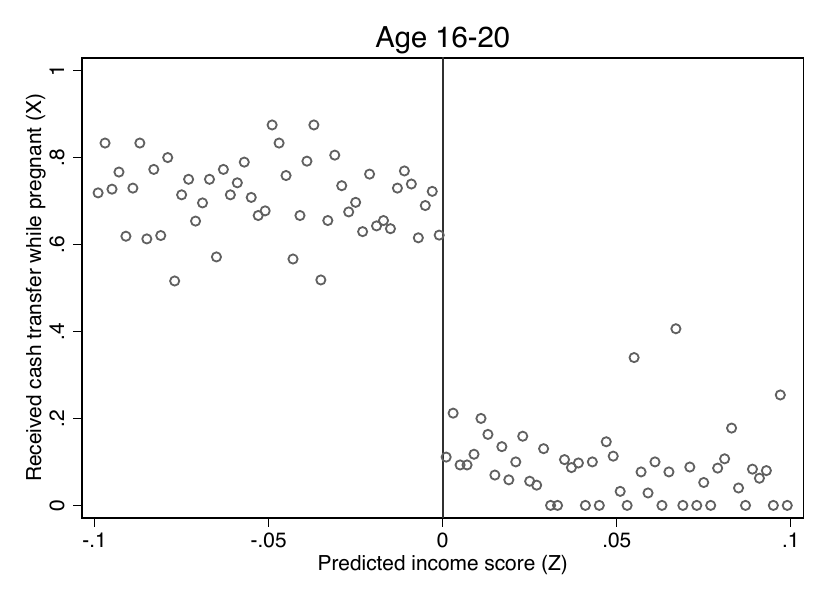}
\medspace\medspace \includegraphics[scale=0.55]{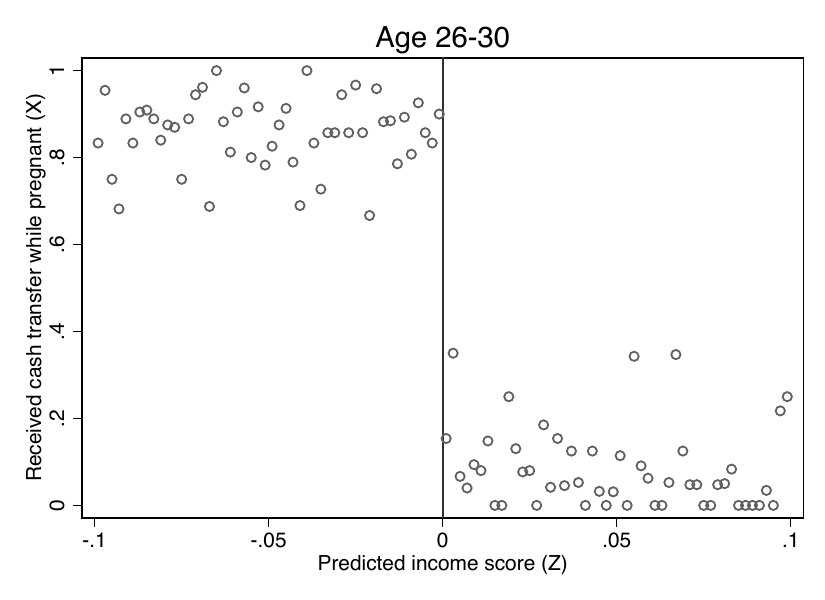}
\end{center}

\vspace{-0.4in}
{\footnotesize \singlespacing \textit{Notes}: The plot shows the average of $X_i$ for each value of $Z_i$, measured in bins of size  $0.002$. Each panel conditions the sample to a different given value of $W,$ represented in the caption.}

\end{figure}

Figure~\ref{fig: application first stage by W} shows first-stages for different values of the covariates $W_i$, specifically for mothers aged 16–20 (left panel) and mothers aged 26–30 (right panel). The first-stages are clearly heterogeneous in age, with younger mothers being less likely to be compliers.

Table \ref{tab: main results application} reports estimated effects of $X_i$
(receiving cash transfers while pregnant) on $Y_i$ (probability of low
birthweight) for a range of bandwidths. Column~1 reports the standard fuzzy
RDD Wald estimator based on unconditional discontinuities; Column~2 reports the
conventional covariate-adjusted fuzzy RDD estimator (implemented via local
regressions with covariates); Column~3 reports the CWLATE estimator, which
aggregates covariate-cell reduced-form and first-stage discontinuities using
compliance-weighting as in \eqref{eq:CWLATE_discrete}--\eqref{eq:cw_est}. To
make the comparison as clean as possible, we align the estimators on the same
local window: for each bandwidth choice, the CWLATE is computed using the same
observations within that bandwidth, and differs from the covariate-adjusted
Wald estimator only through the implied reweighting across covariate cells.
As covariates $W_i$, we include interactions between
the age of the mother, the education of the father, and the period when the
initial visit to determine eligibility occurred.\footnote{Specifically, we
consider interactions of the following variables: indicators of each age of
the mother from age 16 to 46 (bottom coded and top coded for the two extreme
ages), indicators for the education of the father (incomplete primary or less,
incomplete secondary, incomplete tertiary, complete tertiary, or missing,
which includes lack of presence of father), and indicators of the period when
the official visit to the household to obtain the inputs of the income score
happened (03/2005 or before, 04/2005--06/2005, 07/2005--09/2005,
10/2005--12/2005, 01/2006 or after). The findings were similar for different
breakdowns of these variables.}

\begin{table}[H]
\protect\caption{Main Results \label{tab: main results application}}
\vspace{-.2in}
\begin{center}
\begin{tabularx}{\textwidth}{cc|CC|CC|p{2cm}C}

\toprule
 & & \multicolumn{4}{c}{Unconditional LATE} & \multicolumn{2}{c}{CWLATE} \\ \midrule \vspace{-.2in} Bandwidth & N & \multicolumn{2}{c}{Standard RDD} & \multicolumn{2}{c}{RDD w/ Covariates} & \multicolumn{2}{c}{CWLATE Estimator} \tabularnewline
{}&{}&{}&{}&{}&{}&{}&{} \tabularnewline
\midrule\addlinespace[1.5ex]
0.04&4,591&-0.0427&(0.0289)&-0.0471*&(0.0281)&-0.0519**&(0.0232) \tabularnewline \addlinespace[.05in]
0.05&5,849&-0.0415&(0.0255)&-0.0446*&(0.0248)&-0.0478**&(0.0207) \tabularnewline \addlinespace[.05in]
0.06&7,109&-0.0354&(0.0229)&-0.0373*&(0.0224)&-0.0441**&(0.0192) \tabularnewline \addlinespace[.05in]
0.07&8,300&-0.0300&(0.0211)&-0.0313&(0.0207)&-0.0377**&(0.0180) \tabularnewline \addlinespace[.05in]
0.08&9,247&-0.0265&(0.0197)&-0.0285&(0.0195)&-0.0339**&(0.0171) \tabularnewline \addlinespace[.05in]
0.09&10,258&-0.0240&(0.0187)&-0.0259&(0.0184)&-0.0291*&(0.0162) \tabularnewline \addlinespace[.05in]
0.10&11,276&-0.0230&(0.0177)&-0.0244&(0.0175)&-0.0245&(0.0155) \tabularnewline \addlinespace[.05in]
\bottomrule \addlinespace[1.5ex]

\end{tabularx}

\end{center}

\vspace{-.2in}
\footnotesize \begin{singlespace} 
    Note: Heteroskedasticity-robust standard errors in parentheses. **: Significant at 5\%. *: Significant at 10\%.
 \end{singlespace}
\end{table}

The two standard RDD estimates are very similar to those reported in AMMV (e.g. see AMMV's Figure A5), and tend to be insignificant at standard levels. However, because of the low precision of the RDD estimators in this application, AMMV's leading results are obtained using a different identification strategy.\footnote{Specifically, they implement a ``localized difference-in-differences'' approach with one of the differences being the regression discontinuity in the first column with a large bandwidth, equal to 0.10. For context, the optimal MSE-based bandwidths for both the standard RDD and RDD with covariates in this application are around 0.04.} With this approach, they find an effect of around -0.02 and a standard error of around 0.01.

In contrast to the standard RDD estimates, the CWLATE estimates in the last column tend to be significant at standard levels for most bandwidths.\footnote{To isolate \emph{weighting} as the only difference, CWLATE uses the \emph{same} bandwidth chosen for $\hat\beta_U^{cov}$ and restricts to observations within that bandwidth, thus differences in estimands cannot be due to differences in $f_W$ and $f_{W|z_0}$.} The main reason the CWLATE estimates tend to be more significant than the standard RDD estimates is that their standard error tends to be smaller. Nevertheless, because the CWLATE estimates are slightly more negative than the standard RDD estimators with the same bandwidth, there is also some evidence of selection on gains--values of $W_i$ with higher compliance tend to have better (more negative) treatment effects (see Section \ref{sec: robust identification}).\footnote{These results also allow us to indirectly assess the sign of $\beta_U$ from the sign of $\beta_{CW}$. Under strong monotonicity, if the sign of $\beta(W_i)$ does not change across values of $W_i$, then $\beta_U$ and $\beta_{CW}$ must have the same sign. Under this assumption, our results are evidence that $\beta_U$ is indeed negative, even if the estimates are not significant. This is because our estimates of $\beta_{CW}$ are negative and significant.} Therefore, the CWLATE estimator provides evidence based on a cross-sectional RDD identification strategy that the Compliance-Weighted LATE estimand is negative.

The results in Table \ref{tab: main results application} showcase how the CWLATE estimator could be useful in practice. First, it provides a relatively precise estimate of a meaningful weighted average of treatment effects on compliers based on an RDD identification strategy. There is no \emph{ex ante} reason to accept the Unconditional LATE as a satisfactory estimand and, at the same time, reject the CWLATE in this application. Both estimands consider only observations near the cutoff, and weight positively the exact same covariate groups of compliers. The only difference is that the Unconditional LATE weights by the proportion of compliers while the CWLATE weighs by the square of the proportion of compliers, thus tilting towards the higher compliance groups. 

Second, the large standard errors of the standard RDD estimators may lead researchers to make ``second best'' decisions that effectively would take them farther away from the Unconditional LATE than the use of the CWLATE estimator would. For instance, researchers may use larger-than-optimal bandwidths, as done by AMMV. Note that the vertical differences in Table \ref{tab: main results application} (across different bandwidths) are much more pronounced than the horizontal differences (across the different estimators, and different estimands in the case of the CWLATE), suggesting that the choice of bandwidth may be more consequential than the choice between these two estimands in this context. Researchers may go as far as to abandon the RDD identification strategy altogether in favor of an alternative that yields a more precise estimate, where the identification assumptions might be stronger and the estimand may not be the Unconditional LATE at all. For example, this is the case with the difference-in-differences with controls leading estimate chosen by AMMV, which often deviates from desirable target estimands and may not even identify a treatment effect (see \citealt{caetano2022difference}).

Overall, we are able to confirm with an RDD identification strategy the qualitative findings from \cite{amarante2016cash} that the cash-transfers policy led to a reduction in the incidence of low birthweight among recipients. Our estimates are over twice as large as their estimates, suggesting that cash-transfers contributed to a reduction in the incidence of low-birthweight of the complier population from 10\% to about 5–6\%.
\section{Conclusions}\label{conclusion}

This paper studies identification and estimation in fuzzy RDDs with covariates when treatment compliance varies across observed subgroups. We characterize a broad class of WLATEs, defined as weighted averages
of covariate-specific LATEs at the cutoff, and provide necessary and sufficient
conditions for their point identification. Any identified WLATE admits an IV-type representation, where the specific WLATE is determined by the function of the covariates that plays the role of the ``instrument.'' This representation allows us to single out the Compliance-Weighted LATE (CWLATE) as the WLATE with the strongest ``instrument'' among all WLATEs. In comparison with the estimand of the standard RDD estimators, the CWLATE puts relatively more weight on groups with more compliers.

We provide a policy interpretation of the WLATEs. We also propose simple plug-in estimators in the case of discrete covariates, derive their asymptotic distribution, and extend robust bias-corrected inference methods to this setting, allowing researchers to implement CWLATE estimation with standard RDD tools and bandwidth selectors.

Monte Carlo simulations confirm the theoretical strength of the CWLATE target, and expose the practical issues with the standard RDD estimators. An application
to Uruguay's cash transfer program for women during pregnancy illustrates these issues. While the standard RDD estimators yield inconclusive results, the CWLATE estimator shows that transfers substantially reduced the probability of low birthweight among compliers.

Overall, our results support a simple message for applied work with fuzzy RDDs
and covariates: instead of defaulting to the unconditional LATE, researchers
can transparently report WLATEs, and in particular the CWLATE, which are
identified under standard RDD assumptions and explicitly aligned with where the
design is strongest. This perspective preserves the credibility of the RDD
while improving empirical content in applications where first-stage variation
is unevenly distributed across observed covariates.

\vspace{-.5cm}

\singlespace
\bibliographystyle{apalike}
\bibliography{caetano_escanciano.bib}


\clearpage
\setcounter{page}{1}
\pagenumbering{arabic}
\doublespacing
\setcounter{section}{0}
\renewcommand{\thesection}{\Alph{section}}
\setcounter{table}{0}
\renewcommand{\thetable}{\Alph{section}.\arabic{table}}
\setcounter{figure}{0}
\renewcommand{\thefigure}{\Alph{section}.\arabic{figure}}

\begin{center}
\Large Online Appendix
\end{center}
\section{Policy Interpretation of the WLATEs}

We interpret the WLATE estimands using a policy decision framework, and consider the interpretation of the CWLATE in that context. 

Suppose that a policymaker targets the offer of a treatment incentive based on the value of covariates, $W_i.$ Each individual who will receive the incentive is selected in a two-step process where first a value $w$ is drawn from a distribution chosen by the policymaker (so that incentives are more or less likely to reach certain groups of the population according to the policymaker's wishes). The incentive is then given to a random individual among those with $W_i=w$ at the cutoff. We describe the policy precisely in the following assumption.

\begin{assumption}\label{as:policy} 
Let $W_i$ be continuously distributed in the population, with density function $f_W$.\footnote{This assumption is made for simplicity, as we can also allow $W_i$ to not be continuously distributed. For example, if $W_i$ is discrete with support $\mathcal{W}=\{w_1,\dots,w_m\},$ $p$ in Assumption \ref{as:policy} may be defined as a set of probabilities that a given value of $W_i\in \mathcal{W}$ is drawn: $p_1,\dots,p_m,$ with $\sum_{j=1}^m p_j=1$. All quantities defined henceforth are analogous. For example, let $f_j=\mathbb{P}(W_i=w_j),$ $\bm{P_C}=\sum_{j=1}^m \delta_X(w_j)p_j=\sum_{j=1}^m \left(\frac{p_j}{f_j}\delta(w_j)\right) f_j$. Extension of results to the discrete case is straightforward.} The policymaker observes the value of $W_i$ for all individuals in the population. A treatment incentive is offered to some of these individuals. Each individual who receives the incentive is selected as follows.\begin{itemize}
\item[]\label{as:determ}\textbf{Step 1:} A value $w$ is drawn from a distribution $p,$ where $p(w)\geq0,$ $\int p(w)dw=1,$  $\int p(w)^2dw<\infty$ and, if $f_W(w)=0,$ then $p(w)=0$. 
\item[] \label{as:random}\textbf{Step 2:} The incentive is offered to an individual who is randomly drawn among those with $W_i=w$ and $Z_i=z_0$. 
\end{itemize}

\noindent We further assume that
\begin{enumerate}[(i)]
\item \label{as:mon}Nobody offered the incentive decreases the amount of treatment. 
\item \label{as:sample}A sample of vectors $(Y_i,X_i,Z_i,W_i')'$ which is representative of the population is available, where identical incentives to those intended by the policy were offered to those with $Z_i\geq z_0$.
\end{enumerate}
\end{assumption}

Let $A\in \mathcal{W}$ be a $p$-measurable set. The policy described above assures that roughly a proportion $\int_A p(w)dw$ of the incentives are allocated for individuals with $W_i\in A$. If one were to assign incentives randomly among the population, this likelihood would be $\int_A f_W(w)dw.$ Therefore, $p(w)/f_W(w)$ represents the relative ``tilt'' of the policy in targeting certain values of $w$ over others. Note that Step 1 guarantees that no incentives may be allocated outside the support of the distribution of $W_i$ (thus, for completeness, define $p(w)/f_W(w)=0$ whenever $f_W(w)=0$).

Since the incentive units assigned to those with $W_i=w$ and $Z_i=z_0$ are distributed randomly in that group, some of the incentives may be offered to individuals that will be treated anyway (always takers), and some of the incentives may be offered to individuals that will not take the treatment anyway (never takers). These individuals are not affected by the policy, since their treatment status will be the same whether or not they are offered the incentive. The policy only affects those who change their treatment status as a consequence of receiving the incentive. Assumption \ref{as:policy}\eqref{as:mon} rules out the existence of defiers by assuming strong monotonicity. It is feasible, but more cumbersome, to devise a policy interpretation of the WLATEs under weak monotonicity.

Assumption \ref{as:policy}\eqref{as:sample} makes it possible to tie the intended policy to the existing data. Precisely, while the data is generated from an RDD quasi-experiment, the policy would be applied in a population using the exact procedure described in Assumption \ref{as:policy}. Therefore, to identify the effects of the described policy, the data must be representative of the policy population, and the incentives considered by the policy must be identical to the incentives faced by the observations above the cutoff in the data.

If Assumption \ref{ass:rdd} holds, then
per incentive given, the probability it reaches a complier, denoted $\textbf{P}_\textbf{C},$ satisfies
\begin{equation*}
    \textbf{P}_\textbf{C} = \int_{\mathcal{W}} \delta_X(w)p(w)dw=  E\left[\frac{p(W_i)}{f_W(W_i)} \delta_X(W_i)\right].
\end{equation*}
The first equality follows because    $p(w)$ is the probability that $w$ is selected, and $\delta_X(w)$ is the probability that a random individual with $W_i=w$ and $Z_i=z_0$ is a complier. 

Since the incentive has no effect on always takers and never takers, the expected effect of each incentive offered, denoted  $\textbf{APE}$ (for average policy effect per incentive unit) is
\begin{equation*}
     \textbf{APE}= \int_{\mathcal{W}} \beta(w)\delta_X(w)p(w)dw=\E\left[\frac{p(W_i)}{f_W(W_i)} \delta_X(W_i)\beta(W_i)\right].
\end{equation*}
 The expected effect per incentive unit on the compliers that are incentivized is the LAPE, where L stands for ``local'' to compliers. It is given by
\begin{equation*}
      \textbf{LAPE}=\frac{\textbf{APE}}{\textbf{P}_\textbf{C}}
      =\frac{E\left[\frac{p(W_i)}{f_W(W_i)}\delta_X(W_i)\beta(W_i)\right]}{E\left[\frac{p(W_i)}{f_W(W_i)}\delta_X(W_i)\right]}.
\end{equation*}

The following result is an immediate consequence of Theorem \ref{thm:WLATE}.
\begin{theorem}\label{thm:policy}
If Assumptions \ref{ass:rdd} and \ref{as:policy} hold, then the APE and the LAPE are identified if, and only if, $p(W_i)$ is identified and the conditions of Theorem \ref{thm:WLATE} hold a.s.\ when $p(W_i)\delta_X(W_i)>0$. The LAPE is a WLATE with weights
\begin{equation*}
    \omega(W_i)=\frac{p(W_i)\delta_X(W_i)/f_W(W_i)}{E\left[p(W_i)\delta_X(W_i)/f_W(W_i)\right]},
\end{equation*}
and ``instrument''
\begin{equation*}
        b(W_i)= c\cdot p(W_i)/f_W(W_i),
\end{equation*} 
where $c>0$ is a constant. Conversely, a WLATE with ``instrument'' $W_i$ is a LAPE with $p(w)=b(w)f_W(w)/\E[b(W_i)]$.
\end{theorem}

This result informs us that every identifiable WLATE with ``instrument'' $b$ can be interpreted as the expected per-incentive unit effect of a policy of the type described in Assumption \ref{as:policy} among those who are affected by it. For each incentive unit offered, the policy consists on drawing a value $w$ from the distribution function $p(w)=b(w)f_W(w)/E[b(W_i)],$ and giving the incentive to a random individual with $W_i=w$ and $Z_i=z_0$. 

By this reasoning, the estimands discussed in the examples in Section \ref{identification} can be interpreted as the expected effect among affected compliers per unit of incentive, where the incentives are distributed according to specific targeted policies. The corresponding targeting distribution $p$ is as follows.  

\begin{example}\label{exaa} [Unconditional LATE, cont.] $\beta_U$ (Example \ref{ex:U-simple}) is the LAPE of a policy with 
\begin{equation*}
    p(w)=f_{W|z_{0}}(w).
\end{equation*}
\end{example}
The Unconditional LATE estimand $\beta_U$ is well known to correspond to the average effect among compliers of a policy where incentives are distributed entirely randomly among those at the cutoff. This ``random'' policy can be understood equivalently as a ``targeted'' policy such as the one described in Assumption \ref{as:policy} where the values of $w$ are drawn from the population distribution of $W_i$ at the cutoff. 

\begin{example}\label{hh} [Average LATE, cont.] $\beta_A$ (Example \ref{ex:A-simple}) is the LAPE of a policy with 
\begin{equation*}
    p(w)=\frac{f_W(w)}{\delta_X(w) E[\delta_X(W_i)^{-1}]}.
\end{equation*}
\end{example}

\begin{example}\label{hjh} [Counterfactual Average LATE, cont.] $\beta_C$ (Example \ref{ex:C-simple}) is the LAPE of a policy with
\begin{equation*}
    p(w)=\frac{f^*(w)}{\delta_X(w) E[\delta_X(W_i)^{-1}|W_i\in \mathcal{W}^*]}.
\end{equation*}
\end{example}

\begin{example}\label{hsh}[Maximal Average Social Welfare Gain, cont.] $\beta_S$ (Example \ref{ex:S-simple}) is the LAPE of a policy with 
\begin{equation*}
    p(w)=\frac{f_W(w)}{\delta_X(W_i) E[\delta_X(W_i)^{-1}|\delta_Y(W_i)\geq 0]}\cdot\frac{1(\delta_Y(W_i)\geq0)}{\mathbb{P}(\delta_Y(W_i)\geq 0)}.
\end{equation*} 
\end{example}

The policies implied by the estimands $\beta_A,$ $\beta_C$ and $\beta_S$ tilt away from a random distribution in that they are less likely to give incentives to the groups with more compliers (i.e. when $\delta_X(W_i)$ is relatively large). $\beta_C$ also tilts in favor of  $W_i$ values that are proportionally more likely in the counterfactual distribution than in the actual distribution. Note that $\beta_S$ only draws from groups with non-negative conditional LATEs.

\begin{example}\label{iij}[Compliance-Weighted LATE] $\beta_{CW}$ (Section \ref{sec: robust identification}) is the LAPE of a policy with 
\begin{equation*}
        p(w)=\frac{\delta_X(w)f_W(w)}{ E[\delta_X(W_i)]}.
\end{equation*}
\end{example}

The policy implied by the CWLATE estimand, $\beta_{CW},$ tilts from a random distribution in giving more incentives to groups with a larger proportion of compliers. 

Policymakers might have a variety of different goals when distributing incentive units to individuals. For instance, a policymaker who wants to reduce income inequality might want to tilt an incentive policy towards values of $W_i$ that are associated with poor families. Moreover, the policy decision may depend on contrasting expected effects per incentive unit (APEs) or expected effects on affected compliers per incentive unit (LAPEs) of several policies at the same time, taking into account targeting parameters (e.g., the income cutoff to receive the treatment) and other criteria (e.g., the policymaker may want to take into account the policy cost in the decision). Theorem \ref{thm:policy} determines the range of policies for which the effect can be identified, and gives the formula of these effects.

\newpage 
\section{Identification Proofs}\label{app: id proofs}

\begin{proof}[Proof of Theorem \ref{thm:WLATE}]

\noindent\textbf{(a)}
Using $Y_i = Y_i(0)+\beta_i X_i$,
\begin{align*}
\delta_Y(w)
&= \E[Y_i\mid Z_i=z_0^+,W_i=w]-\E[Y_i\mid Z_i=z_0^-,W_i=w]\\
&= \underbrace{\E[Y_i(0)\mid Z_i=z_0^+,W_i=w]-\E[Y_i(0)\mid Z_i=z_0^-,W_i=w]}_{=0}
\\[-0.2em]
&\quad
+ \lim_{e\downarrow0}\Big\{
\E[\beta_i X_i\mid Z_i=z_0+e,W_i=w]
-\E[\beta_i X_i\mid Z_i=z_0-e,W_i=w]
\Big\}.
\end{align*}
By Assumption \ref{ass:rdd}(i) the $Y_i(0)$ term is continuous at $z_0$.
For $e>0$,
\[
\E[\beta_i X_i\mid Z_i=z_0+e,W_i=w]
-\E[\beta_i X_i\mid Z_i=z_0-e,W_i=w]
= \E[\beta_i \Delta_i(e)\mid Z_i=z_0,W_i=w],
\]
where $\Delta_i(e)=X_i(z_0+e)-X_i(z_0-e)$.
Assumption \ref{ass:rdd}(ii) and (iii), plus dominated convergence, imply
$\lim_{e\downarrow0}\E[\beta_i \Delta_i(e)\mid Z_i=z_0,W_i=w]
= \E[\beta_i\Delta_i\mid Z_i=z_0,W_i=w]$,
where $\Delta_i=\lim_{e\downarrow0}\Delta_i(e)$.
By weak monotonicity (iv), at each $w$ compliers and defiers do not coexist, so
\[
\E[\beta_i\Delta_i\mid Z_i=z_0,W_i=w]
= \E[\beta_i\mid Z_i=z_0,W_i=w,\Delta_i\neq0]\,
  \E[\Delta_i\mid Z_i=z_0,W_i=w].
\]
By definition,
$\beta(w)=\E[\beta_i\mid Z_i=z_0,W_i=w,\Delta_i\neq0]$ and
$\delta_X(w)=\E[\Delta_i\mid Z_i=z_0,W_i=w]$,
hence $\delta_Y(w)=\beta(w)\delta_X(w)$.

\medskip
\noindent\textbf{(b)}
\emph{Sufficiency.}
If $\delta_X(w)\neq0$ and $w\in\mathcal W_{z_0}$, then
$\E[Y_i\mid Z_i=z,W_i=w]$ and $\E[X_i\mid Z_i=z,W_i=w]$ are identified for
$z\to z_0^\pm$, so $\delta_Y(w)$ and $\delta_X(w)$ are identified and $\beta(w)=\delta_Y(w)/\delta_X(w)$.

\emph{Necessity.}
If $\delta_X(w)=0$, then for any $\phi$,
$\tilde\beta(w)=\beta(w)+\phi$ satisfies
$\tilde\beta(w)\delta_X(w)=\beta(w)\delta_X(w)=\delta_Y(w)$, so $\beta(w)$ is
not unique.
If $w\notin\mathcal W_{z_0}$, then $w$ has no support arbitrarily close to one
side of the cutoff, so at least one of $\delta_Y(w),\delta_X(w)$ is not
identified, and any $\beta(w)$ consistent with $\delta_Y(w)=\beta(w)\delta_X(w)$
is observationally equivalent.
Hence conditions (i)–(ii) are necessary and sufficient.

\medskip
\noindent\textbf{(c)}
Let $\omega\ge0$ a.s.\ with $\E[\omega]=1$ and
$\beta_\omega=\E[\omega(W_i)\beta(W_i)]$.

\emph{Sufficiency.}
If (i)–(ii) in (b) hold for all $w$ with $\omega(w)>0$, then each such
$\beta(w)$ is identified, so $\beta_\omega$ is identified.

\emph{Necessity.}
Suppose there exists $w^\ast$ with $\omega(w^\ast)>0$ where either
$\delta_X(w^\ast)=0$ or $w^\ast\notin\mathcal W_{z_0}$.
By (b) we can change $\beta(w^\ast)$ without affecting observables, obtaining
$\tilde\beta$ such that
\[
\tilde\beta_\omega
= \E[\omega(W_i)\tilde\beta(W_i)]
= \beta_\omega + \phi\,\omega(w^\ast)
\]
for arbitrary $\phi$, so $\beta_\omega$ is not point identified.
Thus the stated condition is also necessary.

\medskip
\noindent\textbf{(d)}
\emph{From WLATE to IV form.}
If $\beta_\omega$ is identified, then by (c) for all $w$ with $\omega(w)>0$,
$\delta_X(w)\neq0$ and $w\in\mathcal W_{z_0}$.
Define
\[
b(w)
=
\begin{cases}
c\,\dfrac{\omega(w)}{\delta_X(w)}, & \omega(w)>0,\\[0.3em]
0, & \omega(w)=0,
\end{cases}
\]
with $c>0$ fixed.
Then $b(w)\delta_X(w)=c\,\omega(w)\ge0$ and
$\E[b(W_i)\delta_X(W_i)]=c>0$, so
\[
\omega(W_i)
= \frac{b(W_i)\delta_X(W_i)}{\E[b(W_i)\delta_X(W_i)]}.
\]
Using part (a),
$\beta(W_i)\delta_X(W_i)=\delta_Y(W_i)$ wherever relevant, hence
\[
\beta_\omega
= \E[\omega(W_i)\beta(W_i)]
= \frac{\E[b(W_i)\delta_Y(W_i)]}{\E[b(W_i)\delta_X(W_i)]},
\]
which yields \eqref{eq:iv-rep}.

\emph{From IV form to WLATE.}
Conversely, let $b(W_i)$ be identified with
$b(W_i)\delta_X(W_i)\ge0$ a.s.\ and $\E[b(W_i)\delta_X(W_i)]>0$, and assume
$\delta_Y(W_i),\delta_X(W_i)$ are identified where $b(W_i)\delta_X(W_i)>0$.
Define
\[
\omega(W_i)
= \frac{b(W_i)\delta_X(W_i)}{\E[b(W_i)\delta_X(W_i)]},
\]
so $\omega\ge0$ and $\E[\omega]=1$.
By (a),
\[
\frac{\E[b(W_i)\delta_Y(W_i)]}{\E[b(W_i)\delta_X(W_i)]}
= \E[\omega(W_i)\beta(W_i)]
= \beta_\omega.
\]
Thus any such ratio corresponds to an identified WLATE with weights $\omega$.
\end{proof}

\bigskip 
\begin{proof}[Proof of Theorem \ref{thm:CW-instrument}]

We restate the problem:
\[
b^{\ast}
= \arg\max_b\; |\rho(b(W_i),\delta_X(W_i))|
\quad\text{s.t.}\quad
b(W_i)\delta_X(W_i)\ge0\ \text{a.s.},
\]
where the maximization is over measurable $b(\cdot)$ with
$0<Var(b(W_i))<\infty$ and $\mathbb{P}(\delta_X(W_i)\neq0)>0$.

Because $\rho$ is a correlation coefficient, we have
$|\rho(b(W_i),\delta_X(W_i))|\le1$, with equality if and only if there
exist constants $c_0\in\mathbb R$ and $c_1\neq0$ such that
\[
b(W_i)=c_0 + c_1\delta_X(W_i)
\quad\text{a.s.}
\]

Now impose the sign restriction $b(W_i)\delta_X(W_i)\ge0$ a.s.:
\[
b(W_i)\delta_X(W_i)
= c_0\delta_X(W_i) + c_1\delta_X(W_i)^2 \ge0\quad\text{a.s.}
\]

We interpret this restriction as a structural requirement: for any data
generating process satisfying Assumption \ref{ass:rdd} and
$\mathbb{P}(\delta_X(W_i)\neq0)>0$, the resulting $b(\cdot)$ must satisfy
$b(W_i)\delta_X(W_i)\ge0$ a.s.\ under that process. Under
Assumption~\ref{ass:rdd}(iv), for each $w$ the sign of $\delta_X(w)$ is
well-defined (weak monotonicity), but across $w$ both positive and negative
values of $\delta_X(w)$ are admissible. Hence we must allow for designs in
which $\delta_X(W_i)$ takes both positive and negative values.

If $c_0\neq0$, then for some admissible design we can choose values of
$\delta_X(W_i)$ with opposite signs and sufficiently small magnitude so that
$c_0\delta_X(W_i)$ changes sign while $c_1\delta_X(W_i)^2$ is arbitrarily
small, violating $b(W_i)\delta_X(W_i)\ge0$ with positive probability.
Therefore, to satisfy the sign restriction for all admissible designs we must
have $c_0=0$. Given $\mathbb{P}(\delta_X(W_i)\neq0)>0$, the inequality with $c_0=0$
implies $c_1>0$. Thus any admissible maximizer must be of the form
\[
b^{\ast}(W_i)=c\,\delta_X(W_i),\qquad c>0.
\]

Conversely, for any $c>0$, the choice $b(W_i)=c\,\delta_X(W_i)$ satisfies
$b(W_i)\delta_X(W_i)\ge0$ a.s.\ and yields
$|\rho(b(W_i),\delta_X(W_i))|=1$, so it attains the maximum of the
optimization problem.

Hence $b^{\ast}$ solves \eqref{eq:opt-b} if and only if
$b^{\ast}(w)=c\,\delta_X(w)$ with $c>0$.
\end{proof}

\bigskip

\noindent {\fontsize{12.5}{13.2}\selectfont \textbf{Ordering of $\bm{\beta_U}$ and $\bm{\beta_{CW}}$}

\begin{proposition}
    Suppose that strong monotonicity holds, $f_{W|z_0}(w)=f_W(w)$ a.s., and that $\beta_U$ and $\beta_{CW}$ are identified. Then $\rho(\delta_X(W_i),\beta(W_i))$ and $\beta_{CW}-\beta_U$ have the same sign.
\end{proposition}

\begin{proof} 
    Define a new probability measure $Q$ such that for any random variable $V$, the expectation under $Q$ is weighted by $\delta_X(W_i)$:$$E_Q[V] = \frac{E[\delta_X(W_i)V]}{E[\delta_X(W_i)]}$$
    Under this measure, we can rewrite  $\beta_U = E_Q[\beta(W_i)]$ and  $\beta_{CW} = \frac{E_Q[\delta_X(W_i)\beta(W_i)]}{E_Q[\delta_X(W_i)]}$.  
    
    By the definition of covariance under measure $Q$:$$E_Q[\delta_X(W_i)\beta(W_i)] - E_Q[\delta_X(W_i)]E_Q[\beta(W_i)] = \text{cov}_Q(\delta_X(W_i), \beta(W_i)).$$
Since $E_Q[\delta_X(W_i)] > 0$, we divide both sides by it, and obtain
    $$\beta_{CW}-\beta_U=\frac{\text{cov}_Q(\delta_X(W_i), \beta(W_i))}{E_Q[\delta_X(W_i)]}.$$
    The result then follows because the sign of the left-hand side must be the same as the sign of $\rho(\delta_X(W_i),\beta(W_i))$.
\end{proof}

\newpage 
\section{Asymptotic Theory\label{AppendixSectionRRD}}

\subsection{Point Estimation and Asymptotic Theory for WLATEs}

In this section, we develop point estimation and asymptotic theory for WLATEs. For estimation, we
consider the practical situation where $W_i$ takes on a finite set of values. Let $\mathcal{W}=\{w_{1},...,w_{m}\}$
denote the support of $W,$ with $\pi_{j}=\mathbb{P}\left(  W_i=w_{j}\right)  >0.$ We
assume we observe an independent and identically distributed (iid) random
sample $\{Y_{i},X_{i},Z_{i},W_{i}\}_{i=1}^{n}$ of size $n.$ Let $\{Y_i,X_i,Z_i,W_i\}$
be a random vector with the same distribution as $\{Y_{i},X_{i},Z_{i}%
,W_{i}\}.$ For a generic random vector $V,$ we use the notation%
\begin{align*}
\mu_{V}^{+}(w) &  =\E[V|Z_i=z_0^+,W_i=w],\\
\mu_{V}^{-}(w) &  =\E[V|Z_i=z_0^-,W_i=w],
\end{align*}
$\mu_{V}^{+}=(\mu_{V}^{+\prime}(w_{1}),...,\mu_{V}^{+\prime}(w_{m}))^{\prime
},$ $\mu_{V}^{-}=(\mu_{V}^{-\prime}(w_{1}),...,\mu_{V}^{-\prime}(w_{m}))^{\prime}$ and $\delta_{V}=\mu_{V}^{+}-\mu_{V}^{-}.$ When $V$ is a function of $W,$ our notation is such that  
$\mu_{V}^{+}$ and $\mu_{V}^{-}$ correspond to the standard RDD quantities that
do not condition on $W_i.$ For example, for $V=W_i,$ $
\mu_{V}^{+}  =\E[W_i|Z_i=z_0^+],$ and $
\mu_{V}^{-} =\E[W_i|Z_i=z_0^-].$
 We follow similar
vector notation for functions of $w$. For example $\pi=(\pi_{1},...,\pi
_{m})^{\prime}$. Without loss of generality, assume hereinafter that $z_{0}=0.$ 

Define $c=(c_{1},...,c_{m})^{\prime}$ with $c_{j}=\pi
_{j}b(w_{j}),$ for $j=1,...,m.$ With this notation in place, the generic WLATE
estimand is%
\[
\beta_{\omega}=\frac{\sum_{j=1}^{m}\pi_{j}b(w_{j})\delta_{Y}(w_{j})}%
{\sum_{j=1}^{m}\pi_{j}b(w_{j})\delta_{X}(w_{j})}=\frac{c^{\prime}\delta_{Y}%
}{c^{\prime}\delta_{X}}\equiv\frac{\tau_{Y}}{\tau_{X}},
\]
where the dependence of the $\tau^{\prime}s$ on $c$ is dropped for simplicity.
To account for the different examples mentioned above, we assume that $c$ is
generated according to the following specification%
\begin{equation}
c=g(\mu_{V}^{+},\mu_{V}^{-},\pi)\label{c1}%
\end{equation}
for a known vector-valued function $g,$ and different choices of $V.$ When $g$
depends on $\mu_{V}^{+}$ and $\mu_{V}^{-}\ $only through the term $\mu_{V}%
^{+}-\mu_{V}^{-},$ we write $c=g(\delta_{V},\pi).$ 

The proposed estimator of $\beta_{\omega}$ in this setting is%
\[
\hat{\beta}_{\omega}=\frac{\hat{\tau}_{Y}}{\hat{\tau}_{X}},
\]
where $\hat{\tau}_{Y}=\hat{c}^{\prime}\hat{\delta}_{Y},$ $\hat{\tau}_{X}%
=\hat{c}^{\prime}\hat{\delta}_{X},$ $\hat{c}=g_{j}(\hat{\mu}_{V}^{+},\hat{\mu
}_{V}^{-},\hat{\pi}),$ $\hat{\delta}_{V}=\hat{\mu}_{V}^{+}-\hat{\mu}_{V}^{-},$
$\hat{\pi}$ is the vector of sample probabilities, and $\left(  \hat{\mu}%
_{V}^{+},\hat{\mu}_{V}^{-}\right)  $ are local polynomial estimates. In the
main text we present results for local-linear, but our asymptotic results in
the Appendix allow for local-polynomial estimates of any order. To introduce
the local-linear estimators, define for a generic random variable $V,$ kernel function $k$
and bandwidth $h_{n}$ the quantities%
\begin{align}
\hat{\delta}_{V} &  \equiv\hat{\delta}_{V}(h_{n})=\hat{\mu}_{V}^{+}%
(h_{n})-\hat{\mu}_{V}^{-}(h_{n}),\nonumber\\
\hat{\mu}_{V}^{+} &  \equiv\hat{\mu}_{V}^{+}(h_{n})=e_{0}^{\prime}\hat{\beta
}_{V}^{+}(h_{n}),\nonumber\\
\hat{\mu}_{V}^{-} &  \equiv\hat{\mu}_{V}^{-}(h_{n})=e_{0}^{\prime}\hat{\beta
}_{V}^{-}(h_{n}),\nonumber\\
\hat{\beta}_{V}^{+}(h_{n}) &  =\arg\min_{\beta}\sum_{i=1}^{n}\left(
V_{i}-\beta^{\prime}X_{i,1}\right)  ^{2}k_{h_{n}}^{+}(Z_{i}),\label{betahat+}%
\\
\hat{\beta}_{V}^{-}(h_{n}) &  =\arg\min_{\beta}\sum_{i=1}^{n}\left(
V_{i}-\beta^{\prime}X_{i,1}\right)  ^{2}k_{h_{n}}^{-}(Z_{i}),\label{betahat-}%
\end{align}
where $X_{i,1}=r_{1}(Z_{i})\otimes\tilde{W}_{i},$ $\tilde{W}_{i}%
=(W_{1i},...,W_{mi})^{\prime}$, $W_{ij}=1(W_{i}=w_{j})$,  $r_{1}(z)=(1,z)^{\prime},$ $k_{h_{n}}%
^{+}(z)=h_{n}^{-1}k(z/h_{n})1(z\geq0),$ $k_{h_{n}}^{-}(z)=h_{n}^{-1}%
k(z/h_{n})1(z<0),$ and $e_{0}$ is the conformable $m\times2m$ matrix
$e_{0}=[I_{m},0]$ with $I_{m}$ the identity matrix of order $m$. Of course,
the same estimator can be obtained by running separate local linear
regressions over subsamples defined by covariates. The benefit of our approach
in (\ref{betahat+}) and (\ref{betahat-}) is that simultaneous inference on
$\beta(w)$ for different $w^{\prime}s$ can be constructed, based on the
estimator $\hat{\beta}(w_{j})=\hat{\delta}_{Y}(w_{j})/\hat{\delta}_{X}%
(w_{j}),$ for $j=1,...,m,$ where $\hat{\delta}_{V}(w_{j})$ is the $j$-th
component of $\hat{\delta}_{V}.$ Each of $\hat{\beta}(w_{j})$ is a standard
local linear RDD estimator. However, we are particularly interested in the
common situation where joint inference on $(\beta(w_{1}),...,\beta
(w_{m}))^{\prime}$ is not reliable because $\delta_{X}(w_{j})$ for some $j$ is
close to zero or sample sizes for some class $j$ are not large enough for
inference to be reliable with nonparametric methods. 

A Delta Method argument leads to the expansion%
\begin{equation}
\hat{\beta}_{\omega}-\beta_{\omega}=\frac{\hat{\tau}_{Y}-\tau_{Y}}{\tau_{X}%
}-\frac{\tau_{Y}\left(  \hat{\tau}_{X}-\tau_{X}\right)  }{\tau_{X}^{2}}%
+\hat{\mathcal{R}}\label{expansion}%
\end{equation}
where the remainder term is given by
\begin{equation}
\hat{\mathcal{R}}=\frac{\beta_{\omega}\left(  \hat{\tau}_{X}-\tau_{X}\right)  ^{2}}%
{\hat{\tau}_{X}\tau_{X}}-\frac{\left(  \hat{\tau}_{Y}-\tau_{Y}\right)  \left(
\hat{\tau}_{X}-\tau_{X}\right)  }{\tau_{X}\hat{\tau}_{X}}.\label{R}%
\end{equation}
This delta method expansion is known to be sensitive to a small
\textquotedblleft first-stage\textquotedblright\ $\tau_{X},$ which motivates
our robust estimand. Indeed, a small $\tau_{X}$ makes the variance of the
leading term in (\ref{expansion}) large, and introduces also a nonlinearity
bias through the remainder term $\hat{\mathcal{R}}.$

In our setting where $\hat{\tau}_{V}$ involves two nonparametric estimators
($\hat{c}$ and $\hat{\delta}_{V}$), (\ref{expansion}) is not a linearization
in terms of the more primitive objects $\hat{\mu}_{V}^{+},$ $\hat{\mu}_{V}%
^{-}$ and $\hat{\delta}_{V}.$ Thus, a further linearization is required, which
in turn may lead to further problems regarding close-to-zero denominators. 

In this section, we establish the asymptotic theory for the proposed estimator
$\hat{\beta}_{\omega}.$ We introduce some further notation and assumptions. Let
$\varepsilon_{V_{i}}=V_{i}-\E[V_{i}|Z_{i},W_{i}]$ denote the regression errors
for $V=Y_i$ and $V=X_i.$ In the proofs we use the generic notation, for a generic
$V_{i},$ $j,g=1,...,m,$
\[
\mu_{jg}^{V}(z)=\E[W_{ij}W_{ig}V_{i}|Z_{i}=z]\text{ and }\sigma_{jg}%
^{2,V}(z)=\E[W_{ij}W_{ig}V_{i}^{2}|Z_{i}=z]
\]
When $V$ is identically $1,$ we drop the reference to $V$ above.

We set $\mathcal{Y}_{n}=(Y_{1},...,Y_{n})^{\prime},$ $\mathcal{Z}_{n}%
=(Z_{1},...,Z_{n})^{\prime},$ $\mathcal{W}_{n}=(\tilde{W}_{1}^{\prime
},...,\tilde{W}_{n}^{\prime}),$ $\mathcal{V}_{n}=(V_{1}^{\prime}%
,...,V_{n}^{\prime})^{\prime},$ and $\mathcal{S}_{n}=\{\mathcal{Z}%
_{n},\mathcal{W}_{n}\}.$ Define the vector $\varepsilon_{V}=(\varepsilon
_{V_{1}},...,\varepsilon_{V_{n}})^{\prime}$ and $\Sigma_{UV}=\E[\varepsilon
_{U}\varepsilon_{V}^{\prime}|\mathcal{S}_{n}]=diag(\sigma_{UV}^{2}%
(Z_{1},\tilde{W}_{1}),\dots$ $\dots,\sigma_{UV}^{2}(Z_{n},\tilde{W}_{n})),$ with

\begin{align*}
\sigma_{UV}^{2}(z,\tilde{w})  &  =Cov[U_{i},V_{i}|Z_{i}=z,\tilde{W}_{i}%
=\tilde{w}],\\
\sigma_{UV}^{2}(z)  &  =\E[\varepsilon_{U_{i}}\varepsilon_{V_{i}}%
|Z_{i}=z],\text{ and}\\
\psi_{UV}(z)  &  =\E[\tilde{W}_{i}\tilde{W}_{i}^{\prime}\varepsilon_{U_{i}%
}\varepsilon_{V_{i}}|Z_{i}=z].
\end{align*}
With some abuse of notation, we denote $\mu_{W}(z)=\E[\tilde{W}_{i}|Z_{i}=z]$
and $\sigma_{W}^{2}(z)=\E[\tilde{W}_{i}\tilde{W}_{i}^{\prime}|Z_{i}=z].$

Consider the local right-hand side approximation%
\begin{align*}
E[V_{i}|Z_{i},W_{i}]  &  =\mu_{1V}(Z_{i})W_{1i}+\cdots++\mu_{mV}(Z_{i}%
)W_{mi},\\
&  =\tilde{W}_{i}^{\prime}\mu_{V}(Z_{i}),\\
&  \approx\tilde{W}_{i}^{\prime}\left(  \mu_{V}^{+}(0)+\mu_{V}^{+(1)}%
(0)Z_{i}+\cdots+\frac{\mu_{V}^{+(p)}(0)}{p!}Z_{i}^{p}\right)  ,\\
&  =\left(  r_{p}(Z_{i})\otimes\tilde{W}_{i}\right)  ^{\prime}\beta_{V,p}^{+},
\end{align*}
where $\otimes$ denotes Kronecker product, $\beta_{V,p}^{+}=(\mu_{V}^{+}%
{}^{\prime},\mu_{V}^{+(1)\prime},...,\mu_{V}^{+(p)\prime}/p!),$ and
\[
\mu_{V}^{+(s)}\equiv\mu_{V}^{+(s)}(0)=\lim_{z\downarrow0}\frac{\partial^{s}%
\mu_{V}(z)}{\partial z^{s}}.
\]
Implicit in the notation above is that $\mu_{V}(z)=(\mu_{1V}(z),...,\mu
_{mV}(z)),$ where%
\[
\mu_{jV}(z)=\E[V_{i}|Z_{i}=z,W_{i}=w_{j}].
\]
We use the analogous notation for left hand side approximations and
derivatives (with $-$ replacing $+)$.

We investigate the asymptotic properties of $\hat{\beta}_{\omega}$ under the
following assumptions, which parallel those of \cite{HTV2001}:\bigskip

\newpage 
\begin{assumption}
\label{AssumptionRRD} Suppose that with $z$ in a neighborhood of zero:

\begin{enumerate}
\item The sample $\{\chi_{i}\}_{i=1}^{n}$ is an iid sample, where $\chi
_{i}=(Y_{i},X_{i}^{\prime},W_{i}^{\prime},Z_{i})^{\prime}.$

\item (i) the density of $Z_i,$ $f(z),$ is continuous, bounded and bounded away
from zero; (ii) $\E[Y_{i}^{4}|Z_{i}=z,W_{i}=w]$ and $\sigma_{UV}^{2}%
(z,\tilde{w})$ are bounded. The matrices $\Gamma_{+,p}$ and $\Gamma_{-,p}$,
defined below, are positive definite, and $\tau_{X}\neq0.$

\item The kernel $k$ is continuous, symmetric and nonnegative-valued with
compact support.

\item The functions $\mu_{jg}(z),$ $\sigma_{jg}^{2}(z),$ $\sigma_{UV}^{2}(z),$
$\mu_{jg}^{\varepsilon_{Y}\varepsilon_{X}}(z),$ $\sigma_{jg}^{2,\varepsilon
_{Y}\varepsilon_{X}}(z)$ are uniformly bounded$,$ with well-defined and finite
left and right limits to $z=0,$ for $j,g=1,...,m.$

\item The bandwidth satisfies $nh_{n}^{2p+5}\rightarrow0$ and $nh_{n}%
\rightarrow\infty.$

\item For $V_{i}=Y_{i}$ and $V_{i}=X_{i}$, for $z>0$ or $z<0$, and all $j:$
(i) $\E[V_{i}|Z_{i}=z,W_{i}=w_{j}]$ is $d$ times continuously differentiable,
$d\geq p+2$; (ii) $Var[V_{i}|Z_{i}=z,W_{i}=w_{j}]$ are continuous in $z$ and
bounded away from zero$.$

\item $c=g(\mu_{V}^{+},\mu_{V}^{-},\pi),$ with $g$ continuously differentiable
in its components at the true values $(\mu_{V}^{+},\mu_{V}^{-},\pi).$ When
$V=W_{ij},$ for $z>0$ or $z<0,$ $\E[W_{ij}|Z_{i}=z]$ is $d$ times continuously
differentiable, $d\geq p+2,$ and bounded away from zero and one.

\item The sequences $h_n$ and $b_n$ satisfy $n\min\{h_{n}^{5},b_{n}^{5}\}\max\{h_{n}^{2},b_{n}^{2}%
\}\rightarrow0$ and $n\min\{h_{n},b_{n}\}\rightarrow\infty.$\medskip
\end{enumerate}
\end{assumption}

To introduce the estimators, define for a generic random variable $V,$ degree
$p,$ $0\leq v\leq p,$ and bandwidth $h_{n}$ the quantities%
\begin{align*}
\hat{\delta}_{V,p}(h_{n})  &  =\hat{\mu}_{V,p}^{+}(h_{n})-\hat{\mu}_{V,p}%
^{-}(h_{n}),\\
\hat{\mu}_{V,p}^{+(v)}(h_{n})  &  =v!e_{v}^{\prime}\hat{\beta}_{V,p}^{+}%
(h_{n}),
\\
\hat{\mu}_{V,p}^{-(v)}(h_{n})  &  =v!e_{v}^{\prime}\hat{\beta}_{V,p}^{-}%
(h_{n}),
\end{align*}
\begin{align*}
\hat{\beta}_{V,p}^{+}(h_{n})  &  =\arg\min_{\beta}\sum_{i=1}^{n}\left(
V_{i}-\beta^{\prime}X_{i,p}\right)  ^{2}k_{h_{n}}^{+}(Z_{i}),\\
\hat{\beta}_{V,p}^{-}(h_{n})  &  =\arg\min_{\beta}\sum_{i=1}^{n}\left(
V_{i}-\beta^{\prime}X_{i,p}\right)  ^{2}k_{h_{n}}^{-}(Z_{i}),
\end{align*}
where $\hat{\mu}_{V,p}^{+}(h_{n})=\hat{\mu}_{V,p}^{+(0)}(h_{n}),$ $\hat{\mu
}_{V,p}^{-}(h_{n})=\hat{\mu}_{V,p}^{-(0)}(h_{n}),$ $X_{i,p}=r_{p}%
(Z_{i})\otimes\tilde{W}_{i},$ $r_{p}(z)=(1,z,...,z^{p})^{\prime},$ $k_{h_{n}%
}^{+}(z)=h_{n}^{-1}k(z/h_{n})1(z\geq0),$ $k_{h_{n}}^{-}(z)=h_{n}^{-1}%
k(z/h_{n})1(z<0),$ and $e_{v}$ is a conformable $m\times m(1+p)$ matrix, that
selects the $m$ elements corresponding to the $v$-th derivative. For example,
$e_{0}=[I_{m},0...,0]$.

Define the matrices%
\begin{align*}
X_{p}(h)  &  =[r_{p}(Z_{1}/h)\otimes\tilde{W}_{1},...,r_{p}(Z_{n}%
/h)\otimes\tilde{W}_{n}],\\
S_{p}(h)  &  =[(Z_{1}/h)^{p}\otimes\tilde{W}_{1},...,(Z_{n}/h)^{p}%
\otimes\tilde{W}_{n}],\\
K^{+}(h)  &  =[k_{h}^{+}(Z_{1}),...,k_{h}^{+}(Z_{n})],\\
K^{-}(h)  &  =[k_{h}^{-}(Z_{1}),...,k_{h}^{-}(Z_{n})],\\
\Gamma_{+,p}(h)  &  =X_{p}(h)^{\prime}K^{+}(h)X_{p}(h)/n,\\
\Gamma_{-,p}(h)  &  =X_{p}(h)^{\prime}K^{-}(h)X_{p}(h)/n,\\
\vartheta_{p,q}^{+}(h)  &  =X_{p}(h)^{\prime}K^{+}(h)S_{q}(h)/n,\\
\vartheta_{p,q}^{-}(h)  &  =X_{p}(h)^{\prime}K^{-}(h)S_{q}(h)/n,\\
\Psi_{UV,p,q}^{+}(h,b)  &  =X_{p}(h)^{\prime}K^{+}(h)\Sigma_{UV}K^{+}%
(b)X_{q}(b)/n,\\
\Psi_{UV,p,q}^{-}(h,b)  &  =X_{p}(h)^{\prime}K^{-}(h)\Sigma_{UV}K^{-}%
(b)X_{q}(b)/n.
\end{align*}
When $U=V,$ $p=q$ or $h=b$ in the above expressions one of the terms is drop,
so we denote, for example, $\Sigma_{U}=\Sigma_{UU}$ and $\Psi_{U,p}%
^{+}(h)=\Psi_{UU,p,p}^{+}(h,h).$ Letting
\[
H_{p}(h)=diag(1,...,1,h^{-1},...,h^{-1},...,h^{-p},...,h^{-p}),
\]
where each element is repeated $m$ times, it follows that%
\begin{align*}
\hat{\beta}_{V,p}^{+}  &  =H_{p}(h_{n})\Gamma_{+,p}^{-1}(h_{n})X_{p}^{\prime
}(h_{n})K^{+}(h_{n})\mathcal{V}_{n}/n,\\
\hat{\beta}_{V,p}^{-}  &  =H_{p}(h_{n})\Gamma_{-,p}^{-1}(h_{n})X_{p}^{\prime
}(h_{n})K^{-}(h_{n})\mathcal{V}_{n}/n.
\end{align*}
Define%
\begin{align*}
\mu_{W}^{+}  &  =\lim_{z\downarrow0}\mu_{W}(z)\qquad\mu_{W}^{-}=\lim
_{z\uparrow0}\mu_{W}(z),\\
\sigma_{W}^{2+}  &  =\lim_{z\downarrow0}\sigma_{W}^{2}(z)\qquad\sigma_{W}%
^{2-}=\lim_{z\uparrow0}\sigma_{W}^{2}(z),
\end{align*}
and%
\begin{align*}
\Gamma_{p}^{+}  &  =\int_{0}^{\infty}r_{p}(u)r_{p}^{\prime}(u)k(u)du,\qquad
\Gamma_{p}^{-}=\int_{0}^{\infty}r_{p}(-u)r_{p}^{\prime}(-u)k(u)du,\\
\Upsilon_{p}^{+}  &  =\int_{0}^{\infty}r_{p}(u)r_{p}^{\prime}(u)k^{2}%
(u)du,\qquad \Upsilon_{p}^{-}=\int_{0}^{\infty}r_{p}(-u)r_{p}^{\prime}%
(-u)k^{2}(u)du,\\
\theta_{p,q}^{+}  &  =\int_{0}^{\infty}r_{p}(u)u^{q}k(u)du,\qquad\theta
_{p,q}^{-}=(-1)^{q}\int_{0}^{\infty}r_{p}(-u)u^{q}k(u)du.
\end{align*}

\begin{lemma}
\label{Delta+} Under Assumption \ref{AssumptionRRD}(1)-(6),
\[
\Gamma_{+,p}(h_{n})\rightarrow_{p}\Gamma_{+,p},
\]
where
\[
\Gamma_{+,p}=f(0)\left(  \Gamma_{p}^{+}\otimes\sigma_{W}^{2+}\right)  .
\]

\begin{proof}
Using that $X_{i,p}X_{i,p}^{\prime}=r_{p}(Z_{i}/h_{n})r_{p}^{\prime}%
(Z_{i}/h_{n})\otimes\tilde{W}_{i}\tilde{W}_{i}^{\prime}$ and by the change of
variables $u=Z/h_{n},$%
\begin{align*}
\E[\Gamma_{+,p}(h_{n})]  &  =\E\left[  \frac{1}{n}\sum_{i=1}^{n}X_{i,p}%
X_{i,p}^{\prime}k_{ih_{n}}^{+}\right] \\
&  =\E\left[  \left(  r_{p}(Z_{i}/h_{n})r_{p}^{\prime}(Z_{i}/h_{n}%
)\otimes\sigma_{W}^{2}(Z_{i})\right)  k_{ih_{n}}^{+}\right] \\
&  =\int_{0}^{\infty}\left(  r_{p}(u)r_{p}^{\prime}(u)\otimes\sigma_{W}%
^{2}(uh_{n})\right)  k(u)f(uh_{n})du\\
&  \equiv\tilde{\Gamma}_{+,p}(h_{n}).
\end{align*}
By Assumptions \ref{AssumptionRRD}(2) and \ref{AssumptionRRD}(4) and Dominated
Convergence theorem%
\[
\tilde{\Gamma}_{+,p}(h_{n})=\Gamma_{+,p}+o(1),
\]
where%
\[
\Gamma_{+,p}=f(0)\int_{0}^{\infty}\left(  r_{p}(u)r_{p}^{\prime}%
(u)\otimes\sigma_{W}^{2+}\right)  k(u)du.
\]
Let
\[
\tau_{ljg}^{+}=\frac{1}{n}\sum_{i=1}^{n}\left(  \frac{Z_{i}}{h_{n}}\right)
^{l}W_{ij}W_{ig}k_{ih_{n}}^{+},\qquad l=0,1,...,2p,j,g=1,...,m.
\]
Then,
\begin{align*}
Var(\tau_{ljg}^{+})  &  \leq n^{-1}\E\left[  \left(  \frac{Z_{i}}{h_{n}%
}\right)  ^{2l}W_{ij}^{2}W_{ig}^{2}k_{ih_{n}}^{+2}\right] \\
&  =\left(  nh_{n}\right)  ^{-1}\int_{0}^{\infty}u^{2l}k^{2}(u)\sigma_{jg}%
^{2}(uh_{n})f(uh_{n})du\\
&  =o(1),
\end{align*}
again by the Dominated Convergence theorem. Conclude by applying 
Markov's inequality.
\end{proof}

\end{lemma}

\begin{lemma}
\label{Delta-} Under Assumption \ref{AssumptionRRD}(1)-(6),
\[
\Gamma_{-,p}(h_{n})\rightarrow_{p}\Gamma_{-,p},
\]
where
\[
\Gamma_{-,p}=f(0)\left(  \Gamma_{p}^{-}\otimes\sigma_{W}^{2-}\right)  .
\]

\begin{proof}
As with $\Gamma_{+,p}(h_{n}),$ by the change of variables $u=Z/h_{n},$%
\begin{align*}
\E[\Gamma_{-,p}(h_{n})]  &  =\E\left[  \frac{1}{n}\sum_{i=1}^{n}X_{i,p}%
X_{i,p}^{\prime}k_{ih_{n}}^{-}\right] \\
&  =\E\left[  \left(  r_{p}(Z_{i}/h_{n})r_{p}^{\prime}(Z_{i}/h_{n}%
)\otimes\sigma_{W}^{2}(Z_{i})\right)  k_{ih_{n}}^{-}\right] \\
&  =\int_{-\infty}^{0}\left(  r_{p}(u)r_{p}^{\prime}(u)\otimes\sigma_{W}%
^{2}(uh_{n})\right)  k(u)f(uh_{n})du\\
&  =\int_{0}^{\infty}\left(  r_{p}(-u)r_{p}^{\prime}(-u)\otimes\sigma_{W}%
^{2}(-uh_{n})\right)  k(-u)f(-uh_{n})du\\
&  \equiv\tilde{\Gamma}_{-,p}(h_{n}).
\end{align*}
The rest of the proof follows as for $\Gamma_{+,p}(h_{n})$.
\end{proof}
\end{lemma}

\begin{lemma}
\label{Theta+} Under Assumption \ref{AssumptionRRD}(1)-(6),
\[
\vartheta_{p,q}^{+}(h_{n})\rightarrow_{p}\vartheta_{p,q}^{+},
\]
where
\[
\vartheta_{p,q}^{+}=f(0)\left(  \theta_{p,q}^{+}\otimes\mu_{W}^{+}\right)  .
\]

\begin{proof}
By the change of variables $u=Z/h_{n},$%
\begin{align*}
\E[\vartheta_{p,q}^{+}(h_{n})]  &  =\E\left[  \frac{1}{n}\sum_{i=1}^{n}%
X_{i,p}(Z_{i}/h_{n})^{q}k_{ih_{n}}^{+}\right] \\
&  =\E\left[  \left(  r_{p}(Z_{i}/h_{n})\otimes\mu_{W}(Z_{i})\right)
(Z_{i}/h_{n})^{q}k_{ih_{n}}^{+}\right] \\
&  =\int_{0}^{\infty}\left(  r_{p}(u)\otimes\mu_{W}(uh_{n})\right)
u^{q}k(u)f(uh_{n})du\\
&  \equiv\tilde{\vartheta}_{p,q}^{+}(h_{n}).
\end{align*}
By Assumptions \ref{AssumptionRRD}(2) and \ref{AssumptionRRD}(4) and Dominated
Convergence theorem%
\[
\tilde{\vartheta}_{p,q}^{+}(h_{n})=\vartheta_{p,q}^{+}+o(1),
\]
where%
\[
\vartheta_{p,q}^{+}=f(0)\int_{0}^{\infty}\left(  r_{p}(u)\otimes\mu_{W}%
^{+}\right)  u^{q}k(u)du.
\]
The rest of the proof follows the same arguments as for $\Gamma_{+,p}(h_{n}).$
\end{proof}
\end{lemma}

\begin{lemma}
\label{Theta-} Under Assumption \ref{AssumptionRRD}(1)-(6),
\[
\vartheta_{p,q}^{-}(h_{n})\rightarrow_{p}\vartheta_{p,q}^{-},
\]
where
\[
\vartheta_{p,q}^{-}=f(0)\left(  \theta_{p,q}^{-}\otimes\mu_{W}^{-}\right)  .
\]

\begin{proof}
By the change of variables $u=Z/h_{n},$%
\begin{align*}
\E[\vartheta_{p,q}^{-}(h_{n})]  &  =\E\left[  \frac{1}{n}\sum_{i=1}^{n}%
X_{i,p}(Z_{i}/h_{n})^{q}k_{ih_{n}}^{-}\right] \\
&  =\E\left[  \left(  r_{p}(Z_{i}/h_{n})\otimes\mu_{W}(Z_{i})\right)
(Z_{i}/h_{n})^{q}k_{ih_{n}}^{-}\right] \\
&  =\int_{0}^{\infty}\left(  r_{p}(-u)\otimes\mu_{W}(-uh_{n})\right)
(-u)^{q}k(-u)f(-uh_{n})du\\
&  \equiv\tilde{\vartheta}_{p,q}^{-}(h_{n}).
\end{align*}
By Assumptions \ref{AssumptionRRD}(2) and \ref{AssumptionRRD}(4) and Dominated
Convergence theorem%
\[
\tilde{\vartheta}_{p,q}^{-}(h_{n})=\vartheta_{p,q}^{-}+o(1),
\]
where%
\[
\vartheta_{p,q}^{-}=f(0)(-1)^{q}\int_{0}^{\infty}\left(  r_{p}(-u)\otimes
\mu_{W}^{-}\right)  u^{q}k(u)du.
\]
The rest of the proof follows the same arguments as for $\Gamma_{+,p}(h_{n}).$
\end{proof}
\end{lemma}

\begin{lemma}
\label{Omega+} Under Assumption \ref{AssumptionRRD}(1)-(6), for $U$ and $V$ equal
$Y_i$ or $X_i$%
\[
h_{n}\Psi_{UV,p}^{+}(h_{n})\rightarrow_{p}\Psi_{UV,p}^{+},
\]
where
\[
\Psi_{UV,p}^{+}=f(0)\int_{0}^{\infty}\left(  r_{p}(u)r_{p}^{\prime}%
(u)\sigma_{UV}^{2+}\otimes\psi_{UV}^{+}\right)  k^{2}(u)du.
\]

\begin{proof}
Note
\[
\E[\sigma_{UV}^{2}(Z_{i},\tilde{W}_{i})|Z_{i}=z]=\sigma_{UV}^{2}(z)
\]
and
\[
\E[\tilde{W}_{i}\tilde{W}_{i}^{\prime}\sigma_{UV}^{2}(Z_{i},\tilde{W}%
_{i})|Z_{i}=z]=\psi_{UV}(z).
\]
By the change of variables $u=Z/h_{n},$ it follows that%
\begin{align*}
\E[h_{n}\Psi_{UV,p}^{+}(h_{n})]  &  =\E[\frac{h_{n}}{n}\sum_{i=1}^{n}%
X_{i,p}X_{i,p}^{\prime}\sigma_{UV}^{2}(Z_{i},\tilde{W}_{i})k_{ih_{n}}^{2+}]\\
&  =\int_{0}^{\infty}\left(  r_{p}(u)r_{p}^{\prime}(u)\sigma_{UV}^{2}%
(uh_{n})\otimes\psi_{UV}(uh_{n})\right)  k^{2}(u)f(uh_{n})du\\
&  \equiv\tilde{\Psi}_{UV,p}^{+}(h_{n}).
\end{align*}
By Assumption \ref{AssumptionRRD}(4) and Dominated Convergence theorem%
\[
\tilde{\Psi}_{UV,p}^{+}(h_{n})=\Psi_{UV,p}^{+}+o(1),
\]
Also, since $\sigma_{UV}^{2}(Z_{i},\tilde{W}_{i})$ is bounded,
\begin{align*}
h_{n}^{2}\E\left[  \left\vert \Psi_{UV,p}^{+}(h_{n})-\E[\Psi_{UV,p}^{+}%
(h_{n})]\right\vert ^{2}\right]   &  \leq Cn^{-1}h_{n}^{-1}\int_{0}^{\infty
}\left\vert r_{p}(u)\right\vert ^{4}k^{4}(u)f(uh_{n})du\\
&  =O\left(  n^{-1}h_{n}^{-1}\right)  .
\end{align*}
\end{proof}
\end{lemma}

\begin{lemma}
\label{Omega-} Under Assumption \ref{AssumptionRRD}(1)-(6),
\[
h_{n}\Psi_{UV,p}^{-}(h_{n})\rightarrow_{p}\Psi_{UV,p}^{-},
\]
where
\[
\Psi_{UV,p}^{-}=f(0)\int_{0}^{\infty}\left(  r_{p}(u)r_{p}^{\prime}%
(u)\sigma_{UV}^{2-}\otimes\psi_{UV,p}^{-}\right)  k^{2}(u)du.
\]

\begin{proof}
The proof follows the same arguments as the previous one.
\end{proof}
\end{lemma}

\begin{lemma}
\label{Omega++} Under Assumption \ref{AssumptionRRD}(1)-(7), for $U$ and $V$
equal $Y_i$ or $X_i$ and $m_{n}=\min\{h_{n},b_{n}\}$ with $m_{n}\rightarrow0$ and
$nm_{n}\rightarrow\infty,$%
\[
\frac{h_{n}b_{n}}{m_{n}}\Psi_{UV,p,q}^{+}(h_{n},b_{n})=\tilde{\Psi}_{UV,p}%
^{+}(h_{n},b_{n})+o_{P}(1),
\]
where
\[
\tilde{\Psi}_{UV,p,q}^{+}(h_{n},b_{n})=\int_{0}^{\infty}\left(  r_{p}\left(
\frac{m_{n}u}{h_{n}}\right)  r_{q}^{\prime}\left(  \frac{m_{n}u}{b_{n}%
}\right)  \sigma_{UV}^{2}(um_{n})\otimes\psi_{UV}(um_{n})\right)  k\left(
\frac{m_{n}u}{h_{n}}\right)  k\left(  \frac{m_{n}u}{b_{n}}\right)
f(um_{n})du.
\]

\begin{proof}
By the change of variables $u=Z/h_{n},$ it follows that%
\begin{align*}
&  \E\left[  \frac{h_{n}b_{n}}{m_{n}}\Psi_{UV,p,q}^{+}(h_{n},b_{n})\right] \\
&  =\E\left[  \frac{h_{n}b_{n}}{m_{n}}\frac{1}{n}\sum_{i=1}^{n}X_{i,p}%
X_{i,q}^{\prime}\sigma_{UV}^{2}(Z_{i},\tilde{W}_{i})k_{ih_{n}}^{+}k_{ib_{n}%
}^{+}\right] \\
&  =\int_{0}^{\infty}\left(  r_{p}\left(  \frac{m_{n}u}{h_{n}}\right)
r_{q}^{\prime}\left(  \frac{m_{n}u}{b_{n}}\right)  \sigma_{UV}^{2}%
(um_{n})\otimes\psi_{UV}(um_{n})\right)  k\left(  \frac{m_{n}u}{h_{n}}\right)
k\left(  \frac{m_{n}u}{b_{n}}\right)  f(um_{n})du\\
&  \equiv\tilde{\Psi}_{UV,p,q}^{+}(h_{n},b_{n}).
\end{align*}
Also, since $\sigma_{UV}^{2}(Z_{i},\tilde{W}_{i})$ is bounded,
\begin{align*}
& \E\left[  \left\vert \frac{h_{n}b_{n}}{m_{n}}\Psi_{UV,p}^{+}(h_{n}%
)-\E[\frac{h_{n}b_{n}}{m_{n}}\Psi_{UV,p}^{+}(h_{n})]\right\vert ^{2}\right] \\
&  \leq Cn^{-1}m_{n}^{-1}\int_{0}^{\infty}k^{2}\left(  \frac{m_{n}u}{h_{n}%
}\right)  k^{2}\left(  \frac{m_{n}u}{b_{n}}\right)  \left\vert r_{p}\left(
\frac{m_{n}u}{h_{n}}\right)  \right\vert ^{2}\left\vert r_{q}\left(
\frac{m_{n}u}{b_{n}}\right)  \right\vert ^{2}f(uh_{n})du\\
&  =O\left(  n^{-1}m_{n}^{-1}\right)  .
\end{align*}

\end{proof}
\end{lemma}

\begin{lemma}
\label{Omega--} Under Assumption \ref{AssumptionRRD}(1)-(7), for $U$ and $V$
equal $Y_i$ or $X_i$ and $m_{n}=\min\{h_{n},b_{n}\}$ with $m_{n}\rightarrow0$ and
$nm_{n}\rightarrow\infty,$%
\[
\frac{h_{n}b_{n}}{m_{n}}\Psi_{UV,p,q}^{-}(h_{n},b_{n})=\tilde{\Psi}_{UV,p}%
^{-}(h_{n},b_{n})+o_{P}(1),
\]
where
\[
\tilde{\Psi}_{UV,p,q}^{-}(h_{n},b_{n})=\int_{-\infty}^{0}\left(  r_{p}\left(
\frac{m_{n}u}{h_{n}}\right)  r_{q}^{\prime}\left(  \frac{m_{n}u}{b_{n}%
}\right)  \sigma_{UV}^{2}(um_{n})\otimes\psi_{UV}(um_{n})\right)  k\left(
\frac{m_{n}u}{h_{n}}\right)  k\left(  \frac{m_{n}u}{b_{n}}\right)
f(um_{n})du.
\]

\begin{proof}
The proof is analogous to the previous one.
\end{proof}
\end{lemma}

\noindent The following lemma gives the asymptotic bias for the local
polynomial estimator of $\mu_{V}^{+(s)}.$ Define the bias terms%
\begin{align}
\mathcal{B}_{s,l,m}^{+}(h_{n})  &  =s!e_{s}^{\prime}\Gamma_{+,l}^{-1}%
(h_{n})\vartheta_{l,m}^{+}(h_{n})\label{biasgen+}\\
&  =s!e_{s}^{\prime}\Gamma_{+,l}^{-1}\vartheta_{l,m}^{+}+o_{P}(1)\nonumber
\end{align}
and
\begin{align}
\mathcal{B}_{s,l,m}^{-}(h_{n})  &  =s!e_{s}^{\prime}\Gamma_{-,l}^{-1}%
(h_{n})\vartheta_{l,m}^{-}(h_{n})\label{biasgen-}\\
&  =s!e_{s}^{\prime}\Gamma_{-,l}^{-1}\vartheta_{l,m}^{-}+o_{P}(1).\nonumber
\end{align}

\begin{lemma}
\label{Bias}Under Assumption \ref{AssumptionRRD}(1)-(6), for $V=Y_i$ or $X_i,$ and
$d\geq l+2$%
\begin{align*}
E[\hat{\mu}_{V,l}^{+(s)}(h_{n})|\mathcal{S}_{n}]  &  =s!e_{s}^{\prime}%
\beta_{V,l}^{+}+h_{n}^{1+l-s}\frac{\mu_{V}^{+(l+1)}}{(l+1)!}\mathcal{B}%
_{s,l,l+1}^{+}(h_{n})\\
&  +h_{n}^{2+l-s}\frac{\mu_{V}^{+(l+2)}}{(l+2)!}\mathcal{B}_{s,l,l+2}%
^{+}(h_{n})+o_{P}(h_{n}^{2+l-s}),
\end{align*}
and%
\begin{align*}
E[\hat{\mu}_{V,l}^{-(s)}(h_{n})|\mathcal{S}_{n}]  &  =s!e_{s}^{\prime}%
\beta_{V,l}^{-}+h_{n}^{1+l-s}\frac{\mu_{V}^{-(l+1)}}{(l+1)!}\mathcal{B}%
_{s,l,l+1}^{-}(h_{n})\\
&  +h_{n}^{2+l-s}\frac{\mu_{V}^{-(l+2)}}{(l+2)!}\mathcal{B}_{s,l,l+2}%
^{-}(h_{n})+o_{P}(h_{n}^{2+l-s}).
\end{align*}
Hence,%
\begin{align*}
E[\hat{\delta}_{V,l}(h_{n})|\mathcal{S}_{n}]  &  =\delta_{V}+h_{n}%
^{l+1}B_{V,l,l+1}(h_{n})\\
&  +h_{n}^{2+l}B_{V,l,l+2}(h_{n})+o_{P}(h_{n}^{2+l}),
\end{align*}
where
\begin{equation}
B_{V,l,m}=\frac{\mu_{V}^{+(m)}}{m!}\mathcal{B}_{l,m}^{+}(h_{n})-\frac{\mu
_{V}^{-(m)}}{m!}\mathcal{B}_{l,m}^{-}(h_{n}), \label{biasdelta}%
\end{equation}
with $\mathcal{B}_{l,m}^{+}(h_{n})=\mathcal{B}_{0,l,m}^{+}(h_{n}%
)=e_{0}^{\prime}\Gamma_{+,l}^{-1}(h_{n})\vartheta_{l,m}^{+}(h_{n})$ and
$\mathcal{B}_{l,m}^{-}(h_{n})=\mathcal{B}_{0,l,m}^{-}(h_{n})=e_{0}^{\prime
}\Gamma_{-,l}^{-1}(h_{n})\vartheta_{l,m}^{-}(h_{n}).$

\bigskip 
\begin{proof}
A Taylor series expansion gives%
\begin{align*}
\E[\hat{\beta}_{V,l}^{+}(h_{n})|\mathcal{S}_{n}]  &  =\beta_{V,l}^{+}%
+h_{n}^{1+l}H_{l}(h_{n})\Gamma_{+,l}^{-1}(h_{n})X_{l}^{\prime}(h_{n}%
)K^{+}(h_{n})S_{l+1}(h_{n})\frac{\mu_{V}^{+(l+1)}}{n(l+1)!}\\
&  +h_{n}^{2+l}H_{l}(h_{n})\Gamma_{+,l}^{-1}(h_{n})X_{l}^{\prime}(h_{n}%
)K^{+}(h_{n})S_{l+2}(h_{n})\frac{\mu_{V}^{+(l+2)}}{n(l+2)!}+H_{l}(h_{n}%
)o_{P}(h_{n}^{2+l}),\\
&  =\beta_{V,l}^{+}+h_{n}^{1+l}H_{l}(h_{n})\Gamma_{+,l}^{-1}(h_{n}%
)\vartheta_{l,l+1}^{+}(h_{n})\frac{\mu_{V}^{+(l+1)}}{(l+1)!}\\
&  +h_{n}^{2+l}H_{l}(h_{n})\Gamma_{+,l}^{-1}(h_{n})\vartheta_{l,l+2}^{+}%
(h_{n})\frac{\mu_{V}^{+(l+2)}}{(l+2)!}+H_{l}(h_{n})o_{P}(h_{n}^{2+l}).
\end{align*}
The expression for $\E[\hat{\mu}_{V,l}^{+(s)}(h_{n})|\mathcal{S}_{n}]$ follows
from previous lemmas and $e_{s}^{\prime}H_{l}(h_{n})=h_{n}^{-s}I_{m}.$ The
same arguments apply to $\E[\hat{\beta}_{V,l}^{-}(h_{n})|\mathcal{S}_{n}].$
\end{proof}
\end{lemma}

\noindent The following lemma gives the asymptotic variance for the local
polynomial estimator of $\mu_{V}^{+(s)}.$

\begin{lemma}
\label{Variance}Under Assumption \ref{AssumptionRRD}(1)-(6), for $V=Y\ $or $X,$
and $d\geq l+2$%
\[
Var[\hat{\mu}_{V,l}^{+(s)}(h_{n})|\mathcal{S}_{n}]=\mathcal{V}_{V,l}%
^{+(s)}(h_{n}),
\]
with
\begin{align}
\mathcal{V}_{V,l}^{+(s)}(h_{n})  &  =\frac{1}{nh_{n}^{2s}}\left(  s!\right)
^{2}e_{s}^{\prime}\Gamma_{+,l}^{-1}(h_{n})\Psi_{V,l}^{+}(h_{n})\Gamma
_{+,l}^{-1}(h_{n})e_{s}\label{varder+}\\
&  =\frac{1}{nh_{n}^{1+2s}}\left(  s!\right)  ^{2}e_{s}^{\prime}\Gamma
_{+,l}^{-1}\Psi_{V,l}^{+}\Gamma_{+,l}^{-1}e_{s}\left[  1+o_{P}(1)\right]
,\nonumber
\end{align}
and%
\[
Var[\hat{\mu}_{V,l}^{-(s)}(h_{n})|\mathcal{S}_{n}]=\mathcal{V}_{V,l}%
^{-(s)}(h_{n}),
\]
with
\begin{align}
\mathcal{V}_{V,l}^{-(s)}(h_{n})  &  =\frac{1}{nh_{n}^{2s}}\left(  s!\right)
^{2}e_{s}^{\prime}\Gamma_{-,l}^{-1}(h_{n})\Psi_{V,l}^{-}(h_{n})\Gamma
_{-,l}^{-1}(h_{n})e_{s}\label{varder-}\\
&  =\frac{1}{nh_{n}^{1+2s}}\left(  s!\right)  ^{2}e_{s}^{\prime}\Gamma
_{-,l}^{-1}\Psi_{V,l}^{-}\Gamma_{-,l}^{-1}e_{s}\left[  1+o_{P}(1)\right]
.\nonumber
\end{align}
Furthermore,
\begin{equation}
Var[\hat{\delta}_{V,l}(h_{n})|\mathcal{S}_{n}]=\mathbf{V}_{V,l}(h_{n}%
)\equiv\mathcal{V}_{V,l}^{+}(h_{n})+\mathcal{V}_{V,l}^{-}(h_{n}),
\label{variancedelta}%
\end{equation}
with $\mathcal{V}_{V,l}^{+}(h_{n})\equiv\mathcal{V}_{V,l}^{+(0)}(h_{n})$ and
$\mathcal{V}_{V,l}^{-}(h_{n})\equiv\mathcal{V}_{V,l}^{-(0)}(h_{n}).$

\bigskip
\begin{proof}
By $e_{s}^{\prime}H_{l}(h_{n})=h_{n}^{-s}I_{m}$ and definition of $\Psi
_{V,l}^{+}(h_{n})$%
\[
Var[s!e_{s}^{\prime}\hat{\beta}_{V,p}^{+}(h_{n})|\mathcal{S}_{n}]=h_{n}%
^{-2s}\left(  s!\right)  ^{2}e_{s}^{\prime}\Gamma_{+,l}^{-1}(h_{n})\Psi
_{V,l}^{+}(h_{n})\Gamma_{+,l}^{-1}(h_{n})e_{s}.
\]
The result follows from previous lemmas. The same arguments apply to $\hat
{\mu}_{V,l}^{-(s)}(h_{n}).$
\end{proof}
\end{lemma}

\begin{lemma}
\label{AN}Under Assumption \ref{AssumptionRRD}(1)-(6), for $V=Y_i\ $or $X_i,$
$nh_{n}^{2p+5}\rightarrow0$ and $d\geq p+2$%
\begin{equation}
\left(  \mathcal{V}_{V,p}^{+(s)}(h_{n})\right)  ^{-1/2}\left(  \hat{\mu}%
_{V,p}^{+(s)}(h_{n})-\mu_{V}^{+(s)}-h_{n}^{1+p-s}\frac{\mu_{V+}^{(p+1)}%
}{(p+1)!}\mathcal{B}_{s,p,p+1}^{+}(h_{n})\right)  \rightarrow_{d}N\left(
0,I_{m}\right)  , \label{A1}%
\end{equation}%
\[
\left(  \mathcal{V}_{V,l}^{-(s)}(h_{n})\right)  ^{-1/2}\left(  \hat{\mu}%
_{V,p}^{-(s)}(h_{n})-\mu_{V}^{-(s)}-h_{n}^{1+p-s}\frac{\mu_{V-}^{(p+1)}%
}{(p+1)!}\mathcal{B}_{s,p,p+1}^{-}(h_{n})\right)  \rightarrow_{d}N\left(
0,I_{m}\right)  ,
\]
and
\[
\left(  \mathbf{V}_{V,p}(h_{n})\right)  ^{-1/2}\left(  \hat{\delta}%
_{V,p}(h_{n})-\delta_{V,p}(h_{n})-h_{n}^{1+p}B_{p,p+1}^{+}(h_{n})\right)
\rightarrow_{d}N\left(  0,I_{m}\right)  .
\]

\begin{proof}
Write the left hand side of (\ref{A1}) as%
\begin{align*}
\left(  \mathcal{V}_{V,p}^{+(s)}(h_{n})\right)  ^{-1/2}\left(  s!e_{s}%
^{\prime}\hat{\beta}_{V,p}^{+}(h_{n})-s!e_{s}^{\prime}\beta_{V}^{+}%
(h_{n})-h_{n}^{1+p-s}\frac{\mu_{V+}^{(p+1)}}{(p+1)!}\mathcal{B}_{s,p,p+1}%
^{+}(h_{n})\right)   &  =\xi_{1n}+\xi_{2n},\\
&  =\xi_{1n}+o_{P}(1),
\end{align*}
where%
\begin{align*}
\xi_{1n}  &  =\left(  \mathcal{V}_{V,p}^{+(s)}(h_{n})\right)  ^{-1/2}\left(
s!e_{s}^{\prime}\hat{\beta}_{V,p}^{+}(h_{n})-\E[s!e_{s}^{\prime}\hat{\beta
}_{V,p}^{+}(h_{n})|\mathcal{S}_{n}]\right) \\
&  =\left(  \mathcal{V}_{V,p}^{+(s)}(h_{n})\right)  ^{-1/2}s!e_{s}^{\prime
}H_{p}(h_{n})\Gamma_{+,p}^{-1}(h_{n})X_{p}^{\prime}(h_{n})K^{+}(h_{n}%
)\varepsilon_{V},
\end{align*}%
\begin{align*}
\xi_{2n}  &  =\left(  \mathcal{V}_{V,p}^{+(s)}(h_{n})\right)  ^{-1/2}\left(
E[s!e_{s}^{\prime}\hat{\beta}_{V,p}^{+}(h_{n})|\mathcal{S}_{n}]-s!e_{s}%
^{\prime}\beta_{V}^{+}(h_{n})-h_{n}^{1+p-s}\frac{\mu_{V+}^{(p+1)}}%
{(p+1)!}\mathcal{B}_{s,p,p+1}^{+}(h_{n})\right) \\
&  =O_{P}\left(  \sqrt{nh_{n}^{1+2s}}\right)  O_{P}\left(  h_{n}%
^{2+p-s}\right)  =O_{P}\left(  \sqrt{nh_{n}^{5+2p}}\right)  =o_{P}(1).
\end{align*}
Write
\begin{align*}
\xi_{1n}  &  =\sum_{i=1}^{n}\varpi_{n,i}\varepsilon_{V_{i}}+o_{P}(1)\\
&  \equiv\tilde{\xi}_{1n}+o_{P}(1),
\end{align*}
where%
\[
\varpi_{n,i}=\left(  n^{-1}h_{n}^{-1-2s}e_{l}^{\prime}\Gamma_{+,l}^{-1}%
\Psi_{V,l}^{+}\Gamma_{+,l}^{-1}e_{l}\right)  ^{-1/2}e_{p}^{\prime}\Gamma
_{+,p}^{-1}(h_{n})X_{i,p}k_{ih_{n}}^{+}/n.
\]
For any $\lambda\in\mathbb{R}^{m}$ with unit norm, $\{\lambda^{\prime}%
\varpi_{n,i}\varepsilon_{V_{i}}\}_{i=1}^{n}$ is a triangular array of
independent random variables where $\E[\tilde{\xi}_{1n}]=0$ and $Var\left(
\tilde{\xi}_{1n}\right)  \rightarrow1.$ Thus, $\tilde{\xi}_{1n}\rightarrow
_{d}N\left(  0,1\right)  $ by Linderberg-Feller central limit theorem for
triangular arrays because%
\begin{align*}
\sum_{i=1}^{n}\E\left[  \left\vert \lambda^{\prime}\varpi_{n,i}\varepsilon
_{V_{i}}\right\vert ^{4}\right]   &  \lesssim n^{2}h_{n}^{2+4s}h_{n}^{-4s}%
\sum_{i=1}^{n}\E\left[  \left\vert e_{s}^{\prime}\Gamma_{+,p}^{-1}%
(h_{n})X_{i,p}k_{ih_{n}}^{+}\right\vert ^{4}/n^{4}\right] \\
&  =O(n^{-1}h_{n}^{-1})\\
&  =o(1).
\end{align*}
The result for $\hat{\mu}_{V,p}^{-(s)}$ and $\hat{\delta}_{V,p}$ follows the
same arguments.
\end{proof}
\end{lemma}

\begin{lemma}
\label{Rsmall}Under Assumption \ref{AssumptionRRD}(1)-(7), for $d\geq p+2,$ $V=Y$
and $X,$
\begin{align*}
\hat{c}-c &  =O_{P}\left(  \frac{1}{\sqrt{nh_{n}}}+h_{n}^{(p+1)}\right)  ,\\
\hat{\delta}_{V}-\delta_{V} &  =O_{P}\left(  \frac{1}{\sqrt{nh_{n}}}%
+h_{n}^{(p+1)}\right)  ,
\end{align*}%
\[
\hat{\mathcal{R}}=O_{P}\left(  \frac{1}{nh_{n}}+h_{n}^{2(p+1)}\right)  ,
\]

\begin{proof}
A Taylor expansion, $\hat{\pi}-\pi=O_{P}(n^{-1/2})$ and previous results yield%
\[
(\hat{c}-c)=G_{+}(\hat{\mu}_{V}^{+}-\mu_{V}^{+})+G_{-}\left(  \hat{\mu}%
_{V}^{-}-\mu_{V}^{-}\right)  +o_{P}(n^{-1/2}h_{n}^{-1/2}+h_{n}^{(p+1)}),
\]
where $G_{+}$ and $G_{-}$ are the derivatives of $g$ with respect to $\mu
_{V}^{+}$ and $\mu_{V}^{-},$ respectively. Previous lemmas imply for $V=Y$ and
$X_i$%
\[
\left(  \hat{\delta}_{V}-\delta_{V}\right)  ^{2}=O_{P}\left(  \frac{1}{nh_{n}%
}+h_{n}^{2(p+1)}\right)  .
\]
Write
\begin{equation}
\hat{\tau}_{Y}-\tau_{Y}\equiv c^{\prime}\left(  \hat{\delta}_{Y}-\delta
_{Y}\right)  +(\hat{c}-c)^{\prime}\delta_{Y}+(\hat{c}-c)^{\prime}\left(
\hat{\delta}_{Y}-\delta_{Y}\right)  \label{A2}%
\end{equation}
Thus, by previous lemmas
\[
\left(  \hat{\tau}_{Y}-\tau_{Y}\right)  ^{2}=O_{P}\left(  \frac{1}{nh_{n}%
}+h_{n}^{2(p+1)}\right)  .
\]
The same results hold for $\left(  \hat{\tau}_{X}-\tau_{X}\right)  ^{2}.$
Furthermore,
\[
\left(  \hat{\tau}_{Y}-\tau_{Y}\right)  \left(  \hat{\tau}_{X}-\tau
_{X}\right)  =O_{P}\left(  \sqrt{\frac{1}{nh_{n}}+h_{n}^{2(p+1)}}\right)
O_{P}\left(  \sqrt{\frac{1}{nh_{n}}+h_{n}^{2(p+1)}}\right)  .
\]
The results follows from the expression for $\hat{\mathcal{R}}$ in (\ref{R}).
\end{proof}
\end{lemma}

\noindent Define the linear term in the expansion (\ref{expansion})%
\[
\hat{l}_{\omega}=\frac{\hat{\tau}_{Y}-\tau_{Y}}{\tau_{X}}-\frac{\tau
_{Y}\left(  \hat{\tau}_{X}-\tau_{X}\right)  }{\tau_{X}^{2}},
\]
where each term $\hat{\tau}_{V}$ has been estimated by a $p-th$ order local polynomial.

\begin{lemma}
\label{Bias2}Under Assumption \ref{AssumptionRRD}(1)-(7), for $d\geq p+2,$%
\[
\E[\hat{l}_{\omega}|\mathcal{S}_{n}]=h_{n}^{p+1}\mathbf{B}_{\omega,p,p+1}%
(h_{n})+h_{n}^{p+2}\mathbf{B}_{\omega,p,p+2}(h_{n})+o_{P}(h_{n}^{2+p}),
\]
where
\begin{equation}
\mathbf{B}_{\omega,p,m}=\frac{B_{\tau_{Y},p,m}(h_{n})}{\tau_{X}}-\frac
{\tau_{Y}B_{\tau_{X},p,m}(h_{n})}{\tau_{X}^{2}}, \label{biasw}%
\end{equation}%
\[
B_{\tau_{Y},p,m}(h_{n})=c^{\prime}B_{Y,p,m}(h_{n})+B_{C,p,m}^{\prime}%
(h_{n})\delta_{Y},
\]%
\[
B_{\tau_{X},p,m}(h_{n})=c^{\prime}B_{X,p,m}(h_{n})+B_{C,p,m}^{\prime}%
(h_{n})\delta_{X},
\]
with $B_{V,l,m}$ defined (\ref{biasdelta}) and%
\[
B_{C,p,m}(h_{n})=G_{+}\frac{\mu_{V}^{+(m)}}{m!}\mathcal{B}_{p,m}^{+}%
(h_{n})+G_{-}\frac{\mu_{V}^{-(m)}}{m!}\mathcal{B}_{p,m}^{-}(h_{n}).
\]

\begin{proof}
By the delta method%
\[
\E[\hat{c}-c|\mathcal{S}_{n}]=h_{n}^{1+p}B_{C,p,p+1}(h_{n})+h_{n}%
^{2+p}B_{C,p,p+2}(h_{n})+o_{P}(h_{n}^{2+p}).
\]
From (\ref{A2}) and the bias expression for $\E[\hat{\delta}_{V,l}%
(h_{n})|\mathcal{S}_{n}]$ it holds%
\[
\E[\hat{\tau}_{Y}-\tau_{Y}|\mathcal{S}_{n}]=h_{n}^{1+p}B_{\tau_{Y},p,p+1}%
(h_{n})+h_{n}^{2+p}B_{\tau_{Y},p,p+2}(h_{n})+o_{P}(h_{n}^{2+p}),
\]
The same expression holds for $\E[\hat{\tau}_{X}-\tau_{X}|\mathcal{S}_{n}].$
The result follows from the expression for $\hat{l}_{\omega}.$
\end{proof}
\end{lemma}

\begin{lemma}
\label{Variance2}Under Assumption \ref{AssumptionRRD}(1)-(7), for $d\geq p+2,$%
\begin{equation}
\mathbf{V}_{\omega,p}(h_{n})=Var[\hat{l}_{\omega}|\mathcal{S}_{n}]=\frac
{1}{\tau_{X}^{2}}\mathbf{V}_{\tau,Y,p}(h_{n})-\frac{2\tau_{Y}}{\tau_{X}^{3}%
}\mathbf{C}_{\tau,p}(h_{n})+\frac{\tau_{Y}^{2}}{\tau_{X}^{4}}\mathbf{V}%
_{\tau,X,p}(h_{n}), \label{variancew}%
\end{equation}
where

\vspace{-1.5cm}
\begin{align*}
\mathbf{V}_{\tau,Y,p}(h_{n})  &  =c^{\prime}\mathbf{V}_{Y,p}(h_{n}%
)c+\delta_{Y}^{\prime}G_{+}\mathcal{V}_{V,p}^{+}(h_{n})G_{+}^{\prime}%
\delta_{Y}+\delta_{Y}^{\prime}G_{-}\mathcal{V}_{V,p}^{-}(h_{n})G_{-}^{\prime
}\delta_{Y},\\
\mathbf{C}_{\tau,p}(h_{n})  &  =c^{\prime}\mathbf{C}_{p}(h_{n})c+c^{\prime
}\mathbf{C}_{YC,p}(h_{n})\delta_{X}+\delta_{Y}^{\prime}\mathbf{C}_{XC,p}%
(h_{n})c+\delta_{Y}^{\prime}\mathbf{V}_{C,p}(h_{n})\delta_{X},
\\
\mathbf{C}_{p}(h_{n})  &  =\mathcal{C}_{YX,p}^{+}(h_{n})+\mathcal{C}%
_{YX,p}^{-}(h_{n}),\\
\mathbf{C}_{YC,p}(h_{n})  &  =G_{+}\mathcal{C}_{YV,p}^{+}(h_{n})-G_{-}%
\mathcal{C}_{YV,p}^{-}(h_{n}),\\
\mathbf{C}_{XC,p}(h_{n})  &  =G_{+}\mathcal{C}_{XV,p}^{+}(h_{n})-G_{-}%
\mathcal{C}_{XV,p}^{-}(h_{n}),
\\
\mathbf{V}_{C,p}(h_{n})  &  =G_{+}\mathcal{V}_{V,p}^{+}(h_{n})G_{+}^{\prime
}+G_{-}\mathcal{V}_{V,p}^{-}(h_{n})G_{-}^{\prime},\\
\mathcal{C}_{UV,p}^{+}(h_{n})  &  =n^{-1}e_{0}^{\prime}\Gamma_{+,p}^{-1}%
(h_{n})\Psi_{UV,p}^{+}(h_{n})\Gamma_{+,p}^{-1}(h_{n})e_{0},
\\
\mathcal{C}_{UV,p}^{-}(h_{n})  &  =n^{-1}e_{0}^{\prime}\Gamma_{-,p}^{-1}%
(h_{n})\Psi_{UV,p}^{-}(h_{n})\Gamma_{-,p}^{-1}(h_{n})e_{0}.
\end{align*}

\begin{proof}
From the expression of $\hat{l}_{\omega},$
\[
Var[\hat{l}_{\omega}|\mathcal{S}_{n}]=\frac{1}{\tau_{X}^{2}}Var[\hat{\tau}%
_{Y}-\tau_{Y}|\mathcal{S}_{n}]-\frac{2\tau_{Y}}{\tau_{X}^{3}}Cov[\hat{\tau
}_{Y}-\tau_{Y},\hat{\tau}_{X}-\tau_{X}|\mathcal{S}_{n}]+\frac{\tau_{Y}^{2}%
}{\tau_{X}^{4}}Var[\hat{\tau}_{X}-\tau_{X}|\mathcal{S}_{n}].
\]
Fromt the expansion of $\hat{\tau}_{V}-\tau_{V},$
\begin{align*}
Var[\hat{\tau}_{Y}-\tau_{Y}|\mathcal{S}_{n}]  &  =c^{\prime}Var[\hat{\delta
}_{Y}-\delta_{Y}|\mathcal{S}_{n}]c+\delta_{Y}^{\prime}G_{+}Var[\hat{\mu}%
_{Y}^{+}-\mu_{Y}^{+}|\mathcal{S}_{n}]G_{+}^{\prime}\delta_{Y}\\
&  +\delta_{Y}^{\prime}G_{-}Var[\hat{\mu}_{Y}^{+}-\mu_{Y}^{+}|\mathcal{S}%
_{n}]G_{-}^{\prime}\delta_{Y},+o_{P}(1),\\
Cov[\hat{\tau}_{Y}-\tau_{Y},\hat{\tau}_{X}-\tau_{X}|\mathcal{S}_{n}]  &
=c^{\prime}Cov[\hat{\delta}_{Y}-\delta_{Y},\hat{\delta}_{X}-\delta
_{X}|\mathcal{S}_{n}]c\\
&  +c^{\prime}Cov[\hat{\delta}_{Y}-\delta_{Y},\hat{c}-c|\mathcal{S}_{n}%
]\delta_{X}\\
&  +\delta_{Y}^{\prime}Cov[\hat{c}-c,\hat{\delta}_{Y}-\delta_{Y}%
|\mathcal{S}_{n}]c\\
&  +\delta_{Y}^{\prime}Var[\hat{c}-c|\mathcal{S}_{n}]\delta_{X}+o_{P}(1)\\
Var[\hat{\tau}_{X}-\tau_{X}|\mathcal{S}_{n}]  &  =c^{\prime}Var[\hat{\delta
}_{X}-\delta_{X}|\mathcal{S}_{n}]c+\delta_{X}^{\prime}G_{+}Var[\hat{\mu}%
_{X}^{+}-\mu_{X}^{+}|\mathcal{S}_{n}]G_{+}^{\prime}\delta_{X}\\
&  +\delta_{X}^{\prime}G_{-}Var[\hat{\mu}_{X}^{+}-\mu_{X}^{+}|\mathcal{S}%
_{n}]G_{-}^{\prime}\delta_{X},+o_{P}(1).
\end{align*}
Similarly to our previous results, it is shown that%
\begin{align*}
Cov[\hat{\delta}_{Y}-\delta_{Y},\hat{\delta}_{X}-\delta_{X}|\mathcal{S}_{n}]
&  =Cov[\hat{\mu}_{Y}^{+}-\mu_{Y}^{+},\hat{\mu}_{X}^{+}-\mu_{X}^{+}%
|\mathcal{S}_{n}]+Cov[\hat{\mu}_{Y}^{-}-\mu_{Y}^{-},\hat{\mu}_{X}^{-}-\mu
_{X}^{-}|\mathcal{S}_{n}]\\
&  =\mathcal{C}_{YX,p}^{+}(h_{n})+\mathcal{C}_{YX,p}^{-}(h_{n}),\\
Cov[\hat{\delta}_{Y}-\delta_{Y},\hat{c}-c|\mathcal{S}_{n}]  &  =G_{+}%
\mathcal{C}_{YV,p}^{+}(h_{n})-G_{-}\mathcal{C}_{YV,p}^{-}(h_{n})+o_{P}(1),\\
Var[\hat{c}-c|\mathcal{S}_{n}]  &  =G_{+}\mathcal{V}_{V,p}^{+}(h_{n}%
)G_{+}^{\prime}+G_{-}\mathcal{V}_{V,p}^{-}(h_{n})G_{-}^{\prime}+o_{P}(1).
\end{align*}

\end{proof}
\end{lemma}

\begin{lemma}
\label{MainAN}Under Assumption \ref{AssumptionRRD}(1)-(7), $nh_{n}^{2p+5}%
\rightarrow0$ and $d\geq p+2$%
\[
\left(  \mathbf{V}_{\omega,p}(h_{n})\right)  ^{-1/2}\left(  \hat{\beta
}_{\omega}-\beta_{\omega}-h_{n}^{p+1}\mathbf{B}_{\omega,p,p+1}(h_{n})\right)
\rightarrow_{d}N\left(  0,1\right)  ,
\]

\begin{proof}
Write the left hand side of the last display as%
\begin{align*}
\left(  \mathbf{V}_{\omega,p}(h_{n})\right)  ^{-1/2}\left(  \hat{l}_{\omega
}-h_{n}^{p+1}\mathbf{B}_{\omega,p,p+1}(h_{n})\right)  +\left(  \mathbf{V}%
_{\omega,p}(h_{n})\right)  ^{-1/2}\hat{\mathcal{R}} &  =\xi_{1n}+\xi_{2n},\\
&  =\xi_{1n}+o_{P}(1),
\end{align*}
where%
\begin{align*}
\xi_{2n} &  =O_{P}\left(  \frac{\sqrt{nh_{n}}}{nh_{n}}+\sqrt{nh_{n}}%
h_{n}^{2(p+1)}\right)  \\
&  =o_{P}(1)
\end{align*}
and%
\begin{align*}
\xi_{1n} &  =\left(  \mathbf{V}_{\omega,p}(h_{n})\right)  ^{-1/2}\left(
\hat{l}_{\omega}-\E[\hat{l}_{\omega}|\mathcal{S}_{n}]\right)  +\left(
\mathbf{V}_{\omega,p}(h_{n})\right)  ^{-1/2}\left(  h_{n}^{p+2}\mathbf{B}%
_{\omega,p,p+2}(h_{n})+o_{P}(h_{n}^{2+p})\right)  \\
&  \rightarrow_{d}N\left(  0,1\right)  ,
\end{align*}
by Lemma \ref{AN}.
\end{proof}
\end{lemma}

\begin{theorem}
\label{TheoremATEs} Let Assumption \ref{ass:rdd} hold.
Suppose also that Assumption \ref{AssumptionRRD}(1)-(7) in Section
\ref{AppendixSectionRRD} in the Appendix hold with $p=1$. Then
\[
\left(  \mathbf{V}_{\omega,1}(h_{n})\right)  ^{-1/2}\left(  \hat{\beta
}_{\omega}-\beta_{\omega}-h_{n}^{2}\mathbf{B}_{\omega,1,2}(h_{n})\right)
\rightarrow_{d}N\left(  0,1\right)  ,
\]
where $\mathbf{B}_{\omega,1,2}(h_{n})$ and  $\mathbf{V}_{\omega,1}(h_{n})$ are
given in \eqref{biasw} and  \eqref{variancew}, respectively,  in Section
\ref{AppendixSectionRRD} in the Appendix.
\end{theorem}

\noindent\textbf{Proof of Theorem \ref{TheoremATEs}}: It follows from Lemma
\ref{MainAN} for $p=1.$ $\blacksquare$\bigskip

\noindent When $g$ depends on $\delta_{X}$ the previous formulas simplify. In
this case and with an abuse of notation, we define the linear term as a linear
combination of $\hat{\delta}_{Y}-\delta_{Y}$ and $\hat{\delta}_{X}-\delta_{X}$
as%
\[
\hat{l}_{\omega}=c_{1}^{\prime}(\hat{\delta}_{Y}-\delta_{Y})+c_{2}^{\prime
}\left(  \hat{\delta}_{X}-\delta_{X}\right)  ,
\]
where $c_{1}=c/\tau_{X}$, $c_{2}^{\prime}=\left[  \delta_{Y}^{\prime}%
G_{\delta}-\beta_{w}(c^{\prime}+\delta_{X}^{\prime}G_{\delta})\right]
/\tau_{X}$ and $G_{\delta}$ is the derivative of $g$ with respect to
$\delta_{X}$.

\begin{lemma}
\label{Biasc2}Under Assumption \ref{AssumptionRRD}(1)-(7), for $d\geq p+2,$%
\[
\E[\hat{l}_{\omega}|\mathcal{S}_{n}]=h_{n}^{p+1}\mathbf{B}_{\omega,p,p+1}%
(h_{n})+h_{n}^{p+2}\mathbf{B}_{\omega,p,p+2}(h_{n})+o_{P}(h_{n}^{2+p}),
\]
where now
\begin{equation}
\mathbf{B}_{\omega,p,m}=c_{1}^{\prime}B_{Y,p,m}(h_{n})+c_{2}^{\prime}%
B_{X,p,m}(h_{n}),\label{BiasCase2}%
\end{equation}
and $B_{V,p,m}(h_{n})$ is defined in (\ref{biasdelta}).

\bigskip
\begin{proof}
It follows by a simple Taylor expansion and our previous results.
\end{proof}
\end{lemma}

\noindent Recall the definitions for conditional covariances and variances of
$(\hat{\mu}_{U}^{+}-\mu_{U}^{+})$ and $\left(  \hat{\mu}_{V}^{+}-\mu_{V}%
^{+}\right)  $ (and similarly for $-$ parts)%
\begin{align*}
\mathcal{C}_{UV,p}^{+}(h_{n})  &  =n^{-1}e_{0}^{\prime}\Gamma_{+,p}^{-1}%
(h_{n})\Psi_{UV,p}^{+}(h_{n})\Gamma_{+,p}^{-1}(h_{n})e_{0},\\
\mathcal{C}_{UV,p}^{-}(h_{n})  &  =n^{-1}e_{0}^{\prime}\Gamma_{-,p}^{-1}%
(h_{n})\Psi_{UV,p}^{-}(h_{n})\Gamma_{-,p}^{-1}(h_{n})e_{0},\\
\mathcal{V}_{V,p}^{+}(h_{n})  &  =\mathcal{C}_{VV,p}^{+}(h_{n})\text{ and
}\mathcal{V}_{V,p}^{-}(h_{n})=\mathcal{C}_{VV,p}^{-}(h_{n}).
\end{align*}

\begin{lemma}
\label{Variancec2}Under Assumption \ref{AssumptionRRD}(1)-(7), for $d\geq p+2,$%
\begin{equation}
\mathbf{V}_{\omega,p}(h_{n})=Var[\hat{l}_{\omega}|\mathcal{S}_{n}%
]=V_{+1}+V_{-1}, \label{VarianceCase2}%
\end{equation}
where%
\[
V_{+1}=c_{1}^{\prime}\mathcal{V}_{Y,p}^{+}(h_{n})c_{1}+2c_{1}^{\prime
}\mathcal{C}_{YX,p}^{+}(h_{n})c_{2}+c_{2}^{\prime}\mathcal{V}_{X,p}^{+}%
(h_{n})c_{2},
\]
and
\[
V_{-1}=c_{1}^{\prime}\mathcal{V}_{Y,p}^{-}(h_{n})c_{1}+2c_{1}^{\prime
}\mathcal{C}_{YX,p}^{-}(h_{n})c_{2}+c_{2}^{\prime}\mathcal{V}_{X,p}^{-}%
(h_{n})c_{2}.
\]
An equivalent representation is%
\[
\mathbf{V}_{\omega,p}(h_{n})=c_{1}^{\prime}\mathbf{V}_{Y,p}(h_{n})c_{1}%
+2c_{1}^{\prime}\mathbf{C}_{p}(h_{n})c_{2}+c_{2}^{\prime}\mathbf{V}%
_{X,p}(h_{n})c_{2},
\]
where $\mathbf{V}_{V,p}(h_{n})$ is defined in (\ref{variancedelta}), and
\[
\mathbf{C}_{p}(h_{n})=\mathcal{C}_{YX,p}^{+}(h_{n})+\mathcal{C}_{YX,p}%
^{-}(h_{n}).
\]

\begin{proof}
Write $\hat{l}_{\omega}=\hat{l}_{\omega}^{+}-\hat{l}_{\omega}^{+},$ where
\begin{align*}
\hat{l}_{\omega}^{+}  &  =c_{1}^{\prime}(\hat{\mu}_{Y}^{+}-\mu_{Y}^{+}%
)+c_{2}^{\prime}\left(  \hat{\mu}_{X}^{+}-\mu_{X}^{+}\right)  ,\\
\hat{l}_{\omega}^{-}  &  =c_{1}^{\prime}(\hat{\mu}_{Y}^{-}-\mu_{Y}^{-}%
)+c_{2}^{\prime}\left(  \hat{\mu}_{X}^{-}-\mu_{X}^{-}\right)  .
\end{align*}
The proof follows from previous results.
\end{proof}
\end{lemma}

The following theory is for the bias-correction. Assume we use a $q$-th order
polynomial for estimating the bias term (the $p+1$ derivative), $q\geq p+1$.
Linearizing the bias-corrected estimator, we obtain
\begin{align*}
\hat{\beta}_{w}^{bc}-\beta_{w} &  =\hat{\beta}_{w}-\beta_{w}-h_{n}^{p+1}%
\hat{B}_{w,p,q}(h_{n},b_{n})\\
&  =\hat{l}_{\omega}-h_{n}^{p+1}\tilde{B}_{w,p,q}(h_{n},b_{n})+R_{n}%
-h_{n}^{p+1}\left(  \hat{B}_{w,p,q}(h_{n},b_{n})-\tilde{B}_{w,p,q}(h_{n}%
,b_{n})\right)  \\
&  \equiv\hat{l}_{\omega}^{bc}+R_{n}-R_{n}^{bc},
\end{align*}
where%
\begin{align*}
\hat{B}_{w,p,q}(h_{n},b_{n}) &  =\hat{c}_{1}^{\prime}\hat{B}_{Y,p,q}%
(h_{n},b_{n})+\hat{c}_{2}^{\prime}\hat{B}_{X,p,q}(h_{n},b_{n}),\\
\tilde{B}_{w,p,q}(h_{n},b_{n}) &  =c_{1}^{\prime}\hat{B}_{Y,p,q}(h_{n}%
,b_{n})+c_{2}^{\prime}\hat{B}_{X,p,q}(h_{n},b_{n}),\\
\hat{l}_{\omega}^{bc} &  =\hat{l}_{\omega}-h_{n}^{p+1}\tilde{B}_{w,p,q}%
(h_{n},b_{n}),\\
R_{n}^{bc} &  =h_{n}^{p+1}\left(  \hat{B}_{w,p,q}(h_{n},b_{n})-\tilde
{B}_{w,p,q}(h_{n},b_{n})\right)  ,\\
\hat{B}_{V,p,q}(h_{n},b_{n}) &  =\frac{\hat{\mu}_{V,q}^{+(p+1)}(b_{n}%
)}{(p+1)!}\mathcal{B}_{p,p+1}^{+}(h_{n})-\frac{\hat{\mu}_{V,q}^{-(p+1)}%
(b_{n})}{(p+1)!}\mathcal{B}_{p,p+1}^{-}(h_{n})
\end{align*}
and $\mathcal{B}_{l,m}^{+}(h_{n})=e_{0}^{\prime}\Gamma_{+,l}^{-1}%
(h_{n})\vartheta_{l,m}^{+}(h_{n})$ and $\mathcal{B}_{l,m}^{-}(h_{n}%
)=e_{0}^{\prime}\Gamma_{-,l}^{-1}(h_{n})\vartheta_{l,m}^{-}(h_{n}).$

\begin{lemma}
\label{Rbcsmall}Under Assumption \ref{AssumptionRRD}(1)-(7), for $d\geq p+2,$
$n\min\{h_{n},b_{n}\}\rightarrow\infty,$ $h_{n}\rightarrow0,$ and
$nh_{n}\rightarrow\infty,$ then%
\[
R_{n}=O_{P}\left(  \frac{1}{nh_{n}}+h_{n}^{2(p+1)}\right)  ,
\]%
\[
R_{n}^{bc}=O_{P}\left(  \frac{h_{n}^{p+1}}{\sqrt{nh_{n}}}+h_{n}^{2(p+1)}%
\right)  O_{P}\left(  \frac{1}{\sqrt{nb_{n}^{3+2p}}}+1\right)  .
\]

\begin{proof}
Note%
\begin{equation}
R_{n}=\frac{R_{\tau X}}{\tau_{X}}-\frac{\tau_{Y}R_{\tau Y}}{\tau_{X}^{2}}%
+\hat{\mathcal{R}},\label{Rn}%
\end{equation}
where $R_{\tau X}$ and $R_{\tau Y}$ are implicitly defined from
\[
\hat{\tau}_{Y}-\tau_{Y}\equiv c^{\prime}\left(  \hat{\delta}_{Y}-\delta
_{Y}\right)  +\delta_{Y}^{\prime}G_{\delta}(\hat{\delta}_{X}-\delta
_{X})+R_{\tau Y}%
\]
and%
\[
\hat{\tau}_{X}-\tau_{X}\equiv\left(  c^{\prime}+\delta_{X}^{\prime}G_{\delta
}\right)  (\hat{\delta}_{X}-\delta_{X})+R_{\tau X}.
\]
By continuous differentiability of $g$ and Lemmas \ref{Bias} and
\ref{Variance}
\[
R_{\tau V}=O_{P}\left(  \frac{1}{nh_{n}}+h_{n}^{2(p+1)}\right)  .
\]
The rate for $R_{n}$ then follows from (\ref{Rn}), the last display and Lemma
\ref{Rsmall}. On the other hand, by previous results
\begin{align*}
\left\vert R_{n}^{bc}\right\vert  &  \lesssim h_{n}^{p+1}\left\vert \hat
{c}_{1}-c_{1}\right\vert \left\{  \left\vert e_{p+1}^{\prime}\hat{\beta}%
_{Y,q}^{+}(b_{n})\right\vert +\left\vert e_{p+1}^{\prime}\hat{\beta}_{Y,q}%
^{-}(b_{n})\right\vert \right\}  \\
&  +h_{n}^{p+1}\left\vert \hat{c}_{2}-c_{2}\right\vert \left\{  \left\vert
e_{p+1}^{\prime}\hat{\beta}_{X,q}^{+}(b_{n})\right\vert +\left\vert
e_{p+1}^{\prime}\hat{\beta}_{X,q}^{-}(b_{n})\right\vert \right\}  \\
&  =O_{P}\left(  \frac{h_{n}^{p+1}}{\sqrt{nh_{n}}}+h_{n}^{2(p+1)}\right)
O_{P}\left(  \frac{1}{\sqrt{nb_{n}^{3+2p}}}+1\right)  ,
\end{align*}
By\ Lemmas \ref{Rsmall}, \ref{Bias} and \ref{Variance} .
\end{proof}
\end{lemma}

\begin{lemma}
\label{biasbc}Under Assumption \ref{AssumptionRRD}(1)-(7), for $d\geq p+2,$
$n\min\{h_{n},b_{n}\}\rightarrow\infty,$ $\max\{h_{n},b_{n}\}\rightarrow0,$
then%
\begin{align*}
\E[\hat{l}_{\omega}^{bc}|\mathcal{S}_{n}] &  =h_{n}^{p+2}\mathbf{B}%
_{\omega,p,p+2}(h_{n})[1+o_{P}(1)]\\
&  +h_{n}^{p+1}b_{n}^{p-q}B_{\omega,p,q}^{bc}(h_{n},b_{n})[1+o_{P}(1)],
\end{align*}
where
\[
B_{\omega,p,q}^{bc}(h_{n},b_{n})=c_{1}^{\prime}B_{Y,p,q}^{bc}(h_{n}%
,b_{n})+c_{2}^{\prime}B_{X,p,q}^{bc}(h_{n},b_{n})
\]%
\[
B_{V,p,q}^{bc}(h_{n},b_{n})=\frac{\mu_{V}^{+(q+1)}}{(q+2)!}\mathcal{B}%
_{p+1,q,q+1}^{+}(b_{n})\frac{\mathcal{B}_{p,p+1}^{+}(h_{n})}{(p+1)!}-\frac
{\mu_{V}^{-(q+1)}}{(q+1)!}\mathcal{B}_{p+1,q,q+1}^{-}(b_{n})\frac
{\mathcal{B}_{p,p+1}^{-}(h_{n})}{(p+1)!}.
\]

\begin{proof}
Note with
\begin{align*}
\E[\hat{l}_{\omega}^{bc}|\mathcal{S}_{n}] &  =\E[\hat{l}_{\omega}-h_{n}%
^{p+1}\mathbf{B}_{\omega,p,p+1}(h_{n})|\mathcal{S}_{n}]\\
&  -h_{n}^{p+1}\E[\tilde{B}_{\omega,p,q}(h_{n},b_{n})-\mathbf{B}_{\omega
,p,p+1}(h_{n})|\mathcal{S}_{n}]\\
&  \equiv B_{1}-h_{n}^{p+1}B_{2}.
\end{align*}
By\ Lemma \ref{Bias}
\[
B_{1}=h_{n}^{p+2}\mathbf{B}_{\omega,p,p+2}(h_{n})[1+o_{P}(1)]
\]
and
\begin{align*}
B_{2} &  =c_{1}^{\prime}\E[\hat{B}_{Y,p,q}(h_{n},b_{n})-B_{Y,p,p+1}%
(h_{n})|\mathcal{S}_{n}]+c_{2}^{\prime}\E[\hat{B}_{X,p,q}(h_{n},b_{n}%
)-B_{X,p,p+1}(h_{n})|\mathcal{S}_{n}]\\
&  =c_{1}^{\prime}\E[e_{p+1}^{\prime}\hat{\beta}_{Y,q}^{+}(b_{n})-e_{p+1}%
^{\prime}\beta_{Y}^{+}|\mathcal{S}_{n}]\mathcal{B}_{p,p+1}^{+}(h_{n}%
)-c_{1}^{\prime}\E[e_{p+1}^{\prime}\hat{\beta}_{Y,q}^{-}(b_{n})-e_{p+1}%
^{\prime}\beta_{Y}^{-}|\mathcal{S}_{n}]\mathcal{B}_{p,p+1}^{-}(h_{n})\\
&  +c_{2}^{\prime}\E[e_{p+1}^{\prime}\hat{\beta}_{X,q}^{+}(b_{n})-e_{p+1}%
^{\prime}\beta_{X}^{+}|\mathcal{S}_{n}]\mathcal{B}_{p,p+1}^{+}(h_{n}%
)-c_{2}^{\prime}\E[e_{p+1}^{\prime}\hat{\beta}_{X,q}^{-}(b_{n})-e_{p+1}%
^{\prime}\beta_{X}^{-}|\mathcal{S}_{n}]\mathcal{B}_{p,p+1}^{-}(h_{n})\\
&  =c_{1}^{\prime}b_{n}^{p-q}\left[  \frac{\mu_{Y}^{+(q+1)}}{(q+1)!}%
\mathcal{B}_{p+1,q,q+1}^{+}(b_{n})\frac{\mathcal{B}_{p,p+1}^{+}(h_{n}%
)}{(p+1)!}-\frac{\mu_{Y}^{-(q+1)}}{(q+1)!}\mathcal{B}_{p+1,q,q+1}^{-}%
(b_{n})\frac{\mathcal{B}_{p,p+1}^{-}(h_{n})}{(p+1)!}\right]  \\
&  +c_{2}^{\prime}b_{n}^{p-q}\left[  \frac{\mu_{X}^{+(q+1)}}{(q+1)!}%
\mathcal{B}_{p+1,q,q+1}^{+}(b_{n})\frac{\mathcal{B}_{p,p+1}^{+}(h_{n}%
)}{(p+1)!}-\frac{\mu_{X}^{-(q+1)}}{(q+1)!}\mathcal{B}_{p+1,q,q+1}^{-}%
(b_{n})\frac{\mathcal{B}_{p,p+1}^{-}(h_{n})}{(p+1)!}\right]  \\
&  =b_{n}^{q-p}B_{\omega,p,p+2}^{bc}(h_{n},b_{n}).
\end{align*}

\end{proof}
\end{lemma}

Define%
\begin{align*}
\mathcal{C}_{UV,p,q}^{+}(h_{n},b_{n})  &  =Cov[\hat{\mu}_{U,p}^{+}%
(h_{n}),\frac{\hat{\mu}_{V,q}^{+(p+1)}(b_{n})}{(p+1)!}|\mathcal{S}_{n}]\\
&  =\frac{1}{nb_{n}^{p+1}}e_{0}^{\prime}\Gamma_{+,p}^{-1}(h_{n})\Psi
_{UV,p,q}^{+}(h_{n},b_{n})\Gamma_{+,q}^{-1}(b_{n})e_{p+1},
\end{align*}
and%
\begin{align*}
\mathcal{C}_{UV,p,q}^{-}(h_{n},b_{n})  &  =Cov[\hat{\mu}_{U,p}^{-}%
(h_{n}),\frac{\hat{\mu}_{V,q}^{-(p+1)}(b_{n})}{(p+1)!}|\mathcal{S}_{n}]\\
&  =\frac{1}{nb_{n}^{p+1}}e_{0}^{\prime}\Gamma_{-,p}^{-1}(h_{n})\Psi
_{UV,p,q}^{-}(h_{n},b_{n})\Gamma_{-,q}^{-1}(b_{n})e_{p+1}.
\end{align*}
We drop one $U$ when $V=U$, e.g. $\mathcal{C}_{U,p,q}^{+}(h_{n},b_{n}%
)=\mathcal{C}_{UU,p,q}^{+}(h_{n},b_{n}).$

\begin{lemma}
\label{Variancebc}Under Assumption \ref{AssumptionRRD}(1)-(7), for $d\geq p+2,$
$n\min\{h_{n},b_{n}\}\rightarrow\infty,$ $\max\{h_{n},b_{n}\}\rightarrow0,$
then%
\begin{equation}
\mathbf{V}_{\omega,p,q}^{bc}(h_{n},b_{n})\equiv Var[\hat{l}_{\omega}%
^{bc}|\mathcal{S}_{n}]=\mathbf{V}_{\omega,p}(h_{n})+\mathbf{C}_{\omega
,p,q}^{bc}(h_{n},b_{n}),\label{Varbc}%
\end{equation}
where
\begin{equation}
\mathbf{C}_{\omega,p,q}^{bc}(h_{n},b_{n})=2h_{n}^{p+1}\left(  C_{+12}%
+C_{-12}\right)  +h_{n}^{2p+2}\left(  V_{+2}+V_{-2}\right)  \label{Cbc}%
\end{equation}%
\[
C_{+12}=\left[  c_{1}^{\prime}\mathcal{C}_{Y,p,q}^{+}(h_{n},b_{n})c_{1}%
+c_{1}^{\prime}\mathcal{C}_{YX,p,q}^{+}(h_{n},b_{n})c_{2}+c_{1}^{\prime
}\mathcal{C}_{XY,p,q}^{+}(h_{n},b_{n})c_{2}+c_{2}^{\prime}\mathcal{C}%
_{X,p,q}^{+}(h_{n},b_{n})c_{2}\right]  \frac{\mathcal{B}_{p,p+1}^{+}(h_{n}%
)}{(p+1)!},
\]%
\[
C_{-12}=\left[  c_{1}^{\prime}\mathcal{C}_{Y,p,q}^{-}(h_{n},b_{n})c_{1}%
+c_{1}^{\prime}\mathcal{C}_{YX,p,q}^{-}(h_{n},b_{n})c_{2}+c_{1}^{\prime
}\mathcal{C}_{XY,p,q}^{-}(h_{n},b_{n})c_{2}+c_{2}^{\prime}\mathcal{C}%
_{X,p,q}^{-}(h_{n},b_{n})c_{2}\right]  \frac{\mathcal{B}_{p,p+1}^{-}(h_{n}%
)}{(p+1)!},
\]%
\[
V_{+2}=\left(  c_{1}^{\prime}\mathcal{V}_{Y,p}^{+(q)}(b_{n})c_{1}%
+2c_{1}^{\prime}\mathcal{C}_{YX,p,q}^{+}(b_{n})c_{2}+c_{2}^{\prime}%
\mathcal{V}_{X,p}^{+(q)}(b_{n})c_{2}\right)  \left(  \frac{\mathcal{B}%
_{p,p+1}^{+}(h_{n})}{(p+1)!}\right)  ^{2}%
\]
and%
\[
V_{-2}=\left(  c_{1}^{\prime}\mathcal{V}_{Y,p}^{-(q)}(b_{n})c_{1}%
+2c_{1}^{\prime}\mathcal{C}_{YX,p,q}^{-}(b_{n})c_{2}+c_{2}^{\prime}%
\mathcal{V}_{X,p}^{-(q)}(b_{n})c_{2}\right)  \left(  \frac{\mathcal{B}%
_{p,p+1}^{-}(h_{n})}{(p+1)!}\right)  ^{2}%
\]

\begin{proof}
Note that $\hat{l}_{\omega}^{bc}=\hat{l}_{\omega}^{+bc}-\hat{l}_{\omega}%
^{-bc},$ so that
\[
Var[\hat{l}_{\omega}^{bc}|\mathcal{S}_{n}]=Var[\hat{l}_{\omega}^{+bc}%
|\mathcal{S}_{n}]+Var[\hat{l}_{\omega}^{-bc}|\mathcal{S}_{n}],
\]
where
\begin{align*}
\hat{l}_{\omega}^{+bc} &  =\hat{l}_{\omega}^{+}-h_{n}^{p+1}\tilde{B}%
_{w,p,q}^{+}(h_{n},b_{n}),\\
\hat{l}_{\omega}^{-bc} &  =-\hat{l}_{\omega}^{-}+h_{n}^{p+1}\tilde{B}%
_{w,p,q}^{-}(h_{n},b_{n}),\\
\hat{l}_{\omega}^{+} &  =c_{1}^{\prime}(\hat{\mu}_{Y}^{+}-\mu_{Y}^{+}%
)+c_{2}^{\prime}\left(  \hat{\mu}_{X}^{+}-\mu_{X}^{+}\right)  ,\\
\hat{l}_{\omega}^{-} &  =c_{1}^{\prime}(\hat{\mu}_{Y}^{-}-\mu_{Y}^{-}%
)+c_{2}^{\prime}\left(  \hat{\mu}_{X}^{-}-\mu_{X}^{-}\right)  ,\\
\tilde{B}_{w,p,q}^{\pm}(h_{n},b_{n}) &  =c_{1}^{\prime}\hat{B}_{Y,p,q}^{\pm
}(h_{n},b_{n})+c_{2}^{\prime}\hat{B}_{X,p,q}^{\pm}(h_{n},b_{n}),\\
\hat{B}_{V,p,q}^{\pm}(h_{n},b_{n}) &  =\frac{\hat{\mu}_{V,q}^{\pm(p+1)}%
(b_{n})}{(p+1)!}\mathcal{B}_{p,p+1}^{\pm}(h_{n}).
\end{align*}
This decomposition implies
\[
Var[\hat{l}_{\omega}^{+bc}|\mathcal{S}_{n}]=V_{+1}-2h_{n}^{p+1}C_{+12}%
+h_{n}^{2p+2}V_{+2},
\]
where%
\begin{align*}
C_{+12} &  =Cov[\hat{l}_{\omega}^{+},\tilde{B}_{w,p,q}^{+}(h_{n}%
,b_{n})|\mathcal{S}_{n}]\\
&  =\left(  c_{1}^{\prime}\mathcal{C}_{Y,p,q}^{+}c_{1}+c_{1}^{\prime
}\mathcal{C}_{YX,p,q}^{+}c_{2}+c_{1}^{\prime}\mathcal{C}_{XY,p,q}^{+}%
c_{2}+c_{2}^{\prime}\mathcal{C}_{X,p,q}^{+}c_{2}\right)  \frac{\mathcal{B}%
_{p,p+1}^{+}(h_{n})}{(p+1)!},
\end{align*}
and
\begin{align*}
V_{+2} &  =Var[\tilde{B}_{w,p,q}^{+}(h_{n},b_{n})|\mathcal{S}_{n}]\\
&  =\left(  c_{1}^{\prime}\mathcal{V}_{Y,p}^{-(q)}(b_{n})c_{1}+2c_{1}^{\prime
}\mathcal{C}_{YX,p,q}^{+}(b_{n})c_{2}+c_{2}^{\prime}\mathcal{V}_{X,p}%
^{-(q)}(b_{n})c_{2}\right)  \left(  \frac{\mathcal{B}_{p,p+1}^{+}(h_{n}%
)}{(p+1)!}\right)  ^{2}.
\end{align*}
For the left hand side limits the proof is the same.
\end{proof}
\end{lemma}

\begin{lemma}
\label{ANBC}Under Assumption \ref{AssumptionRRD},%
\[
\left(  \mathbf{V}_{\omega,p,q}^{bc}(h_{n},b_{n})\right)  ^{-1/2}\left(
\hat{\beta}_{\omega}^{bc}-\beta_{\omega}\right)  \rightarrow_{d}N\left(
0,1\right)  .
\]

\begin{proof}
As in \cite{Calonico_Cattaneo_Titiunik} it is shown that
\begin{align*}
\left(  \mathbf{V}_{\omega,p,q}^{bc}(h_{n},b_{n})\right)  ^{-1/2}\left(
\hat{\beta}_{\omega}^{bc}-\beta_{\omega}\right)   &  =\xi_{1n}+\xi_{2n},\\
&  =\xi_{1n}+o_{P}(1),
\end{align*}
where%
\[
\xi_{1n}=\left(  \mathbf{V}_{\omega,p,q}^{bc}(h_{n},b_{n})\right)  ^{-1/2}%
\hat{l}_{\omega}^{bc}.
\]
Proceeding as in Lemma \ref{AN} it is shown that $\xi_{1n}\rightarrow
_{d}N\left(  0,1\right)  $ by Linderberg-Feller central limit theorem for
triangular arrays.
\end{proof}
\end{lemma}

\noindent\textbf{Proof of Theorem \ref{thm:CW_RBC}}: It follows from Lemma
\ref{ANBC} for $p=1,$ $q=2$ and $\omega=CW.$ $\blacksquare$\bigskip

\subsection{Consistent Standard Error Estimators}

In the expressions for $\mathbf{V}_{\omega,p,q}^{bc}(h_{n},b_{n})\ $and
$\mathbf{V}_{\omega,p}(h_{n})$ the only unknown quantities are the vectors
$c_{1}$ and $c_{2},$ and the matrices $\Psi_{UV,p,q}^{+}(h_{n},b_{n})$ and
$\Psi_{UV,p,q}^{-}(h_{n},b_{n}).$ The vectors $c_{1}$ and $c_{2}$ are
estimated by $\hat{c}_{1}$ and $\hat{c}_{2},$ respectively. Natural estimators
for
\[
\Psi_{UV,p,q}^{+}(h_{n},b_{n})=\frac{1}{n}\sum_{i=1}^{n}X_{i,p}X_{i,q}%
^{\prime}\sigma_{UV}^{2}(Z_{i},\tilde{W}_{i})k_{ih_{n}}^{+}k_{ib_{n}}^{+}%
\]
and%
\[
\Psi_{UV,p,q}^{-}(h_{n},b_{n})=\frac{1}{n}\sum_{i=1}^{n}X_{i,p}X_{i,q}%
^{\prime}\sigma_{UV}^{2}(Z_{i},\tilde{W}_{i})k_{ih_{n}}^{-}k_{ib_{n}}^{-}%
\]
are%
\[
\hat{\Psi}_{UV,p,q}^{+}(h_{n},b_{n})=\frac{1}{n}\sum_{i=1}^{n}X_{i,p}%
X_{i,q}^{\prime}\hat{\varepsilon}_{U_{i}}^{+}\hat{\varepsilon}_{V_{i}}%
^{+}k_{ih_{n}}^{+}k_{ib_{n}}^{+}%
\]
and
\[
\hat{\Psi}_{UV,p,q}^{-}(h_{n},b_{n})=\frac{1}{n}\sum_{i=1}^{n}X_{i,p}%
X_{i,q}^{\prime}\hat{\varepsilon}_{U_{i}}^{-}\hat{\varepsilon}_{V_{i}}%
^{-}k_{ih_{n}}^{-}k_{ib_{n}}^{-}%
\]
where $\hat{\varepsilon}_{V_{i}}^{+}=V_{i}-\tilde{W}_{i}^{\prime}\hat{\mu
}_{V,1}^{+}(h_{n})$ and $\hat{\varepsilon}_{V_{i}}^{-}=V_{i}-\tilde{W}%
_{i}^{\prime}\hat{\mu}_{V,1}^{-}(h_{n}).$ Using this, a consistent estimator
for $\mathbf{V}_{\omega,p}(h_{n})$ is%
\[
\mathbf{\hat{V}}_{\omega,p}(h_{n})=\hat{V}_{+1}+\hat{V}_{-1},
\]
where%
\[
\hat{V}_{+1}=\hat{c}_{1}^{\prime}\mathcal{\hat{V}}_{Y,p}^{+}(h_{n})\hat{c}%
_{1}+2c_{1}^{\prime}\hat{C}_{YX,p}^{+}(h_{n})\hat{c}_{2}+\hat{c}_{2}^{\prime
}\mathcal{\hat{V}}_{X,p}^{+}(h_{n})\hat{c}_{2},
\]
and
\[
\hat{V}_{-1}=\hat{c}_{1}^{\prime}\mathcal{\hat{V}}_{Y,p}^{-}(h_{n})\hat{c}%
_{1}+2\hat{c}_{1}^{\prime}\hat{C}_{YX,p}^{-}(h_{n})\hat{c}_{2}+\hat{c}%
_{2}^{\prime}\mathcal{\hat{V}}_{X,p}^{-}(h_{n})\hat{c}_{2}.
\]%
\begin{align*}
\hat{C}_{UV,p}^{+}(h_{n})  &  =n^{-1}e_{0}^{\prime}\Gamma_{+,p}^{-1}%
(h_{n})\hat{\Psi}_{UV,p}^{+}(h_{n})\Gamma_{+,p}^{-1}(h_{n})e_{0},\\
\hat{C}_{UV,p}^{-}(h_{n})  &  =n^{-1}e_{0}^{\prime}\Gamma_{-,p}^{-1}%
(h_{n})\hat{\Psi}_{UV,p}^{-}(h_{n})\Gamma_{-,p}^{-1}(h_{n})e_{0},\\
\mathcal{\hat{V}}_{V,p}^{+}(h_{n})  &  =\hat{C}_{VV,p}^{+}(h_{n})\text{ and
}\mathcal{\hat{V}}_{V,p}^{-}(h_{n})=\hat{C}_{VV,p}^{-}(h_{n}).
\end{align*}
Likewise, we construct a consistent estimator for $\mathbf{V}_{\omega
,p,q}^{bc}(h_{n},b_{n})$ based on plug-in estimated residuals%
\begin{equation}
\mathbf{\hat{V}}_{\omega,p,q}^{bc}(h_{n},b_{n})=\mathbf{\hat{V}}_{\omega
,p}(h_{n})+\mathbf{\hat{C}}_{\omega,p,q}^{bc}(h_{n},b_{n}), \label{se}%
\end{equation}
where $\mathbf{\hat{C}}_{\omega,p,q}^{bc}$ replaces $c_{1}$ and $c_{2}$ by
$\hat{c}_{1}$ and $\hat{c}_{2}$ and $\Psi_{UV,p,q}^{+}(h_{n},b_{n})$ and
$\Psi_{UV,p,q}^{-}(h_{n},b_{n})$ by $\hat{\Psi}_{UV,p,q}^{+}(h_{n},b_{n})$ and
$\hat{\Psi}_{UV,p,q}^{-}(h_{n},b_{n}),$ respectively.

\begin{lemma}
\label{secon}Under Assumption \ref{AssumptionRRD},%
\[
\frac{\mathbf{\hat{V}}_{R,1,2}^{bc}(h_{n},b_{n})}{\mathbf{V}_{R,1,2}%
^{bc}(h_{n},b_{n})}\rightarrow_{p}1.
\]

\begin{proof}
We first show that%
\begin{equation}
\hat{\Psi}_{UV,p,q}^{+}(h_{n},b_{n})=\Psi_{UV,p,q}^{+}(h_{n},b_{n}%
)+o_{P}\left(  \frac{m_{n}}{h_{n}b_{n}}\right)  ,\label{aa1}%
\end{equation}
where $m_{n}=\min\{h_{n},b_{n}\}.$ Standard calculations show $\hat
{\varepsilon}_{V_{i}}^{+}=\varepsilon_{V_{i}}^{+}-\tilde{W}_{i}^{\prime
}\left(  \hat{\mu}_{V,1}^{+}(h_{n})-\mu_{V,1}^{+}\right)  $ and by previous
results
\begin{align*}
\hat{\Psi}_{UV,p,q}^{+}(h_{n},b_{n}) &  =\frac{1}{n}\sum_{i=1}^{n}%
X_{i,p}X_{i,q}^{\prime}\varepsilon_{U_{i}}^{+}\varepsilon_{V_{i}}^{+}%
k_{ih_{n}}^{+}k_{ib_{n}}^{+}\\
&  -\left[  \frac{1}{n}\sum_{i=1}^{n}X_{i,p}X_{i,q}^{\prime}\varepsilon
_{U_{i}}^{+}\tilde{W}_{i}^{\prime}k_{ih_{n}}^{+}k_{ib_{n}}^{+}\right]  \left(
\hat{\mu}_{V,1}^{+}(h_{n})-\mu_{V,1}^{+}\right)  \\
&  -\left[  \frac{1}{n}\sum_{i=1}^{n}X_{i,p}X_{i,q}^{\prime}\varepsilon
_{V_{i}}^{+}\tilde{W}_{i}^{\prime}k_{ih_{n}}^{+}k_{ib_{n}}^{+}\right]  \left(
\hat{\mu}_{U,1}^{+}(h_{n})-\mu_{U,1}^{+}\right)  \\
&  +\left(  \hat{\mu}_{U,1}^{+}(h_{n})-\mu_{U,1}^{+}\right)  ^{\prime}\left[
\frac{1}{n}\sum_{i=1}^{n}X_{i,p}X_{i,q}^{\prime}\tilde{W}_{i}\tilde{W}%
_{i}^{\prime}k_{ih_{n}}^{+}k_{ib_{n}}^{+}\right]  \left(  \hat{\mu}_{V,1}%
^{+}(h_{n})-\mu_{V,1}^{+}\right)  \\
&  =\frac{1}{n}\sum_{i=1}^{n}X_{i,p}X_{i,q}^{\prime}\varepsilon_{U_{i}}%
^{+}\varepsilon_{V_{i}}^{+}k_{ih_{n}}^{+}k_{ib_{n}}^{+}+o_{P}\left(
\frac{m_{n}}{h_{n}b_{n}}\right)  \\
&  \equiv\breve{\Psi}_{UV,p,q}^{+}(h_{n},b_{n})+o_{P}\left(  \frac{m_{n}%
}{h_{n}b_{n}}\right)  ,
\end{align*}
By the change of variables $u=Z/h_{n},$ it follows that%
\begin{align*}
&  \E\left[  \frac{h_{n}b_{n}}{m_{n}}\breve{\Psi}_{UV,p,q}^{+}(h_{n}%
,b_{n})\right]  \\
&  =\E\left[  \frac{h_{n}b_{n}}{m_{n}}\frac{1}{n}\sum_{i=1}^{n}X_{i,p}%
X_{i,q}^{\prime}\sigma_{UV}^{2}(Z_{i},\tilde{W}_{i})k_{ih_{n}}^{+}k_{ib_{n}%
}^{+}\right]  \\
&  =\tilde{\Psi}_{UV,p,q}^{+}(h_{n},b_{n}).
\end{align*}
Also, since $\sigma_{jg}^{2,\varepsilon_{U_{i}}\varepsilon_{V_{i}}}(z)$ is
bounded,
\begin{align*}
&  \E\left[  \left\vert \frac{h_{n}b_{n}}{m_{n}}\breve{\Psi}_{UV,p}^{+}%
(h_{n})-\E[\frac{h_{n}b_{n}}{m_{n}}\breve{\Psi}_{UV,p}^{+}(h_{n})]\right\vert
^{2}\right]  \\
&  \leq Cn^{-1}m_{n}^{-1}\int_{0}^{\infty}k^{2}\left(  \frac{m_{n}u}{h_{n}%
}\right)  k^{2}\left(  \frac{m_{n}u}{b_{n}}\right)  \left\vert r_{p}\left(
\frac{m_{n}u}{h_{n}}\right)  \right\vert ^{2}\left\vert r_{q}\left(
\frac{m_{n}u}{b_{n}}\right)  \right\vert ^{2}f(uh_{n})du\\
&  =O\left(  n^{-1}m_{n}^{-1}\right)  .
\end{align*}
Hence with $m_{n}\rightarrow0$ and $nm_{n}\rightarrow\infty,$%
\[
\frac{h_{n}b_{n}}{m_{n}}\breve{\Psi}_{UV,p,q}^{+}(h_{n},b_{n})=\tilde{\Psi
}_{UV,p}^{+}(h_{n},b_{n})+o_{P}(1).
\]
Then, (\ref{aa1}) follows from Lemma \ref{Omega++}. Simple algebra shows the
rest of the proof, using that
\[
\hat{c}_{1}-c_{1}=O_{P}\left(  \frac{1}{\sqrt{nh_{n}}}+h_{n}^{(p+1)}\right)
\]
and%
\[
\hat{c}_{2}-c_{2}=O_{P}\left(  \frac{1}{\sqrt{nh_{n}}}+h_{n}^{(p+1)}\right)  .
\]

\end{proof}
\end{lemma}

\subsection{MSE-Optimal Bandwidth Selectors}

We consider first the MSE optimal bandwidth selector for estimating the bias
term, i.e. $b_{n}.$ Consider the general case where we use a $q-$th order
polynomial for estimating the $p+1$ derivative, $q>p$. The leading term in the
bias is%
\begin{align*}
\bar{B}_{w,p,q}(b_{n}) &  =c_{1}^{\prime}\hat{B}_{Y,p,q}(b_{n})+c_{2}^{\prime
}\hat{B}_{X,p,q}(b_{n}),\\
\hat{B}_{V,p,q}(b_{n}) &  =\frac{\hat{\mu}_{V,q}^{+(p+1)}(b_{n})}%
{(p+1)!}\mathcal{B}_{p,p+1}^{+}-\frac{\hat{\mu}_{V,q}^{-(p+1)}(b_{n})}%
{(p+1)!}\mathcal{B}_{p,p+1}^{-}%
\end{align*}
and $\mathcal{B}_{p,q}^{+}=e_{0}^{\prime}\Gamma_{+,p}^{-1}\vartheta_{p,q}^{+}$
and $\mathcal{B}_{p,q}^{-}=e_{0}^{\prime}\Gamma_{-,p}^{-1}\vartheta_{p,q}%
^{-}.$ With a slightly abuse of notation denote the true bias%
\begin{equation}
\mathbf{B}_{\omega,p,p+1}=c_{1}^{\prime}B_{Y,p,p+1}+c_{2}^{\prime}%
B_{X,p,p+1},\label{biasexp}%
\end{equation}
where%
\[
B_{V,p,q}=\frac{\mu_{V}^{+(p+1)}}{(p+1)!}\mathcal{B}_{p,q}^{+}-\frac{\mu
_{V}^{-(p+1)}}{(p+1)!}\mathcal{B}_{p,q}^{-}.
\]

\begin{lemma}
\label{MSEb}Under Assumption \ref{AssumptionRRD},
\begin{align*}
\mathbf{MSE}(b_{n}) &  \equiv \E[\left(  \bar{B}_{w,p,q}(b_{n})-\mathbf{B}%
_{\omega,p,p+1}\right)  ^{2}|\mathcal{S}_{n}]\\
&  =b_{n}^{q-p}[B_{\omega,p,q}^{bc}+o_{P}(1)]\\
&  +n^{-1}b_{n}^{-1-2(p+1)}[V_{\omega,p,q}^{bc}+o_{P}(1)],
\end{align*}
where%
\[
B_{\omega,p,q}^{bc}=c_{1}^{\prime}B_{\omega,p,q}^{Y,bc}+c_{2}^{\prime
}B_{\omega,p,q}^{X,bc}%
\]%
\[
B_{\omega,p,q}^{V,bc}=\frac{\mu_{V}^{+(q+1)}}{(q+1)!}\frac{\mathcal{B}%
_{p+1,q,q+1}^{+}\mathcal{B}_{p,p+1}^{+}}{(p+1)!}-\frac{\mu_{V}^{-(q+1)}%
}{(q+1)!}\frac{\mathcal{B}_{p+1,q,q+1}^{-}\mathcal{B}_{p,p+1}^{-}}{(p+1)!}%
\]
and
\[
V_{\omega,p,q}^{bc}=V_{\omega,p,q}^{+bc}+V_{\omega,p,q}^{-bc}%
\]
where%
\begin{align*}
V_{\omega,p,q}^{+bc} &  =\left(  c_{1}^{\prime}\mathcal{V}_{Y,q}^{+(p+1)}%
c_{1}+2c_{1}^{\prime}\mathcal{C}_{YX,p,q}^{+}c_{2}+c_{2}^{\prime}%
\mathcal{V}_{X,q}^{+(p+1)}c_{2}\right)  \left(  \frac{\mathcal{B}_{p,p+1}^{+}%
}{(p+1)!}\right)  ^{2},\\
\mathcal{V}_{V,q}^{+(p+1)} &  =\left(  (p+1)!\right)  ^{2}e_{p+1}^{\prime
}\Gamma_{+,q}^{-1}\Psi_{V,q}^{+}\Gamma_{+,q}^{-1}e_{p+1},\\
\mathcal{C}_{YX,p,q}^{+} &  =\left(  (p+1)!\right)  ^{2}e_{p+1}^{\prime}%
\Gamma_{+,q}^{-1}\Psi_{YX,q}^{+}\Gamma_{+,q}^{-1}e_{p+1}%
\end{align*}
and%
\begin{align*}
V_{\omega,p,q}^{-bc} &  =\left(  c_{1}^{\prime}\mathcal{V}_{Y,q}^{-(p+1)}%
c_{1}+2c_{1}^{\prime}\mathcal{C}_{YX,p,q}^{-}c_{2}+c_{2}^{\prime}%
\mathcal{V}_{X,q}^{-(p+1)}c_{2}\right)  \left(  \frac{\mathcal{B}_{p,p+1}^{-}%
}{(p+1)!}\right)  ^{2},\\
\mathcal{V}_{V,q}^{-(p+1)} &  =\left(  (p+1)!\right)  ^{2}e_{p+1}^{\prime
}\Gamma_{-,q}^{-1}\Psi_{V,q}^{-}\Gamma_{-,q}^{-1}e_{p+1},\\
\mathcal{C}_{YX,p,q}^{-} &  =\left(  (p+1)!\right)  ^{2}e_{p+1}^{\prime}%
\Gamma_{-,q}^{-1}\Psi_{YX,q}^{-}\Gamma_{-,q}^{-1}e_{p+1}.
\end{align*}
The optimal MSE bandwidth has the form%
\begin{align*}
b_{n}^{\ast} &  =C_{MSEb,p}^{1/(2q+3)}n^{-1/(2q+3)}\\
C_{MSEb,p} &  =\frac{(2p+3)V_{\omega,p,q}^{bc}}{2(q-p)\left(  B_{\omega
,p,q}^{bc}\right)  ^{2}}.
\end{align*}

\begin{proof}
Note that by Lemma \ref{Bias}
\begin{align*}
\E[\big(  \bar{B}_{w,p,q}(b_{n})-&\mathbf{B}_{\omega,p,p+1}\big)
|\mathcal{S}_{n}]   =c_{1}^{\prime}\E[\hat{B}_{Y,p,q}(b_{n})-B_{Y,p,p+1}%
|\mathcal{S}_{n}]+c_{2}^{\prime}\E[\hat{B}_{X,p,q}(b_{n})-B_{X,p,p+1}%
|\mathcal{S}_{n}]\\
&  =c_{1}^{\prime}\E[e_{p+1}^{\prime}\hat{\beta}_{Y,q}^{+}(b_{n})-e_{p+1}%
^{\prime}\beta_{Y}^{+}|\mathcal{S}_{n}]\mathcal{B}_{p,p+1}^{+}-c_{1}^{\prime
}\E[e_{p+1}^{\prime}\hat{\beta}_{Y,q}^{-}(b_{n})-e_{p+1}^{\prime}\beta_{Y}%
^{-}|\mathcal{S}_{n}]\mathcal{B}_{p,p+1}^{-}\\
&  +c_{2}^{\prime}\E[e_{p+1}^{\prime}\hat{\beta}_{X,q}^{+}(b_{n})-e_{p+1}%
^{\prime}\beta_{X}^{+}|\mathcal{S}_{n}]\mathcal{B}_{p,p+1}^{+}-c_{2}^{\prime
}\E[e_{p+1}^{\prime}\hat{\beta}_{X,q}^{-}(b_{n})-e_{p+1}^{\prime}\beta_{X}%
^{-}|\mathcal{S}_{n}]\mathcal{B}_{p,p+1}^{-}\\
&  =c_{1}^{\prime}b_{n}^{p-q}\left[  \frac{\mu_{Y}^{+(q+1)}}{(q+1)!}%
\frac{\mathcal{B}_{p+1,q,q+1}^{+}\mathcal{B}_{p,p+1}^{+}}{(p+1)!}-\frac
{\mu_{Y}^{-(q+1)}}{(q+1)!}\frac{\mathcal{B}_{p+1,q,q+1}^{-}\mathcal{B}%
_{p,p+1}^{-}}{(p+1)!}+o_{P}(1)\right]  \\
&  +c_{2}^{\prime}b_{n}^{p-q}\left[  \frac{\mu_{X}^{+(q+1)}}{(q+1)!}%
\frac{\mathcal{B}_{p+1,q,q+1}^{+}\mathcal{B}_{p,p+1}^{+}(h_{n})}{(p+1)!}%
-\frac{\mu_{X}^{-(q+1)}}{(q+1)!}\frac{\mathcal{B}_{p+1,q,q+1}^{-}%
\mathcal{B}_{p,p+1}^{-}(h_{n})}{(p+1)!}+o_{P}(1)\right]  \\
&  =b_{n}^{q-p}\left[  B_{\omega,p,q}^{bc}+o_{P}(1)\right]  .
\end{align*}
As for the variance, write
\begin{align*}
\bar{B}_{w,p,q}(b_{n}) &  =\bar{B}_{w,p,q}^{+}(b_{n})-\bar{B}_{w,p,q}%
^{-}(b_{n}),\\
\bar{B}_{w,p,q}^{+}(b_{n}) &  =c_{1}^{\prime}\hat{B}_{Y,p,q}^{+}(b_{n}%
)+c_{2}^{\prime}\hat{B}_{X,p,q}^{+}(b_{n}),\\
\bar{B}_{w,p,q}^{-}(b_{n}) &  =c_{1}^{\prime}\hat{B}_{Y,p,q}^{-}(b_{n}%
)+c_{2}^{\prime}\hat{B}_{X,p,q}^{-}(b_{n}),\\
\bar{B}_{V,p,q}^{+}(b_{n}) &  =\frac{\hat{\mu}_{V,q}^{+(p+1)}(b_{n})}%
{(p+1)!}\mathcal{B}_{p,p+1}^{+},\\
\hat{B}_{V,p,q}^{-}(b_{n}) &  =\frac{\hat{\mu}_{V,q}^{-(p+1)}(b_{n})}%
{(p+1)!}\mathcal{B}_{p,p+1}^{-}.
\end{align*}
This decomposition and Lemma \ref{Variance} imply
\[
Var[\bar{B}_{w,p,q}(b_{n})|\mathcal{S}_{n}]=n^{-1}b_{n}^{-1-2(p+1)}\left[
V_{\omega,p,q}^{+bc}+V_{\omega,p,q}^{-bc}+o_{P}(1)\right]  ,
\]
where%
\begin{align*}
V_{\omega,p,q}^{+bc} &  =\left(  c_{1}^{\prime}\mathcal{V}_{Y,q}^{+(p+1)}%
c_{1}+2c_{1}^{\prime}\mathcal{C}_{YX,p,q}^{+}c_{2}+c_{2}^{\prime}%
\mathcal{V}_{X,q}^{+(p+1)}c_{2}\right)  \left(  \frac{\mathcal{B}_{p,p+1}^{+}%
}{(p+1)!}\right)  ^{2},\\
\mathcal{V}_{V,q}^{+(p+1)} &  =\left(  (p+1)!\right)  ^{2}e_{p+1}^{\prime
}\Gamma_{+,q}^{-1}\Psi_{V,q}^{+}\Gamma_{+,q}^{-1}e_{p+1},\\
\mathcal{C}_{YX,p,q}^{+} &  =\left(  (p+1)!\right)  ^{2}e_{p+1}^{\prime}%
\Gamma_{+,q}^{-1}\Psi_{YX,q}^{+}\Gamma_{+,q}^{-1}e_{p+1}%
\end{align*}
and similarly%
\begin{align*}
V_{\omega,p,q}^{-bc} &  =\left(  c_{1}^{\prime}\mathcal{V}_{Y,q}^{-(p+1)}%
c_{1}+2c_{1}^{\prime}\mathcal{C}_{YX,p,q}^{-}c_{2}+c_{2}^{\prime}%
\mathcal{V}_{X,q}^{-(p+1)}c_{2}\right)  \left(  \frac{\mathcal{B}_{p,p+1}^{-}%
}{(p+1)!}\right)  ^{2},\\
\mathcal{V}_{V,q}^{-(p+1)} &  =\left(  (p+1)!\right)  ^{2}e_{p+1}^{\prime
}\Gamma_{-,q}^{-1}\Psi_{V,q}^{-}\Gamma_{-,q}^{-1}e_{p+1},\\
\mathcal{C}_{YX,p,q}^{-} &  =\left(  (p+1)!\right)  ^{2}e_{p+1}^{\prime}%
\Gamma_{-,q}^{-1}\Psi_{YX,q}^{-}\Gamma_{-,q}^{-1}e_{p+1}.
\end{align*}

\end{proof}
\end{lemma}

The following MSE calculations for the linearized term in the WLATE estimators
have been obtained before, and are summarized in the following result. Recall
\[
\hat{l}_{\omega}=c_{1}^{\prime}(\hat{\delta}_{Y}-\delta_{Y})+c_{2}^{\prime
}\left(  \hat{\delta}_{X}-\delta_{X}\right)  ,
\]
where $c_{1}=c/\tau_{X}$, $c_{2}^{\prime}=\left[  \delta_{Y}^{\prime}%
G_{\delta}-\beta_{w}(c^{\prime}+\delta_{X}^{\prime}G_{\delta})\right]
/\tau_{X}$ and $G_{\delta}$ is the derivative of $g$ with respect to
$\delta_{X}$.

\begin{lemma}
\label{MSEh}Under Assumption \ref{AssumptionRRD}(1)-(7),
\begin{align*}
\mathbf{MSE}(h_{n})  &  \equiv \E[\hat{l}_{\omega}^{2}|\mathcal{S}_{n}]\\
&  =h_{n}^{2(p+1)}[B_{\omega,p,p+1}+o_{P}(1)]\\
&  +n^{-1}h_{n}^{-1}[V_{\omega,p}+o_{P}(1)],
\end{align*}
where $B_{\omega,p,p+1}$ is defined in (\ref{biasexp}) and
\begin{align*}
V_{\omega,p}  &  =V_{+1}+V_{-1},\\
V_{\pm1}  &  =c_{1}^{\prime}\mathcal{V}_{Y,p}^{\pm}c_{1}+2c_{1}^{\prime
}\mathcal{C}_{YX,p}^{\pm}c_{2}+c_{2}^{\prime}\mathcal{V}_{X,p}^{\pm}c_{2},\\
\mathcal{C}_{YX,p}^{\pm}  &  =e_{0}^{\prime}\Gamma_{\pm,p}^{-1}\Psi
_{YX,p}^{\pm}\Gamma_{\pm,p}^{-1}e_{0},\\
\mathcal{V}_{V,p}^{+}  &  =\mathcal{C}_{VV,p}^{\pm}.
\end{align*}
The optimal MSE bandwidth has the form%
\begin{align*}
h_{n}^{\ast}  &  =C_{MSEh,p}^{1/(2p+3)}n^{-1/(2p+3)}\\
C_{MSEh,p}  &  =\frac{V_{\omega,p}}{2(p+1)B_{\omega,p}^{2}}.
\end{align*}

\begin{proof}
It follows from Lemmas \ref{Biasc2} and \ref{Variancec2}.
\end{proof}
\end{lemma}

\subsection{Implementation of MSE-Optimal Bandwidth Selectors}

In this section we describe the implementation of MSE-optimal bandwidths for local linear estimation $(p=1)$. We consider the following steps.

\begin{description}
\item[Step 1:] Choose a preliminary pilot bandwidth $c_{n}$ as%
\begin{align*}
c_{n} &  =C_{K}\times C_{sd}\times n^{-1/5},\\
C_{K} &  =\left(  \frac{8\sqrt{\pi}\int k^{2}(u)du}{3\left(  \int
u^{2}k(u)du\right)  ^{2}}\right)  ^{1/5}\\
C_{sd} &  =\min\left(  s,\frac{IQR}{1.349}\right)  ,
\end{align*}
where $s^{2}$ and $IQR$ are the sample variance and interquartile range of
$\{Z_{i}\}_{i=1}^{n}.$ The constant $C_{K}$ is 1.843 when $K$ is the uniform
kernel, and 2.576 when $K$ is the triangular kernel.

\item[Step 2:] Choose an optimal bandwidth for bias estimation $b_{n}$
following Lemma \ref{MSEb} with $q=p+1=2$:
\begin{align*}
\hat{b}_{n}^{\ast} &  =\hat{C}_{MSEb,1}^{1/7}n^{-1/7}\\
\hat{C}_{MSEb,1} &  =\frac{5\hat{V}_{\omega,1,2}^{bc}(c_{n})}{2\left(  \hat
{B}_{\omega,1,2}^{bc}(d_{n})\right)  ^{2}}.
\end{align*}
The variance estimate $\hat{V}_{\omega,1,2}^{bc}(c_{n})$ is
\[
\hat{V}_{\omega,1,2}^{bc}(c_{n})=\hat{V}_{\omega,1,2}^{+bc}(c_{n})+\hat
{V}_{\omega,1,2}^{-bc}(c_{n})
\]
where (with $p=1$ and $q=2)$%
\begin{align*}
\hat{V}_{\omega,p,q}^{\pm bc}(c_{n}) &  =\left(  \hat{c}_{1}^{\prime
}\mathcal{V}_{Y,q}^{+(p+1)}(c_{n})\hat{c}_{1}+2\hat{c}_{1}^{\prime}%
\mathcal{C}_{YX,p,q}^{+}(c_{n})\hat{c}_{2}+\hat{c}_{2}^{\prime}\mathcal{V}%
_{X,q}^{+(p+1)}(c_{n})\hat{c}_{2}\right)  \left(  \frac{\mathcal{B}%
_{p,p+1}^{\pm}(c_{n})}{(p+1)!}\right)  ^{2}\!,\\
\mathcal{V}_{V,q}^{\pm(p+1)}(c_{n}) &  =\left(  (p+1)!\right)  ^{2}%
e_{p+1}^{\prime}\Gamma_{\pm,q}^{-1}(c_{n})\hat{\Psi}_{V,q}^{\pm}(c_{n}%
)\Gamma_{\pm,q}^{-1}(c_{n})e_{p+1},\\
\mathcal{C}_{YX,p,q}^{\pm}(c_{n}) &  =\left(  (p+1)!\right)  ^{2}%
e_{p+1}^{\prime}\Gamma_{\pm,q}^{-1}(c_{n})\hat{\Psi}_{YX,q}^{\pm}(c_{n}%
)\Gamma_{\pm,q}^{-1}(c_{n})e_{p+1}.
\end{align*}
The bias term $\hat{B}_{\omega,1,2}^{bc}(d_{n})$ is obtained as (with $p=1$
and $q=2)$
\[
\hat{B}_{\omega,p,q}^{bc}(d_{n})=\hat{c}_{1}^{\prime}\hat{B}_{\omega
,p,q}^{Y,bc}(d_{n})+\hat{c}_{2}^{\prime}\hat{B}_{\omega,p,q}^{X,bc}(d_{n})
\]%
\[
\hat{B}_{\omega,p,q}^{V,bc}(d_{n})=\frac{\hat{\mu}_{V}^{+(q+1)}(d_{n}%
)}{(q+1)!}\frac{\mathcal{B}_{p+1,q,q+1}^{+}(c_{n})\mathcal{B}_{p,p+1}%
^{+}(c_{n})}{(p+1)!}-\frac{\hat{\mu}_{V}^{-(q+1)}(d_{n})}{(q+1)!}%
\frac{\mathcal{B}_{p+1,q,q+1}^{-}(c_{n})\mathcal{B}_{p,p+1}^{-}(c_{n}%
)}{(p+1)!},
\]
where $d_{n}$ is the optimal bandwidth for estimating the third derivatives
$\mu_{V}^{\pm(q+1)}$ with a third order local polynomial. This optimal
bandwidth is obtained for each variable $V$ and $\pm$ sides as (with $q=2)$%
\begin{align*}
d_{V,n}^{\pm} &  =\hat{C}_{V,\pm,q+1}^{1/(2q+5)}n^{-1/(2q+5)}\\
\hat{C}_{V,\pm,q+1}^{1/(2q+5)} &  =\frac{(2q+3)\hat{V}_{V,q+1}^{\pm
(q+1)}(c_{n})}{2\left(  \hat{B}_{q+1,q+1,q+2}^{\pm}(c_{n})\right)  ^{2}},
\end{align*}
where
\[
\hat{V}_{V,q+1}^{\pm(q+1)}(c_{n})=e_{q+1}^{\prime}\Gamma_{+,q+1}^{-1}%
(c_{n})\hat{\Psi}_{V,q}^{\pm}(c_{n})\Gamma_{+,q+1}^{-1}(c_{n})e_{q+1}%
\]
and
\[
\hat{B}_{q+1,q+1,q+2}^{\pm}(c_{n})=\frac{\hat{\mu}_{V}^{+(q+2)}(c_{n}%
)}{(q+2)!}\mathcal{B}_{q+1,q+1,q+2}^{+}(c_{n}).
\]

\item[Step 3:] Choose optimal bandwidth $h_{n}$ for the estimator: use $c_{n}$
and $b_{n}$ to compute for $p=1$
\begin{align*}
h_{n}^{\ast} &  =\hat{C}_{MSEh,p}^{1/(2p+3)}n^{-1/(2p+3)}\\
\hat{C}_{MSEh,p} &  =\frac{\hat{V}_{\omega,p}(c_{n})}{2(p+1)\hat{B}_{\omega
,p}^{2}(b_{n})},
\end{align*}
where
\begin{align*}
\hat{V}_{\omega,p}(c_{n}) &  =\hat{V}_{+1}(c_{n})+\hat{V}_{-1}(c_{n}),\\
\hat{V}_{\pm1}(c_{n}) &  =\hat{c}_{1}^{\prime}\mathcal{V}_{Y,p}^{\pm}%
(c_{n})\hat{c}_{1}+2\hat{c}_{1}^{\prime}\mathcal{C}_{YX,p}^{\pm}(c_{n})\hat
{c}_{2}+\hat{c}_{2}^{\prime}\mathcal{V}_{X,p}^{\pm}(c_{n})\hat{c}_{2},\\
\mathcal{C}_{YX,p}^{\pm}(c_{n}) &  =e_{0}^{\prime}\Gamma_{\pm,p}^{-1}%
(c_{n})\hat{\Psi}_{YX,p}^{\pm}(c_{n})\Gamma_{\pm,p}^{-1}(c_{n})e_{0},\\
\mathcal{V}_{V,p}^{+}(c_{n}) &  =\mathcal{C}_{VV,p}^{\pm}(c_{n}).
\end{align*}
and
\[
\hat{B}_{\omega,p,q}(c_{n},b_{n})=\hat{c}_{1}^{\prime}\hat{B}_{Y,p,q}%
(c_{n},b_{n})+\hat{c}_{2}^{\prime}\hat{B}_{X,p,q}(c_{n},b_{n}),
\]
where%
\[
\hat{B}_{V,p,q}(c_{n},b_{n})=\frac{\hat{\mu}_{V,q}^{+(p+1)}(b_{n})}%
{(p+1)!}\mathcal{B}_{p,p+1}^{+}(c_{n})-\frac{\hat{\mu}_{V,q}^{-(p+1)}(b_{n}%
)}{(p+1)!}\mathcal{B}_{p,p+1}^{-}(c_{n}).
\]

\end{description}

\newpage 

\section{Additional Monte Carlo Simulations}\label{app: MC rob}

In this appendix we consider a few other simulation scenarios.

\subsection{More Intense Treatment Effect Heterogeneity}\label{sec: MC more het}

In Table \ref{tab: RDDbinworse}, we consider the case where $\beta_{XW}=5$ or $\beta_{XW}=10$ (instead of $\beta_{XW}=0$ or $\beta_{XW}=2$, as in Table \ref{tab: RDD}). All other parameters from Section \ref{sec: MC} are kept the same. The increased heterogeneity in treatment effects lead to \emph{all estimators} perform worse, with similar relative results between the CWLATE and the standard RDD estimators.

\begin{table}[H]
\caption{MSE Relative to Own Target Estimand -- Intense Treatment Effect Heterogeneity}\label{tab: RDDbinworse}
\begin{center}
\scalebox{.9}{\small \begin{tabularx}{1.18\linewidth}{c|rrr|rrr|rrr|rrr}

\toprule
 Obs.  & \multicolumn{6}{c}{\(\alpha_{DW}=0\)} & \multicolumn{6}{c}{\(\alpha_{DW}=1\)} \\  & \multicolumn{3}{c}{\(\beta_{XW}=5\)} & \multicolumn{3}{c}{\(\beta_{XW}=10\)} & \multicolumn{3}{c}{\(\beta_{XW}=5\)} & \multicolumn{3}{c}{\(\beta_{XW}=10\)} \\ & \(\hat{\beta}_{U}\) & \(\hat{\beta}^{cov}_{U}\) & \(\hat{\beta}_{CW}\) & \(\hat{\beta}_{U}\) & \(\hat{\beta}^{cov}_{U}\) & \(\hat{\beta}_{CW}\) & \(\hat{\beta}_{U}\) & \(\hat{\beta}^{cov}_{U}\) & \(\hat{\beta}_{CW}\) & \(\hat{\beta}_{U}\) & \(\hat{\beta}^{cov}_{U}\) & \(\hat{\beta}_{CW}\) \tabularnewline
\midrule \addlinespace[\belowrulesep]
300&411.22&405.29&11.66&396825.20&30411.45&35.87&357.78&861.52&3.26&9154.98&6938.44&9.42 \tabularnewline \addlinespace[.05in]
500&188.47&23.10&7.20&108.86&35.65&25.97&29.01&23.72&1.34&648.90&425.38&3.40 \tabularnewline \addlinespace[.05in]
1000&3.10&1.90&3.94&9.53&6.12&14.50&3.22&2.72&.48&10.18&8.35&1.03 \tabularnewline \addlinespace[.05in]
2000&1.18&.85&2.05&4.22&2.69&7.53&1.15&1.03&.19&4.36&3.62&.36 \tabularnewline \addlinespace[.05in]
3000&.78&.54&1.35&2.67&1.72&5.06&.75&.67&.12&2.70&2.22&.19 \tabularnewline \addlinespace[.05in]
4000&.57&.39&1.00&2.05&1.29&3.75&.53&.47&.08&1.96&1.60&.14 \tabularnewline \addlinespace[.05in]
5000&.46&.31&.81&1.67&1.03&3.09&.42&.37&.07&1.53&1.27&.10 \tabularnewline \addlinespace[.05in]
\bottomrule 

\end{tabularx}}
\end{center}\label{tab:increaseheterogeneity}
\end{table}

\subsection{Coarsened $W_i$}\label{sec: MC coarse}

Next, we consider what happens when more complex covariates are used to implement the CWLATE estimator. We use the same Monte Carlo set up discussed in Section \ref{sec: MC}, but with $W_i\in \{-5/3,-4/3,$ $-1,-2/3,-1/3,1/3,2/3,1,4/3,$ $5/3\}$, in contrast to the binary covariate $W_i \in\{-1,1\}$ from Section \ref{sec: MC}.

We consider two separate scenarios: (a) CWLATE is estimated with $\tilde{W}_i=W_i$, and (b) CWLATE is estimated with a coarsened binary version of the covariate ($\tilde{W}_i=1$ if $W_i>0,$ and $\tilde{W}_i=-1$ if $W_i<0$). As in the baseline case, the term representing the heterogeneity in the first-stages across different values of $W_i$ is $\alpha_{DW}(W_i)=W_i\cdot\alpha_{DW}$. Analogously, the term representing the heterogeneity in the treatment effects across different values of $W_i$ is $\beta_{XW}(W_i)=W_i\cdot\beta_{XW}$.

Table~\ref{tab: RDDbintilde} reports the MSE when CWLATE is estimated with the full covariate taking 10 different values. The results are similar to the results in Table \ref{tab: RDD}, with the small difference due to increased heterogeneity when either $\alpha_{DW}\neq0$ or $\beta_{XW}\neq0$, since now the range of $W_i$ increased from $[-1,1]$ to $[-5/3,5/3]$.

\begin{table}[H]
\caption{MSE using Full Covariate -- $\tilde{W}_i= W_i$, $dim(\tilde{W}_i)=dim(W_i)=10$}\label{tab: RDDbintilde}
\begin{center}
\scalebox{.9}{\small \begin{tabularx}{1.11\linewidth}{c|rrr|rrr|rrr|rrr}

 Obs.  & \multicolumn{6}{c}{\(\alpha_{DW}=0\)} & \multicolumn{6}{c}{\(\alpha_{DW}=1\)} \\  & \multicolumn{3}{c}{\(\beta_{XW}=0\)} & \multicolumn{3}{c}{\(\beta_{XW}=2\)} & \multicolumn{3}{c}{\(\beta_{XW}=0\)} & \multicolumn{3}{c}{\(\beta_{XW}=2\)} \\ & \(\hat{\beta}_{U}\) & \(\hat{\beta}^{cov}_{U}\) & \(\hat{\beta}_{CW}\) & \(\hat{\beta}_{U}\) & \(\hat{\beta}^{cov}_{U}\) & \(\hat{\beta}_{CW}\) & \(\hat{\beta}_{U}\) & \(\hat{\beta}^{cov}_{U}\) & \(\hat{\beta}_{CW}\) & \(\hat{\beta}_{U}\) & \(\hat{\beta}^{cov}_{U}\) & \(\hat{\beta}_{CW}\) \tabularnewline
\toprule
\midrule \addlinespace[\belowrulesep]
300&517.82&118.64&1.06&72.50&379.06&1.73&2718.84&27680.72&.93&72.85&5259.80&1.53 \tabularnewline \addlinespace[.05in]
500&6.48&93.96&.71&10.68&9.11&1.27&70.61&63.13&.55&16.31&50.54&.92 \tabularnewline \addlinespace[.05in]
1000&.48&.51&.38&.92&.81&.83&.69&.69&.27&.88&.85&.45 \tabularnewline \addlinespace[.05in]
2000&.22&.22&.20&.41&.35&.48&.25&.26&.13&.39&.36&.21 \tabularnewline \addlinespace[.05in]
3000&.14&.14&.13&.26&.22&.33&.16&.17&.09&.25&.24&.14 \tabularnewline \addlinespace[.05in]
4000&.10&.10&.10&.19&.16&.26&.12&.12&.07&.18&.17&.11\tabularnewline \addlinespace[.05in]
5000&.08&.08&.08&.15&.13&.21&.09&.10&.05&.14&.13&.08 \tabularnewline \addlinespace[.05in]
\bottomrule 

\end{tabularx}}
\end{center}
\end{table}

Table~\ref{tab: RDDbin} reports the MSE when CWLATE is estimated with the coarsened (binary) covariate. The only columns that are different from Table~\ref{tab: RDDbintilde} are the CWLATE ones, since $\hat{\beta}^{cov}_U$ is estimated with the full vector $W_i$ in both tables. The results in both simulations are similar, with a small advantage when using the full covariate vector $W_i,$ especially in smaller samples. 

\begin{table}[H]
\caption{MSE using Coarsened Covariate -- $\tilde{W}_i\neq W_i$, $dim(\tilde{W}_i)=2<10=dim(W_i)$}\label{tab: RDDbin}
\begin{center}
\scalebox{.9}{\small \begin{tabularx}{1.11\linewidth}{c|rrr|rrr|rrr|rrr}

 Obs.  & \multicolumn{6}{c}{\(\alpha_{DW}=0\)} & \multicolumn{6}{c}{\(\alpha_{DW}=1\)} \\  & \multicolumn{3}{c}{\(\beta_{XW}=0\)} & \multicolumn{3}{c}{\(\beta_{XW}=2\)} & \multicolumn{3}{c}{\(\beta_{XW}=0\)} & \multicolumn{3}{c}{\(\beta_{XW}=2\)} \\ & \(\hat{\beta}_{U}\) & \(\hat{\beta}^{cov}_{U}\) & \(\hat{\beta}_{CW}\) & \(\hat{\beta}_{U}\) & \(\hat{\beta}^{cov}_{U}\) & \(\hat{\beta}_{CW}\) & \(\hat{\beta}_{U}\) & \(\hat{\beta}^{cov}_{U}\) & \(\hat{\beta}_{CW}\) & \(\hat{\beta}_{U}\) & \(\hat{\beta}^{cov}_{U}\) & \(\hat{\beta}_{CW}\) \tabularnewline
\toprule
\midrule \addlinespace[\belowrulesep]
300&517.82&118.64&4.35&72.50&379.06&5.68&2718.84&27680.72&2.68&72.85&5259.80&2.15 \tabularnewline \addlinespace[.05in]
500&6.48&93.96&1.45&10.68&9.11&2.50&70.61&63.13&.75&16.31&50.54&.93 \tabularnewline \addlinespace[.05in]
1000&.48&.51&.47&.92&.81&1.11&.69&.69&.32&.88&.85&.37 \tabularnewline \addlinespace[.05in]
2000&.22&.22&.22&.41&.35&.53&.25&.26&.15&.39&.36&.17 \tabularnewline \addlinespace[.05in]
3000&.14&.14&.14&.26&.22&.35&.16&.17&.10&.25&.24&.12 \tabularnewline \addlinespace[.05in]
4000&.10&.10&.10&.19&.16&.26&.12&.12&.07&.18&.17&.09\tabularnewline \addlinespace[.05in]
5000&.08&.08&.08&.15&.13&.21&.09&.10&.06&.14&.13&.06 \tabularnewline \addlinespace[.05in]
\bottomrule 

\end{tabularx}}
\end{center}
\end{table}

In our simulations, using a larger covariate set does not seem to be detrimental, and may in fact be beneficial (see also results in our application). At the same time, using a coarsened covariate in the CWLATE estimator had minimal impact on precision and seems to be a viable alternative.

\subsection{Estimating $\beta_U$ with the CWLATE Estimator}\label{sec:betterestimator}

In this section, we compare the CWLATE estimator to the standard RDD estimators \emph{when all target the unconditional LATE} $\beta_U$. Note that, in the cases we examine, the CWLATE estimator is an inconsistent estimator of $\beta_U$ (since it is an unbiased estimator of $\beta_{CW}\neq \beta_U$). The results can be seen in Table~\ref{tab: RDD2}. We consider the case where $\alpha_{DW}=1$ and either $\beta_{XW}=2$ or $\beta_{XW}=10$. The differences between the estimands increase with $\beta_{XW},$ so the CWLATE is a more severely biased estimator of $\beta_U$ when $\beta_{XW}=10$. Nevertheless, the results show that not only is the MSE of the CWLATE estimator generally lower than that of the standard RDD estimators, but the comparatively better performance intensifies with the higher $\beta_{XW}$ value for moderate sample sizes.
\begin{table}[H]
\caption{MSE and 95\% Coverage Relative to Unconditional LATE Target}\label{tab: RDD2}
\begin{center}
\scalebox{.9}{\small \begin{tabularx}{1.1\linewidth}{c|rrr|rrr|rrr|rrr}

\toprule
 Obs.  & \multicolumn{6}{c}{MSE} & \multicolumn{6}{c}{95\% Coverage (in \%)} \\  & \multicolumn{3}{c}{\(\beta_{XW}=2\)} & \multicolumn{3}{c}{\(\beta_{XW}=10\)} & \multicolumn{3}{c}{\(\beta_{XW}=2\)} & \multicolumn{3}{c}{\(\beta_{XW}=10\)} \\ & \(\hat{\beta}_{U}\) & \(\hat{\beta}^{cov}_{U}\) & \(\hat{\beta}_{CW}\) & \(\hat{\beta}_{U}\) & \(\hat{\beta}^{cov}_{U}\) & \(\hat{\beta}_{CW}\) & \(\hat{\beta}_{U}\) & \(\hat{\beta}^{cov}_{U}\) & \(\hat{\beta}_{CW}\) & \(\hat{\beta}_{U}\) & \(\hat{\beta}^{cov}_{U}\) & \(\hat{\beta}_{CW}\) \tabularnewline
\midrule \addlinespace[\belowrulesep]
300&3363.10&6153.34&1.61&9155.00&6938.44&7.21&98.51&97.67&92.00&97.09&94.75&94.79 \tabularnewline \addlinespace[.05in]
500&267.24&18.65&.73&648.90&425.38&2.85&98.08&97.32&92.54&97.33&95.60&97.94 \tabularnewline \addlinespace[.05in]
1000&.89&.85&.35&10.18&8.35&1.43&96.24&95.99&90.21&96.22&95.48&99.46 \tabularnewline \addlinespace[.05in]
2000&.36&.35&.18&4.36&3.62&1.33&95.43&95.13&88.27&95.23&95.11&99.39 \tabularnewline \addlinespace[.05in]
3000&.23&.23&.14&2.70&2.22&1.33&95.16&94.98&85.16&95.27&95.19&96.64 \tabularnewline \addlinespace[.05in]
4000&.17&.17&.14&1.96&1.60&1.43&95.13&94.96&77.45&95.16&95.11&85.32 \tabularnewline \addlinespace[.05in]
5000&.13&.13&.11&1.53&1.27&1.57&95.03&94.91&75.88&95.19&95.23&63.49 \tabularnewline \addlinespace[.05in]
\bottomrule 

\end{tabularx}}
\end{center}
\vspace{-.3in}
{\footnotesize \singlespacing Notes: This Table shows the MSE and the 95\% confidence interval coverage of each estimator relative to the target $\beta_U$ when $\alpha_{DW}=1$. Note that $\hat{\beta}_{CW}$ is a biased estimator of $\beta_U$, since it is an unbiased estimator of $\beta_{CW}$ and $\beta_{CW}\neq\beta_U$ in these cases.}
\end{table}

\begin{figure}[H]
\begin{center}
\caption{MSE Relative to Unconditional LATE Target -- Decomposition}\label{fig: MSE comp}
\begin{subfigure}[b]{0.45\textwidth}\caption{Bias$^2$}
\includegraphics[scale=0.55]{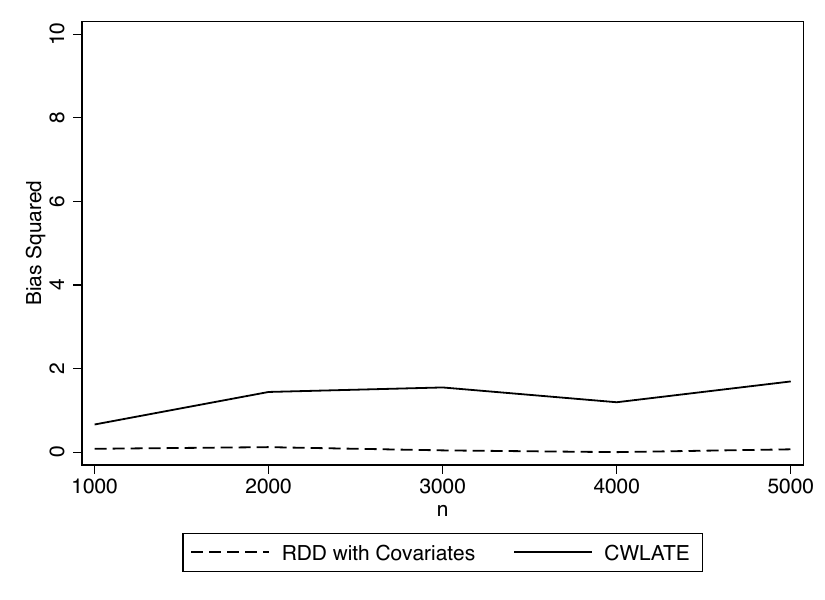}\end{subfigure}\hspace{.1in}
\begin{subfigure}[b]{0.45\textwidth}\caption{Variance}
\includegraphics[scale=0.55]{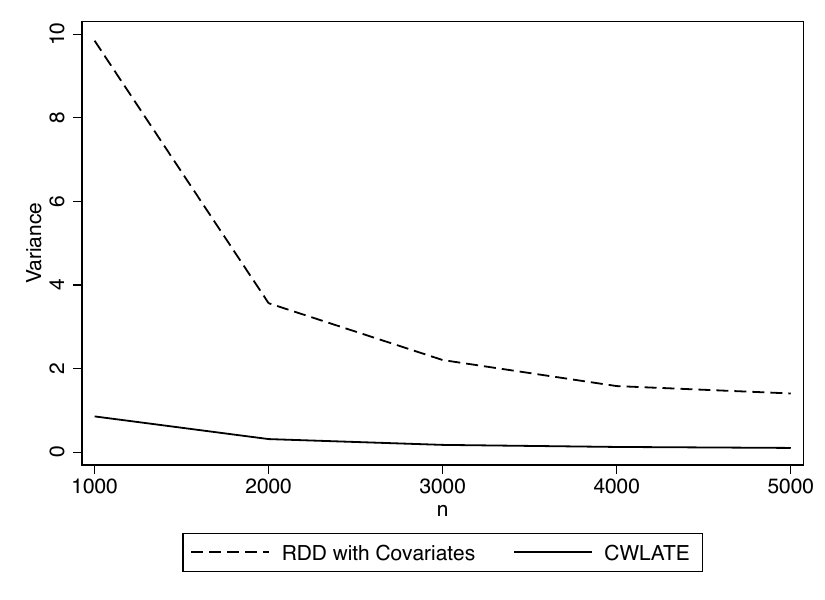}\end{subfigure}
\end{center}
\vspace{-.3in}
{\footnotesize \singlespacing Notes: This figure shows the two components of the MSE of $\hat{\beta}^{cov}_U$ and $\hat{\beta}_{CW}$ when evaluated with respect to $\beta_U$ as the target when $\alpha_{DW}=1$ and $\beta_{XW}=10$.}
\end{figure}

The relative advantage of the CWLATE even when biased is due to the excessively large variance of the standard RDD estimators. The standard RDD estimators show signs of weak identification for the small samples, while the CWLATE's variance remains relatively low. Figure~\ref{fig: MSE comp} shows that, even for moderate samples, the variance of the standard RDD estimators is excessively high. A sample of 5,000 was necessary for the variance to be sufficiently low that the standard RDD with covariates could surpass the biased CWLATE for $\beta_{XW}=10.$ Note that for $\beta_{XW}=2,$ this level is not reached for any of the sample sizes considered.

These results do not mean that the CWLATE is a better estimator of $\beta_U$. Although technically it seems to be superior for point estimation (since it has a smaller MSE), it is not advisable to perform inference with a biased estimator because confidence interval coverages and test sizes can be wrong. In the right two panels of the table, we show that the coverage rate of the 95\% confidence interval using the CWLATE estimator can be seriously off the mark.

Rather, the point of this exercise is to show that the issues of excessive variance should not be overlooked. The fact that an inconsistent estimator can perform better in such a wide range of scenarios is a cautionary tale against the standard RDD as an estimand. In many situations, such as in our application, the CWLATE can be a better target to substantiate causal claims.

\end{document}